\documentclass[11pt]{article}
\usepackage[a4paper,margin=1.1in]{geometry}
\usepackage{amsmath,amssymb,amsthm}
\usepackage[round,authoryear]{natbib}
\usepackage{booktabs}
\usepackage{graphicx}
\usepackage{xcolor}
\usepackage{setspace}
\usepackage{caption}
\usepackage{enumitem}
\usepackage{microtype}
\usepackage[colorlinks=true,linkcolor=blue,citecolor=blue,urlcolor=blue]{hyperref}
\expandafter\def\expandafter\UrlBreaks\expandafter{\UrlBreaks\do\-}
\newtheorem{definition}{Definition}[section]
\newtheorem{proposition}{Proposition}[section]
\newtheorem{corollary}{Corollary}[section]
\newtheorem{lemma}{Lemma}[section]

\newcommand{\E}{\mathbb{E}}
\newcommand{\dd}{\,\mathrm{d}}
\newcommand{\bk}{\mathbf{k}}
\newcommand{\bp}{\mathbf{p}}
\newcommand{\Abar}{\bar{A}}

\newcommand{\rivalshare}{four fifths}                      %
\newcommand{\irtcorr}{0.99}                                %
\newcommand{\sharelead}{0.46}                              %
\newcommand{\sharelag}{0.14}                               %
\title{The Economic Depreciation of Artificial Intelligence Capital}
\author{\setlength{\tabcolsep}{9pt}\begin{tabular}{ccc}
Yukun Zhang & Tianyang Zhang & Qiqi Xu\\[2pt]
{\footnotesize The Chinese University of Hong Kong} & {\footnotesize University of Macau} & {\footnotesize Renmin University of China}\\
{\footnotesize Hong Kong, China} & {\footnotesize Macau, China} & {\footnotesize Beijing, China}\\
{\footnotesize\texttt{215010026@link.cuhk.edu.cn}} & {\footnotesize\texttt{yc57317@um.edu.mo}} & {\footnotesize\texttt{xuqiqi\_cam@163.com}}
\end{tabular}}
\date{\today}

\begin{document}
\maketitle

\begin{abstract}
\noindent A trained model does not wear out, yet its revenue share collapses within months. With weekly data on every model version sold on a major inference marketplace, we measure the depreciation of AI model capital. A frontier version's revenue share halves within twelve weeks, 96 percent a year, newer rival versions take four fifths of the lost share, and the estimated demand system predicts its rate. Developers retire versions after sixteen months, when little value remains, so schedules based on service lives or hardware lives understate depreciation several-fold. Net investment in versions is a seventh to a quarter of gross investment.
\end{abstract}


\clearpage
\section{Introduction}\label{sec:intro}

Investment in artificial intelligence has reached a scale with few precedents. The four largest U.S.\ cloud providers planned to spend about \$325 billion on capital expenditure in 2025, and forecasts for 2026 exceed \$500 billion \citep{Bloomberg2025Capex,GoldmanCapex2026}. Most of that is hardware: data centers and the accelerators in them. The rest builds a different asset, the trained models, and a frontier training run now costs hundreds of millions of dollars \citep{CottierEtAl2024}. Whether this investment pays depends on how fast both assets lose value. The service life of accelerators has become a public dispute between investors and the firms that depreciate them over five or six years \citep{CITP2025Chips}. For model capital, there is no estimate. National accountants capitalize software over service lives of three to five years, implying geometric depreciation rates of 33 to 55 percent a year \citep{BEA2024}, and the rates the BEA applies to business R\&D capital, derived from the estimates of \citet{LiHall2020}, range from 10 to 40 percent a year across industries \citep{BEA2024}. Neither convention was designed for a trained model. A model's weights do not wear out: a system trained in 2024 answers a query today exactly as it did on the day it was released. If its value declines, the decline is entirely obsolescence, and the obsolescence is driven by what competitors release. This paper measures how fast that happens, why, and what it implies for capital measurement and for the welfare economics of the innovation race.

We measure the economic depreciation of individual model versions in the sense of \citet{HultenWykoff1981}: the rate at which a fixed-vintage asset's rental value declines with age. The setting makes this unusually direct, because the asset's quantity is observed. We assemble a weekly panel of every model version sold through a major inference marketplace between 2024 and 2026, with tokens served, per-token prices, and an externally measured capability index for each version, reconstructed from Internet Archive captures of the marketplace and of an independent benchmarking organization.

Four results summarize the paper. First, a frontier version's revenue share halves within twelve weeks of its launch ramp and falls by 96 percent in its first year. Developers announce a version's retirement only after a median of eleven months, when it retains a fortieth of its peak share, and stop serving it after sixteen, so depreciation schedules built from service lives, or from the five-to-six-year lives of the servers that train and serve models, understate the rate several-fold. Second, the decline is obsolescence caused by rivals. Newer versions released by other firms take \rivalshare\ of the share a version loses, own successors a tenth to a third, depending on how closely a developer's versions substitute for one another, and price competition almost nothing, and the loss happens inside each application that uses the version rather than through turnover among applications. Third, the rate is predictable: in the demand system we estimate, a fixed version's share decays at the capability-sensitivity of demand times rivals' capability growth, and this formula, with no free parameter, matches the decay rates of twelve uses on average, in sample and out of sample. Fourth, model development creates two assets with very different lives. Versions depreciate at 3.2 log points a year, so net investment in them is a seventh to a quarter of gross investment; the capability a developer carries into its next versions depreciates at 0.5 to 0.7 log points a year. How much of development spending is version capital decides how much net capital AI investment creates, and no developer discloses it. Holding strategies fixed, the innovation race that generates this depreciation involves too little investment rather than too much: with logit demand and static pricing, business stealing is three quarters of the surplus developers receive and do not pay for, a ratio that depends only on market shares.

The mechanism behind these numbers is that the value of intelligence capital is relative. Developers who buy inference choose among models on capability and price, so what a model earns depends on how it compares with the alternatives available at that moment. When a rival releases a stronger model, the demand facing every incumbent shifts inward although nothing about the incumbents has changed. In the discrete-choice demand system we derive and estimate, the rate at which a fixed model's revenue share decays equals the capability-sensitivity of demand times the share-weighted growth in rivals' capability. The capability coefficient that drives the formula is stable when developer-by-quarter effects absorb brand and ecosystem, when capability is measured from different sets of benchmarks, and under bounds on selection on unobservables. A firm that stops improving its model does not keep its position; it loses value at the rate the frontier advances. This is the economics behind the Red Queen's remark that it takes all the running one can do to stay in the same place.\footnote{The term entered the organizational literature through \citet{BarnettHansen1996}; the economic content is that of creative destruction \citep{AghionHowitt1992}, applied here to an asset whose stock is observed and whose rivalry is measured.}

Section~\ref{sec:data} describes the data and defines the product as a model version, an asset whose characteristics are frozen at release. Section~\ref{sec:facts} documents three facts: steep and regular age profiles, declines that follow competitors' releases, and a quality-adjusted price of inference that falls by about 70 percent a year through entry rather than price cuts. Section~\ref{sec:model} derives the demand system and the decomposition, Section~\ref{sec:demand} estimates it, and Section~\ref{sec:decomp} decomposes each version's depreciation and follows it inside applications. Section~\ref{sec:measure} maps the estimates into the rates capital measurement uses and compares them with official rates, hardware lives, retirement schedules and developers' finances. Section~\ref{sec:dynamics} draws the implications for the innovation race.

The paper contributes to three literatures. The first is the measurement of capital and its depreciation. \citet{HultenWykoff1981} established that economic depreciation should be inferred from the price profiles of used assets and that it combines physical deterioration with obsolescence; \citet{CorradoHultenSichel2009} extended capital accounting to intangibles, and \citet{LiHall2020} and \citet{deRassenfosseJaffe2018} estimated depreciation rates for R\&D and patented innovations from firm value and renewal data. These approaches observe prices without quantities or value without the asset, and must separate obsolescence from deterioration by assumption. We observe the asset's quantity, its price, its quality, and the quality of every substitute at weekly frequency, for an asset that cannot deteriorate, and we show that the implied depreciation rate is four to eighteen times the rates now used to capitalize software, R\&D and the hardware on which models run. The 2025 System of National Accounts extends capitalization to data \citep{SNA2025}, recent proposals for AI satellite accounts treat training compute as model capital \citep{KorinekMcKelvey2026}, and the intangibles that complement a general-purpose technology are known to distort measured productivity while they accumulate \citep{BrynjolfssonRockSyverson2021}; none of these has a measured depreciation rate for the model itself. This paper supplies one, together with the distinction between the version and the capability a developer carries forward.

The second literature is the economics of innovation races and product turnover. In Schumpeterian models the incumbent's rent is destroyed at the rate at which rivals innovate \citep{AghionHowitt1992}, and in step-by-step models firms invest to escape rivals \citep{AghionHarrisHowittVickers2001,AghionBloomBlundellGriffithHowitt2005}. Product-level data have measured how much of growth is creative destruction \citep{GarciaMaciaHsiehKlenow2019,ArgenteLeeMoreira2018} and decomposed firm size with an estimated demand system \citep{HottmanReddingWeinstein2016}; discrete-choice demand has valued new products since \citet{Petrin2002} and \citet{Nevo2003}; and structural work has estimated quality ladders and innovation races in computers \citep{BresnahanGreenstein1999}, microprocessors \citep{GoettlerGordon2011} and hard disk drives \citep{Igami2017}. The welfare question is the classic one of business stealing against unappropriated surplus \citep{Spence1976,MankiwWhinston1986}. We measure the creative-destruction rate itself, at weekly frequency, for a general-purpose technology whose displacing products' quality is observed, and show that with logit demand and static pricing the welfare comparison turns on the outside option's share alone.

The third literature is the economics of artificial intelligence. Models of AI as a general-purpose technology \citep{BresnahanTrajtenberg1995,AgrawalGansGoldfarb2018,Acemoglu2025} treat capability as an input whose price falls, and recent work measures the market for model inference descriptively \citep{Fradkin2025,Demirer2025Market}. We take the supply side seriously: capability is produced by firms from compute \citep{HoffmannEtAl2022}, sold at posted prices, and valued relative to rivals, and its value to the firm that produced it is short-lived.

Two features of the setting discipline the interpretation. The marketplace we observe is a slice of the market for inference, used disproportionately by application developers who can switch, and Section~\ref{sec:data} shows where it is and is not representative; our rates are rates on contestable revenue, which is the margin on which the innovation race is fought. And all of our results concern the depreciation of \emph{model} capital, not of the hardware on which models run. The two are different assets with different lives, and amortizing model development over the life of the accelerators that trained it has no basis in our estimates.

\section{Data and Measurement}\label{sec:data}

Measuring the rental value of a fixed-vintage asset requires, for each model version, its quantity sold, its price, its capability, the capability and price of every substitute, and the dates at which substitutes appeared. We link two public sources, transaction quantities from an inference marketplace and characteristics, prices and release dates from an independent benchmarking organization, and reconstruct their histories from Internet Archive captures. Appendix~\ref{app:data} gives construction details.

\subsection{The unit of observation: a model version}\label{sec:unit}

The asset is a \emph{model version}: a specific set of trained weights, released on a specific date and sold under a specific name and price. Its capability is fixed at release. When a developer releases a new version under an existing family name, the older version continues to be sold, usually under a dated identifier, with unchanged weights. We therefore treat each version as a distinct product and each developer as a portfolio of versions. This makes the Hulten--Wykoff logic applicable: any decline in a version's revenue is obsolescence, not deterioration. Reasoning-effort settings of a version that share a per-token price are one product. Free tiers, served at zero price with rate limits, enter the demand system's choice set as separate products but are excluded from the depreciation measurement, since they earn no rental.

\subsection{Quantities: the OpenRouter marketplace}\label{sec:or}

OpenRouter \citep{OpenRouterRankings2026} resells API access to models from essentially every major developer through a single endpoint; a buyer switches models by changing one string and pays the developer's posted per-token price plus a small fee. Its public rankings, used descriptively by \citet{Fradkin2025}, \citet{Demirer2025Market} and \citet{Lisk2026} and published under a CC~BY~4.0 license, report each model's tokens. We reconstruct a weekly panel from 437 Internet Archive captures between November 2023 and October 2026, whose data payloads embed complete weekly series for every model, and use the official data API for the last 22 weeks, when the archived pages no longer embed the full table; Appendix~\ref{app:validation} validates the reconstruction against the API. The payload also reports requests, which give tokens per request. Platform volume grows from 12 billion tokens a week in February 2024 to 164 trillion by October 2026 (Figure~\ref{fig:volume}A); because a version's revenue can rise while its share collapses, we measure depreciation on revenue \emph{shares} and treat market growth separately.

\paragraph{Usage by category.} OpenRouter classifies requests into twelve uses, from programming (half of classified tokens) and technology (a fifth) to translation and trivia, and publishes tokens by model within each. The weekly category-by-version panel combines the data API (76 weeks from May 2025, the 50 largest models per category) with archived category pages (April 2024 to May 2025, complete tables), and has 50,480 category-version-weeks; it is used in Section~\ref{sec:byuse}.

\subsection{Characteristics, prices, and release dates}\label{sec:aa}

Artificial Analysis (AA) \citep{ArtificialAnalysis2026} is an independent benchmarking organization that runs every widely available model through a common battery of tests and publishes benchmark scores, posted input and output prices, output speed, latency, context window and whether the weights are open. We reconstruct its history from monthly captures of its leaderboard from April 2024 and its API from October 2026: 77,531 model-snapshot observations on 785 versions from 64 developers.

\paragraph{Capability.} AA's composite index is rescaled whenever benchmarks are revised, so we estimate instead a one-factor item-response model on the full benchmark-by-version matrix (Appendix~\ref{app:irt}), anchored so that GPT-4o of May 2024 is zero and the unit is one standard deviation of capability across priced models. Capability is a time-invariant property of the weights. The frontier is the maximum capability among versions available in a week.

\paragraph{Prices and release dates.} Prices are developers' posted list prices for the primary API, which OpenRouter charges; volume discounts and enterprise contracts do not apply to the pay-as-you-go traffic in the panel. Posted prices change rarely (4.8 percent of version-months for closed-weight versions), so a version's revenue path is almost never confounded by its own price changes. Archived captures of OpenRouter's model-listing endpoint and of AA's provider pages give a second record of prices and the third-party hosts of open-weight models (Appendices~\ref{app:data} and~\ref{app:facts}). A version's birth date is the earlier of AA's release date and its first week of traffic on OpenRouter; birth dates form the release calendar of Section~\ref{sec:decomp}.

\subsection{Linking the sources and constructing revenue}\label{sec:link}

A crosswalk matches OpenRouter identifiers to AA versions on developer and name tokens, with every match among the 300 largest models verified by hand. Matched versions carry 73 percent of tokens in the first half of 2024 and 81 to 97 percent thereafter (Table~\ref{tab:coverage}); unmatched traffic, mostly community fine-tunes and small unbenchmarked models, stays in every share's denominator and is the outside option in the demand system. Because the data-API weeks report the long tail as one aggregate, the outside option is measured consistently only in the archived weeks, which the demand estimation therefore uses. Revenue of version $i$ in week $t$ is
\begin{equation}
R_{it}=p^{\mathrm{in}}_{it}\,Q^{\mathrm{in}}_{it}+p^{\mathrm{out}}_{it}\,Q^{\mathrm{out}}_{it},
\label{eq:revenue}
\end{equation}
and its revenue share is $R_{it}/\sum_j R_{jt}$ over paid versions; cached tokens are priced at the posted cache price where reported, and the platform fee cancels in shares. The linked panel contains 15,507 product-weeks on 451 products (358 versions) from 41 developers over 138 weeks, February 2024 to October 2026 (Appendix Table~\ref{tab:summary}). The share distribution is extremely skewed: the ten largest versions account for 88 percent of paid revenue in the median week.

\subsection{Representativeness}\label{sec:represent}

OpenRouter is a slice of the market for inference. Its buyers are application developers who want the option to switch: coding agents and general-purpose agents dominate its largest applications by 2026 (Figure~\ref{fig:volume}C), multihoming is common \citep{Fradkin2025}, and large enterprises with negotiated contracts and developers' own consumer products are absent. Our buyers are therefore the most price- and capability-sensitive segment, and our rates are those at which a model loses value among buyers who can leave. That is the margin through which an entrant gains share. Two checks bound the concern. The platform's token volume grew no faster than the tokens Google reports processing across all its products, although its paid revenue grew faster than the largest developers' revenue (Section~\ref{sec:payback}). And versions lose share inside each application that uses them, not through turnover among applications (Section~\ref{sec:apps}). The segments we do not observe likely differ in opposite directions: developers' own consumer products usually switch their default model at release, so a version's traffic there falls at least as fast, while enterprise deployments with validated integrations migrate more slowly.

\begin{figure}[htbp]
\centering
\includegraphics[width=\textwidth]{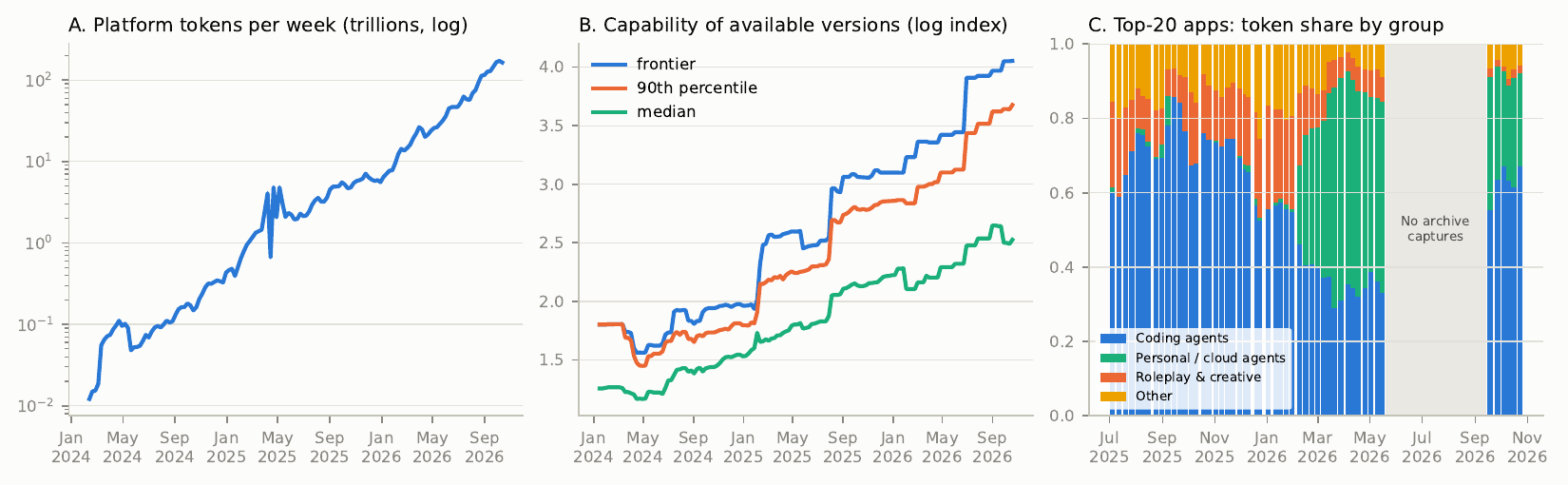}
\caption{The market for inference, 2024--2026}\label{fig:volume}
\begin{minipage}{0.95\textwidth}\footnotesize\emph{Notes:} (A) Platform tokens per week. (B) Log chained capability index of the best, 90th-percentile and median version listed by AA each week. (C) Share of the twenty largest applications' tokens by application group, from the application rankings embedded in the captures (no application data between mid-May and mid-September 2026).\end{minipage}
\end{figure}

\section{Three Facts}\label{sec:facts}

Before imposing a model, we document three regularities that any account of the market for inference must match. The first is the object of the paper: model versions lose revenue share along a steep and regular age profile. The second identifies the cause: comparable rivals arrive almost continuously, and an incumbent's share falls in the weeks after each arrival while its own quality and price do not move. The third is the aggregate counterpart: the quality-adjusted price of inference falls at about 70 percent per year, and the token-weighted posted prices of continuing versions do not fall.

\subsection{Fact 1: Revenue shares of model versions decay fast and regularly}\label{sec:fact1}

For each paid version we compute its weekly revenue share from equation~\eqref{eq:revenue} and index time by age, the number of weeks since the version's first week with traffic on the platform. Because a version's first weeks are a launch ramp during which developers integrate it, we normalize each version's path by its share at age four and follow it for a year. To avoid selecting on noise we do not normalize by the version's peak share; Appendix~\ref{app:facts} shows that doing so steepens the profile further. The sample is the 176 versions whose entry is observed and that are observed for at least 26 weeks; 62 of them are \emph{frontier} versions, defined as versions whose capability at entry is in the top decile of versions available in their entry week.\footnote{This definition of the frontier is used throughout the paper. Capability for this purpose is the chain-linked composite index, which covers every version; the item-response scale covers 88 percent of them. Defining the frontier instead as the top decile among the week's entrants changes the attribution of Section~\ref*{sec:decomp} in the direction of a smaller share for rivals, from four fifths to about two thirds (Appendix Table~\ref{tab:robustsummary}). Throughout, capability is measured in standard deviations of the item-response scale across priced models, except for the comparable-rival band of Fact~2 and the rows labelled chained index, which use log points of the chained index; and ``archived weeks'' are the weeks to May 2026 reconstructed from Internet Archive captures.}

\begin{figure}[htbp]
\centering
\includegraphics[width=0.85\textwidth]{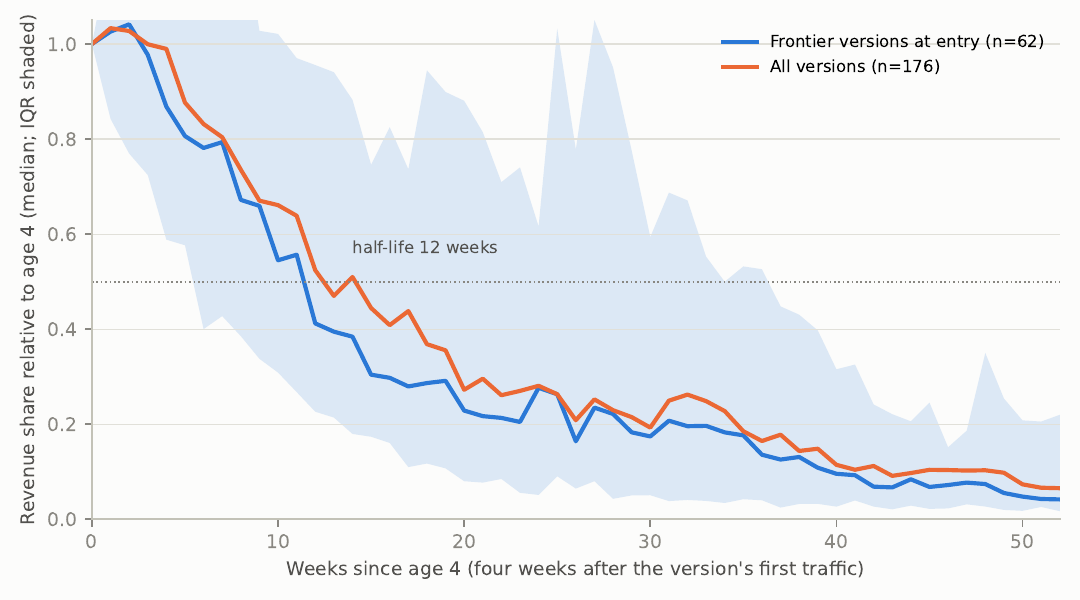}
\caption{The age-revenue profile of model versions}\label{fig:age}
\begin{minipage}{0.9\textwidth}\footnotesize\emph{Notes:} Median revenue share relative to the version's share at age four, by weeks since age four; the shaded band is the interquartile range across frontier versions. A version is a model version sold at a positive list price; frontier versions have capability at entry in the top decile of available versions. Sample: versions entering between March 2024 and April 2026 with at least 26 weeks of observation. Source: OpenRouter (openrouter.ai/rankings) and Artificial Analysis.\end{minipage}
\end{figure}

Figure~\ref{fig:age} shows the result. The median frontier version has lost half of its revenue share twelve weeks after its launch ramp, and three quarters of it after twenty weeks; a year later it retains about four percent. The profile for all versions is marginally flatter, with a half-life of thirteen weeks. The interquartile band is wide in levels but the shape is common: no version in the upper quartile of the paths is still gaining share after month four. A within-version regression of log revenue share on age, with version fixed effects and ages from four to fifty-six weeks, gives a decay of $0.070$ log points per week for frontier versions ($0.061$ for all versions), which implies a half-life of ten weeks (twelve weeks on the median path above, the figure used in the abstract) and a first-year decline of more than ninety-five percent in the share.\footnote{Estimates from 62 frontier versions (2,511 version-weeks) and 176 versions (7,501 version-weeks). The regression uses only within-version variation and is therefore unaffected by differences in peak shares across versions.} The profile is not the same in every use. Within programming, which is half of the platform's classified traffic, a frontier version's share of category tokens decays at 0.081 log points per week, a half-life of nine weeks; within role-play at 0.024, a half-life of seven months; within translation it does not decay at all (Section~\ref{sec:byuse}, Table~\ref{tab:category}). The platform-wide profile is the token-weighted average of these, dominated by programming.

This is a share, measured against a market whose revenue grows by 2.5 log points, about twelvefold, over a typical frontier version's first year; the revenue \emph{level} of the median frontier version that is still sold a year later is about unchanged on the platform, 8 percent lower, or 39 percent lower when the frontier is defined as the top decile of the week's entrants (Section~\ref{sec:decomp}). The platform's revenue grows faster than the total revenue of the largest developers, and measured against their growth the fall in the level would be 67 to 78 percent (Section~\ref{sec:payback}); the share, which is what depreciation refers to, does not depend on which market growth is used. A share can also fall because of the firm's own later versions as well as rivals'; Fact~2 and Section~\ref{sec:decomp} separate the two and find that own successors account for a tenth of the decline under logit demand and about a third when a developer's versions are allowed to substitute closely for one another. And the profile mixes age, cohort and calendar-time effects, which no regression can separate, since age is the week minus the entry week; the accounting of Section~\ref{sec:decomp} resolves the identity by replacing the calendar-time effect with the observed share gains of newer rival versions, as the model implies (Appendix~\ref{app:decomp}).

\subsection{Fact 2: Comparable rivals arrive almost every week, and incumbents' shares fall after they do}\label{sec:fact2}

Figure~\ref{fig:events} shows the flagship version of each of four developers, the version with the largest peak revenue share among those observed for at least forty weeks, with its weekly revenue share plotted against calendar time. Dashed lines mark the release dates of rival versions of comparable or higher capability, defined as versions from other developers whose capability at release is within 0.25 log points of the chained index below the flagship's or above it (about 0.28 standard deviations of the item-response scale); dotted lines mark the developer's own successor releases. The dashed lines are dense: over its life a flagship faces between 48 and 74 comparable rival releases, more than one every two weeks, while its own developer releases a successor seven to nineteen times. Across all versions, a version in the top quartile of capability enters in 24 percent of weeks, at a rate of 0.54 per week, and a top-decile version enters eleven times a year. Over the same period the flagship's own characteristics do not move: its capability is fixed by construction, and its posted list price changes between zero and five times over its life, at most once in two of the four cases. The declines are therefore not responses to anything the flagship did.

\begin{figure}[htbp]
\centering
\includegraphics[width=0.95\textwidth]{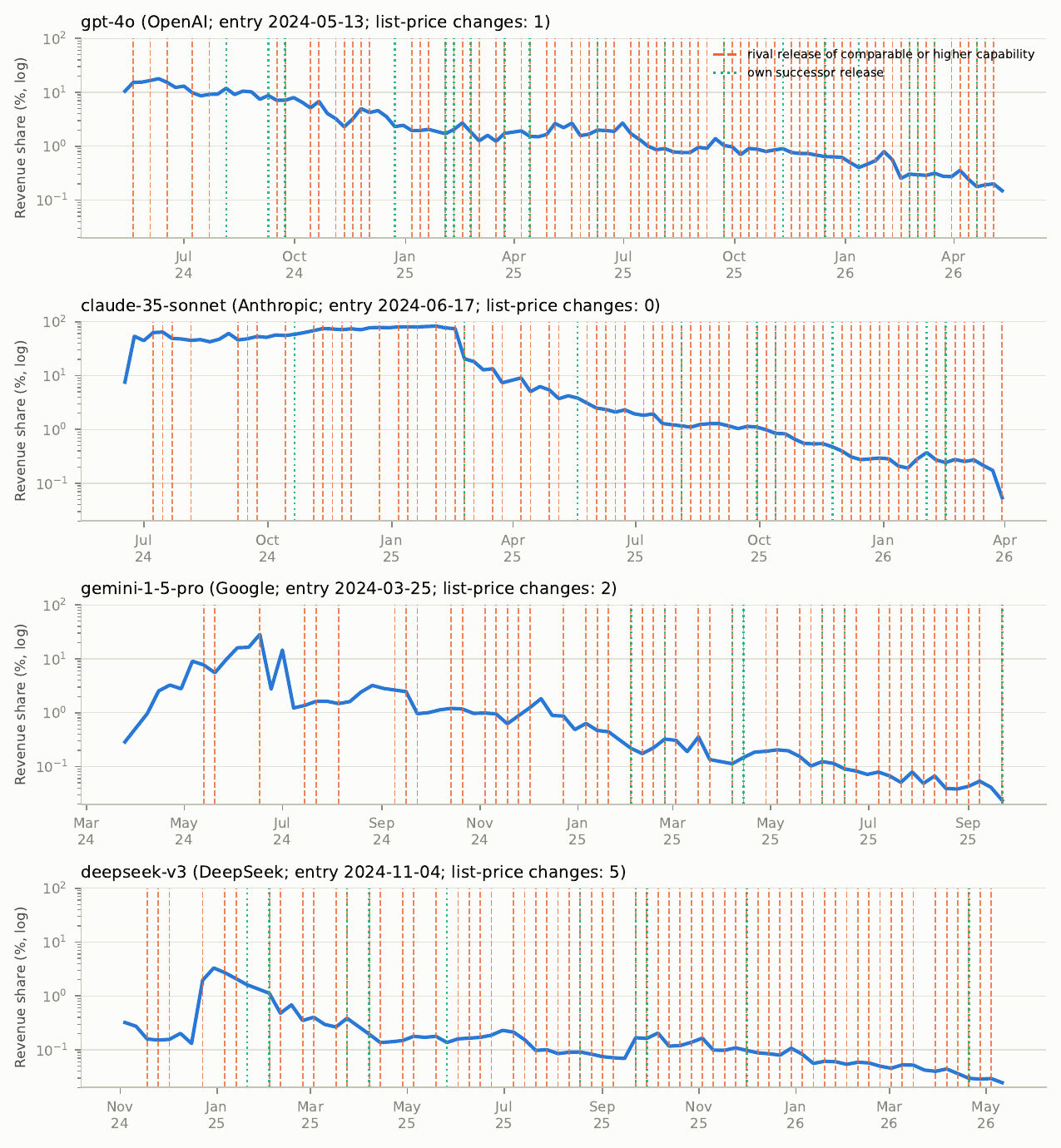}
\caption{Flagship versions' revenue shares around rival and own-successor releases}\label{fig:events}
\begin{minipage}{0.9\textwidth}\footnotesize\emph{Notes:} Weekly revenue share among paid versions. Dashed lines: releases by other developers of versions with capability within 0.25 log points below the flagship's or above; dotted lines: releases by the flagship's own developer of versions with capability at or above the flagship's. The number of list-price changes over the version's life is given in each title.\end{minipage}
\end{figure}

The pattern is the same in the cross-section of all versions. In the four weeks after a comparable rival release an incumbent's log revenue share falls by 0.26 on average relative to the four weeks before (9,306 version-events); after an own-successor release it falls by 0.28 (1,448 events). Because comparable releases are so frequent, windows without one are rare: only 310 four-week windows in the panel contain no comparable release within four weeks, and in them the average change is $-0.24$. An incumbent can therefore almost never be observed for a month without rival innovation, so the attribution of the decline to rivals cannot rest on comparing event windows with quiet windows. It must rest instead on who gains the share that the incumbent loses, as in the accounting of Section~\ref{sec:decomp}, and on the magnitude of the response to each release (Appendix~\ref{app:decomp}).

Claude 3.5 Sonnet in the second panel is the clearest single illustration. It held more than 80 percent of paid revenue on the platform in early 2025, with no price change over its life, and lost three quarters of that share within four weeks of the release of its own successor in late February 2025, with comparable rival releases in the same weeks. GPT-4o in the first panel shows the more common pattern: a ramp to 18 percent, then a stepwise decline across a sequence of rival releases, to under 1 percent a year after launch, with a single price change.

\subsection{Fact 3: The quality-adjusted price of inference falls at about 70 percent per year, and the prices of continuing versions do not fall}\label{sec:fact3}

Figure~\ref{fig:price} plots three price indexes built from the version panel, each normalized to the first week of 2024. The hedonic index holds capability and the other characteristics fixed: it is the exponentiated week effect from a token-weighted regression of the log list price on log capability, log output speed, log context window, a reasoning indicator and an open-weights indicator, with week fixed effects. It falls at 71 percent per year, a halving time of under seven months. The matched-model index, which chains the price changes of versions present in consecutive weeks, falls at 14 percent per year unweighted and \emph{rises} by about 11 percent token-weighted, because two heavily used versions raised their promotional launch prices once adoption was established (Gemini 2.5 Flash in July 2025 and DeepSeek V4 Flash in August 2026). List prices change in only 3.5 percent of version-weeks in the paid version panel, open-weight versions included (4.8 percent of version-months for closed-weight versions alone, Section~\ref{sec:aa}), and when they change on a version that matters they go up as often as down; the archived marketplace listings, which record 2,284 price changes on 1,062 models, say the same (Appendix~\ref{app:facts}).

\begin{figure}[htbp]
\centering
\includegraphics[width=0.85\textwidth]{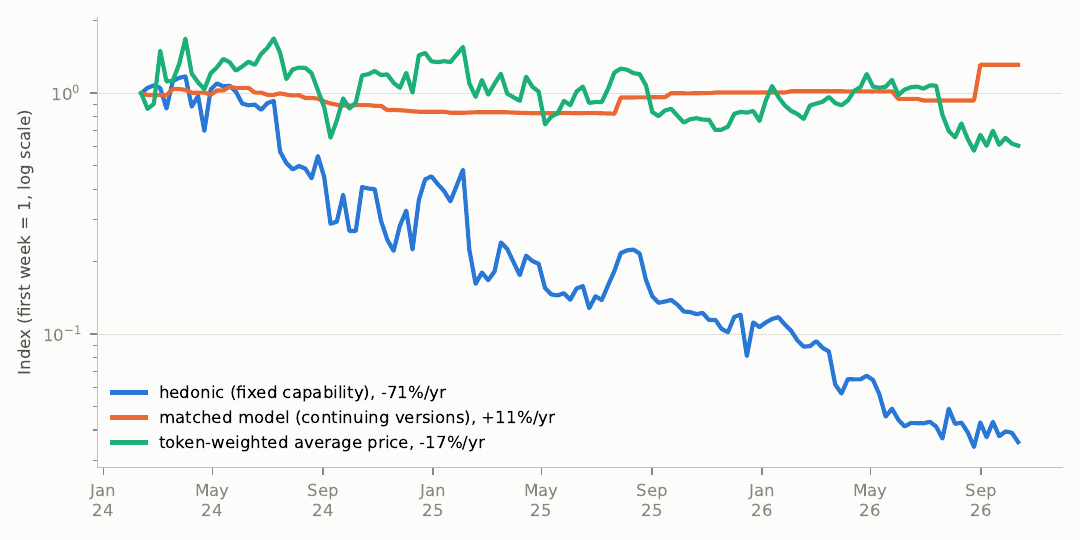}
\caption{Hedonic, matched-model, and average price indexes for inference}\label{fig:price}
\begin{minipage}{0.9\textwidth}\footnotesize\emph{Notes:} Weekly indexes from the paid version panel, February 2024 to October 2026, first week $=1$, log scale. Hedonic: week effects from a token-weighted regression of log price on log capability, log speed, log context, reasoning and open-weights indicators. Matched model: chained token-weighted log price changes of versions present in consecutive weeks. Legend rates are signed annual changes. The token-weighted average price paid falls at 17 percent per year, the composition effect of buyers moving to cheaper capability.\end{minipage}
\end{figure}

The three series together say where the decline comes from: entry. New versions arrive with more capability per dollar, buyers move to them, and the price of capability falls, while the posted prices of the versions already on the market stay where they are. For an incumbent version this means that competition arrives as obsolescence rather than as pressure on its price: its price does not fall, its demand does. Section~\ref{sec:decomp} confirms that price erosion is a negligible component of a version's depreciation, and Section~\ref{sec:measure} uses the hedonic rate to bound the fall in the price of new model capital.

Together the three facts say that model versions are assets that lose their earning power within months, that the loss follows the arrival of better substitutes, which arrive continuously, and that the market-wide expression of those arrivals is a collapse in the price of capability rather than in the prices of existing products. The rest of the paper builds the model that makes these facts into a measurement of depreciation and a statement about welfare.

\section{A Model of the Market for Inference}\label{sec:model}

This section sets out the demand system that the estimation uses, defines the object the paper measures, and states the two results that organize the rest of the paper: the decomposition of a fixed version's depreciation, and the comparison of the private and social returns to capability. Developers buy tokens from upstream firms that own model versions. A developer's choice of supplier follows from its own profit maximization, which gives the demand coefficients a structural interpretation; the microfoundation is in Appendix~\ref{app:micro} and all proofs are in Appendix~\ref{app:proofs}.

\subsection{Demand and pricing}\label{sec:env}

A finite set of upstream firms owns model versions; we index versions by $i$ and write $f(i)$ for the firm that owns $i$. Version $i$ has capability $K_i$, fixed at release, and sells at a per-token price $p_i$. An open-weight alternative, $i=0$, has capability $K_0(t)=\kappa_0\Abar(t)$ and is sold at marginal cost; $\Abar(t)$ is the frontier, growing at rate $g_A$, and $k_i\equiv K_i/\Abar$ is capability relative to it. A developer $j$ that buys from version $i$ earns a maximized profit $\Pi^*_{ji}=\Gamma_jK_i\Psi(p_i)$, where $\Gamma_j$ is the developer's own scale and $\Psi$ is decreasing in the token price (Lemma~\ref{lem:indirect}); it chooses the version that maximizes $\theta\ln\Pi^*_{ji}+\epsilon_{ji}$ with $\epsilon$ type-I extreme value. Because $\Gamma_j$ cancels from the comparison,
\begin{equation}
s_i(\bk,\bp)=\frac{\big[k_i\Psi(p_i)\big]^{\theta}}{\big[\kappa_0\Psi(c_0)\big]^{\theta}+\sum_{m\ge 1}\big[k_m\Psi(p_m)\big]^{\theta}},
\label{eq:logit}
\end{equation}
and the price response comes from the same primitive: with $\lambda_i\equiv-\theta\,\dd\ln\Psi/\dd\ln p_i>0$,
\begin{equation}
\frac{\partial\ln s_i}{\partial\ln p_i}=-\lambda_i(1-s_i).
\label{eq:lambda}
\end{equation}
The capability coefficient $\theta$ is the inverse of the scale of the match noise, and $\lambda$ is the coefficient on the log price that Section~\ref{sec:demand} estimates. Token demand is $Q_i=\Omega_t\,s_i\,K_i\,\chi(p_i)$, where $\Omega_t$ is the mass of downstream developers and $\chi$ the tokens used per unit of output, whose price elasticity $\varepsilon_\chi$ is the intensive margin.

Two features of the market shape pricing. Posted prices are set when a version is released and rarely change afterwards (Fact~3), and developers migrate slowly: a buyer already integrated with a version re-optimizes with probability $1-\rho$ a week (Section~\ref{sec:migrate}). Shares therefore follow $\ln s_{it}=\rho\ln s_{i,t-1}+(1-\rho)\ln s^*_{it}$, where $s^*$ is the logit share \eqref{eq:logit}, and a change in $i$'s price $k$ weeks into its life moves its log share by $-\lambda(1-s_i)(1-\rho^{k+1})$ after $k$ weeks. A version's buyers, however, all adopt it at its launch price, so the elasticity that governs the launch price is the full one. A firm $f$ that sets the launch prices of its versions once, to maximize their discounted lifetime profits given rivals' prices, then sets the multi-product logit margin
\begin{equation}
\frac{p_i-MC_i}{p_i}=\frac{1}{\lambda(1-S_f)+\varepsilon_\chi},\qquad S_f=\sum_{j\in f}s_j ,
\label{eq:lerner}
\end{equation}
common to its versions, where $S_f$ is the firm's inside share: a lower price on one version partly draws buyers from the firm's own versions. Its reduced-form profit $\pi^*_i$ has, holding prices fixed,
\begin{equation}
\frac{\partial\ln\pi^*_i}{\partial\ln k_i}=1+\theta(1-s_i),\qquad
\frac{\partial\ln\pi^*_i}{\partial\ln k_m}=-\theta s_m\quad(m\neq i).
\label{eq:pi_elast}
\end{equation}
Equation~\eqref{eq:lerner} is a restriction that the estimates must respect: profitable pricing requires $\lambda(1-S_f)+\varepsilon_\chi>1$. The same $\lambda$ enters three places: the launch margin, the cross-section of shares, and the path of a version's share after a change in its own price, scaled by $1-\rho^{k+1}$. Section~\ref{sec:lambdas} confronts the three.

\subsection{Discrete releases}\label{sec:releases}

Capability advances reach the market as discrete versions. Let a firm's underlying capability grow at rate $g$, let its released capability $\hat k_f$ stay fixed between releases, and let a release cost a fixed amount $F$ (evaluation, safety work, serving setup).

\begin{proposition}[Release threshold]\label{prop:release}
Let $G(\Delta)\equiv\Pi(\hat ke^{\Delta})-\Pi(\hat k)$ be the flow profit gain from releasing a version $\Delta$ log points above the released one, with rivals held fixed, and suppose $G$ is S-shaped in $\Delta$, as it is under \eqref{eq:logit} for a firm whose share is below one half. The firm releases when $\Delta$ first reaches the threshold $\Delta^*$ that solves
\begin{equation}
g\,G'(\Delta^*)=r\,\big[G(\Delta^*)-rF\big],
\label{eq:release}
\end{equation}
provided $G(\Delta^*)>rF$; the interval between releases is $\Delta^*/g$, and $\Delta^*$ is increasing in $F$.
\end{proposition}

On the convex part of $G$ waiting pays, so firms do not release marginal improvements; on the concave part the gain saturates with the share, so they do not wait indefinitely. Releases are therefore periodic, a firm's own successor displaces its predecessor in one step, and across $N$ firms with staggered cycles a comparable release arrives roughly every $\Delta^*/(Ng)$ units of time, which for the frontier firms in the sample is a few weeks. The decomposition below is written for discrete releases: rivals' capability growth is realized as the entry and adoption of new versions, which is how Section~\ref{sec:decomp} measures it.

\subsection{Economic depreciation}\label{sec:dep}

\begin{definition}[Economic depreciation of a fixed-vintage version]\label{def:dep}
Freeze version $i$ at date $t_0$, so that $K_i$ is constant and $k_i(t)=k_i(t_0)e^{-g_A(t-t_0)}$. Its economic depreciation rate is the proportional decline of its rental value, $\delta^{\mathrm{econ}}_i(t)\equiv-\dd\ln\pi^*_i(t)/\dd t$ holding $K_i$ fixed.
\end{definition}

This is the \citet{HultenWykoff1981} notion applied to a revenue profile. Since weights do not deteriorate, all of $\delta^{\mathrm{econ}}$ is obsolescence. Let $g_m\equiv\dot K_m/K_m$ be the capability growth of alternative $m$ (for a portfolio, the share-weighted growth of its versions), $g_0\equiv g_A$, $g_\Omega\equiv\dot\Omega/\Omega$ the growth of the downstream mass, and $\varpi_i$ the gross margin.

\begin{proposition}[Decomposition of economic depreciation]\label{prop:decomp}
For a fixed-vintage version $i$,
{\small\begin{align}
\delta^{\mathrm{econ}}_i
&=\underbrace{\theta\Big[s_0\,g_A+\sum_{m\ge 1,\,m\ne i}s_m\,g_m\Big]}_{\text{competitive obsolescence}}
+\underbrace{\Big[\lambda_i\hat p_i-\sum_{m\ge 0}s_m\lambda_m\hat p_m\Big]}_{\text{relative price erosion}}
+\underbrace{\varepsilon_\chi\,\hat p_i}_{\text{intensive margin}}
-\underbrace{g_\Omega}_{\text{market growth}}
-\underbrace{\hat\varpi_i}_{\text{margin}} ,
\label{eq:decomp}
\end{align}}
where hats denote proportional rates of change.
\end{proposition}

Every term is either observed ($s_m$, $g_m$, $\hat p_m$, $g_\Omega$) or an estimated demand parameter. The first is the pecuniary externality: a version loses rental value at the capability-sensitivity of demand times the share-weighted capability growth of its alternatives. When $i$ indexes versions, the rival sum splits into versions of other firms, the externality, and newer versions of the same firm, self-cannibalization, which the data separate by the identity of the releasing firm. Because the market-growth term is common to all versions alive at a date, depreciation measured on revenue \emph{shares} is the decline of a version's rental value relative to the market at a given date, which is the concept capital measurement requires (Section~\ref{sec:measure}).

\begin{corollary}[Red Queen]\label{cor:rq}
With constant prices and margins, a version whose rivals keep pace with the frontier loses rental share at rate $\theta(1-s_i)g_A$; if every proprietary firm stops improving, each still loses share at rate $\theta s_0g_A$ through the open-weight option.
\end{corollary}

\begin{proposition}[Rival capability and the marginal value of capability]\label{prop:cross}
Holding prices fixed, $\partial\pi^*_i/\partial k_i$ is strictly decreasing in a rival's capability $k_m$ if and only if $s_i<(1+\theta)/(2\theta)$.
\end{proposition}

For the estimated $\theta$ the threshold, 0.68, exceeds every version's share in the demand system, which includes the outside option (the largest is 0.47, Claude 3.5 Sonnet in February 2025, whose 83 percent is a share of paid revenue), so rival innovation lowers the marginal value of an incumbent's capability as well as its level.

\subsection{Private and social returns to capability}\label{sec:dynmodel}

Firm $f$ raises the capability of its best version by buying compute: $\dot k_f=\mathcal F(c_f)-(\delta_0+g_A)k_f$, where $\mathcal F$ is increasing, $\delta_0\ge0$ is decay of the serving stack, and the term $g_Ak_f$ is the treadmill that frontier growth imposes on relative capability. The firm maximizes the present value of $\pi^*_f(\bk)-p_cc_f-\Xi(c_f)$ with $\Xi$ convex; $q_f\equiv\partial V_f/\partial k_f$ is its shadow value of capability. A planner who maximizes the sum of profits, developers' surplus and the costs of compute has shadow value $q^S_f$. Developers' surplus is the expected maximum of $\theta\ln\Pi^*_{ji}+\epsilon_{ji}$, converted to money with the price coefficient.

\begin{proposition}[Private and social value of capability]\label{prop:welfare}
On a symmetric path, holding rivals' investment policies fixed, $(r+\delta_0+g_A)(q_f-q^S_f)=B_f-U_f$ with
\[
B_f=\theta\,s_f\sum_{m\neq f}\pi^*_m\big/k_f,\qquad U_f=\theta\,R_f\big/(\lambda\,k_f),
\]
where $R_f$ is version $f$'s revenue: $B_f$ is the profit that $f$'s capability takes from rivals (business stealing) and $U_f$ the developers' surplus that $f$ does not appropriate. Equilibrium investment exceeds the planner's if and only if $B_f>U_f$, and
\begin{equation}
\frac{B_f}{U_f}=\lambda\,\varpi\sum_{m\ne f,0}s_m=\lambda\,\varpi\,(1-s_0-s_f),
\label{eq:BU}
\end{equation}
with a common gross margin $\varpi$. If prices satisfy \eqref{eq:lerner} with $\varepsilon_\chi=0$,
\begin{equation}
\frac{B_f}{U_f}=\frac{1-s_0-s_f}{1-s_f}<1\quad\text{whenever } s_0>0 .
\label{eq:BUbertrand}
\end{equation}
\end{proposition}

The capability coefficient cancels from the ratio because it scales both the share that a capability gain moves and the surplus that the gain creates: a log point of capability is worth $\theta/\lambda$ log points of price to every developer who uses the version. The ratio therefore depends on the price coefficient, the margin and the rivals' share, and under static pricing on the shares alone. The wedge in \eqref{eq:BUbertrand} is the outside option. The surplus that a capability gain creates for developers who switch from the open-weight alternative is not taken from any proprietary rival, so with logit demand and static pricing the direct effect always points to under-investment, by a margin equal to the outside option's share of the market not held by $f$. Over-investment can arise only from margins above the static level or from the strategic responses that Proposition~\ref{prop:welfare} holds fixed, which Section~\ref{sec:race} examines.

\section{Demand Estimation}\label{sec:demand}

The Red Queen formula, the decomposition and the welfare comparison require the capability coefficient $\theta$ and the price coefficient $\lambda$, and the pricing equation \eqref{eq:lerner} requires the intensive-margin elasticity $\varepsilon_\chi$. This section estimates them from the version panel. The estimating equation is the log-linearization of \eqref{eq:logit}; the identification problem is the usual one for prices, and the available instruments do not solve it. We report what the instrumented estimate measures and, in Section~\ref{sec:lambdas}, set it against two other sources of the price coefficient.

\subsection{Specification}\label{sec:spec}

Taking logs of \eqref{eq:logit} relative to the outside option and linearizing $\ln\Psi(p)$ in the log price gives the estimating equation
\begin{equation}
\ln s_{it}-\ln s_{0t}=\theta\,k_i-\lambda\,\ln p_{it}+\beta' x_{it}+\xi_{f(i)}+\xi_t+\xi_{it},
\label{eq:est}
\end{equation}
where $s_{it}$ is the version's share of platform tokens, $s_{0t}$ is the share of the outside option (free tiers, unmatched models and the long tail), $k_i$ is the version's capability on the item-response scale of Appendix~\ref{app:irt}, in standard deviations across priced models, $x_{it}$ collects log output speed, log context window, a reasoning-model indicator and an open-weights indicator, and $\xi_{f(i)}$ and $\xi_t$ are developer and week fixed effects. Capability does not vary within a version, so $\theta$ is identified from the cross-section of versions within developer; a specification with version fixed effects identifies $\lambda$ from within-version price changes alone and is reported as a check. Markets are weeks and the sample is the archived weeks, February 2024 to May 2026, in which the outside option is measured consistently (Section~\ref{sec:link}); the data-API weeks are a robustness sample. The main sample has 11,004 paid version-weeks with complete characteristics.

\subsection{Instruments}\label{sec:iv}

Prices are set when versions are released and are likely correlated with unobserved quality, so $\ln p_{it}$ and $\xi_{it}$ are correlated. We examined two instruments, and neither delivers quasi-experimental variation in prices. The first uses the release of DeepSeek's V3 and R1 models in late January 2025 at prices an order of magnitude below incumbents': an indicator for the eight weeks from January 20 to March 16, 2025 interacted with each version's pre-event log price gap to DeepSeek V3. Posted prices barely move (with version fixed effects the instrument shifts the log price by 0.028, standard error 0.011), so inside the window the instrument is essentially the version's own relative price: its first-stage coefficient, 0.49, is the cross-sectional relation between a version's price and its pre-event price, and the estimate is the within-window cross-sectional price gradient. The diagnostic is direct: the same construction applied to every eight-week window in the panel gives a first-stage $F$ above 10 in all 26 windows and a price coefficient between 0.49 and 0.85 (10th and 90th percentiles), with the DeepSeek window at the median, and adding the pre-event gap interacted with a linear time trend and with age, which absorb high-priced versions' own life-cycle trends, removes the first stage (coefficient 0.03, $F=3.3$; Appendix Figure~\ref{fig:pricing}). The second is a cost shifter: the mean log price at which the cheapest third-party host serves, in the same week, other developers' models in the same quartile of parameter count, which moves with hardware and serving efficiency by model size. It predicts prices across versions ($F=52$, $\lambda=0.52$) but not within versions ($F=0.7$). We therefore interpret the estimated $\lambda$ for what it is: the gradient of log shares on log prices across versions, conditional on capability, the controls and developer and week effects. If higher-priced versions have unobserved advantages, it is biased toward zero. Section~\ref{sec:lambdas} sets it against the coefficient implied by within-version price changes and by pricing. Additional instruments in the spirit of \citet{BerryLevinsohnPakes1995} and \citet{GandhiHoude2019}, the mean capability of recently released rival versions and the number of rivals near $i$'s capability, and the log number of hosts serving the version, are weak (Appendix Tables~\ref{tab:demandweak} and~\ref{tab:firststage}).

The window instrument reported in Table~\ref{tab:demand} has a clustered first-stage $F$ of 103 on the main sample, which by the diagnostic above reflects the cross-sectional relation between a version's price and its pre-window price rather than a price shock; adding the recent-rival-capability instrument gives a Sargan statistic \citep{Sargan1958} of 0.06. Using all the alternative instruments together, the over-identifying restrictions are rejected (Sargan statistic 1724.5, Appendix Table~\ref{tab:demandweak}).

\subsection{Results}\label{sec:demres}

Table~\ref{tab:demand} reports the estimates on the main sample and scale: OLS, the instrumented main specification, two checks on the capability coefficient discussed below, and the nested logit of Section~\ref{sec:decres}. In the main specification (column 2), item-response capability, archived weeks and the window instrument, $\theta=2.74$ (standard error 0.37) per standard deviation of capability and $\lambda=0.69$ (0.19); the weak-instrument-robust Anderson--Rubin \citep{AndersonRubin1949} 95 percent set for $\lambda$ is $[0.32,1.07]$, which excludes zero and coincides with the Wald interval. OLS gives $\theta=2.69$ (0.30) and $\lambda=0.66$ (0.10), as the diagnostic above implies. On the chained index the same specification gives $\theta=3.00$ (0.45) per log point and $\lambda=0.84$ (0.19) (Appendix~\ref{app:robust}); the two capability scales have a rank correlation of \irtcorr\ across versions and imply the same share multiplier to within ten percent. Three features of the estimates carry over to the rest of the paper. First, $\theta$ is large: a capability advantage of one standard deviation of the item-response scale multiplies a version's share by a factor of about fifteen (about nine per standard deviation across the versions in the sample). This is the magnitude that makes rival innovation expensive for incumbents. Second, $\lambda$ implies an own-price elasticity of $-\lambda(1-s_{it})$, about $-0.7$ for the median version. That is below one in absolute value, which the pricing condition \eqref{eq:lerner} rules out; Section~\ref{sec:lambdas} compares it with the coefficients implied by within-version price changes and by margins. The within-version estimate with version fixed effects is 0.36 (0.22) by OLS. Third, in the logit specifications the open-weights indicator is positive conditional on capability and price (it is negative in the nested logit), and the reasoning indicator is negative, which we read as the token-intensity cost of reasoning models for a given task. Appendix~\ref{app:robust} shows that the estimates are stable when DeepSeek's own versions are excluded, when the sample is restricted to 2025 onward or extended to the data-API weeks, and it reports the window-by-window diagnostic.

Capability is measured with error, and a mismeasured regressor attenuates $\theta$. The item-response model reports a posterior standard deviation for each version, 0.11 on average, against a standard deviation of capability of 0.81 across versions and 0.45 within developer after the controls, the variation that identifies $\theta$. The implied reliability is 0.98 across versions and 0.93 within developer, so attenuation is at most seven percent. A second check uses the chained composite index, whose errors come from benchmark revisions rather than from the item-response fit, as an instrument for item-response capability alongside the window instrument for price: $\theta=2.78$ (0.41) and $\lambda=0.70$ (0.19), indistinguishable from the main specification (Table~\ref{tab:demand}, column 4). Unobserved quality that is correlated with capability within developer, such as reliability or tool use, is not removed by either check; it would load onto $\theta$, and the cross-use test below, which predicts depreciation rates from $\theta_c$ without any free parameter, is the external check on the magnitude.

\paragraph{Unobserved quality.} The capability coefficient would be biased if benchmark capability were correlated with unmeasured quality, such as reliability, tool use, or a developer's brand and ecosystem. Four checks suggest it is not driven by such factors (Appendix Table~\ref{tab:robustsummary}). Developer-by-quarter fixed effects, which absorb whatever moves with the developer over time, change it by 5 percent (Table~\ref{tab:demand}, column 3). Measuring capability from knowledge items only or agentic items only, benchmark sets with different exposure to developers' targeting, lowers it by 10 and 30 percent, the attenuation expected from noisier measures with fewer items; instrumenting the all-item measure with the agentic-item score, whose measurement error is independent, gives a coefficient 8 percent above the baseline on the same sample. Adding the observable characteristics raises rather than lowers the coefficient, so under the proportional-selection bound of \citet{Oster2019} unobservables would have to work against the observables to remove it, and the bias-adjusted coefficient is 6 percent above the controlled one. The Epoch Capabilities Index \citep{EpochECI2026}, built from a different set of benchmarks by a different evaluator, gives a coefficient within 11 percent of the baseline on the versions both measure (Appendix~\ref{app:eci}).

\subsection{Depreciation and demand by use}\label{sec:byuse}

The demand system of Section~\ref{sec:model} is derived for a developer with a given task, and nothing in it requires the capability coefficient to be the same across tasks: a coding agent that fails when its model fails and a trivia application that does not will weigh capability differently, and Corollary~\ref{cor:rq} then says that a fixed version depreciates faster in the uses where $\theta$ is larger. The category panel of Section~\ref{sec:or} lets us test this directly. For each of the twelve usage categories we compute each version's share of the category's tokens, $s_{ict}$, and estimate two objects: the within-version decay of $\ln s_{ict}$ with age, as in Fact~1, and the capability coefficient $\theta_c$ from
\begin{equation}
\ln s_{ict}-\ln s_{0ct}=\theta_c\,k_i-\lambda\,\ln p_{it}+\beta' x_{it}+\xi_{f(i)}+\xi_{ct}+\xi_{ict},
\label{eq:estc}
\end{equation}
with category-by-week fixed effects, a common price coefficient instrumented as above, and the outside option $s_{0ct}$ defined within the category as the share of its tokens not carried by paid matched versions. A single instrument cannot identify twelve price coefficients, so $\lambda$ is common; the question is about $\theta_c$.

Table~\ref{tab:category} reports the results and Figure~\ref{fig:category} plots them. Depreciation differs across uses by a factor of four. Within programming, a frontier version's share of category tokens decays at 0.081 log points per week (standard error 0.009), a half-life of nine weeks; within technology and science at about 0.048; within finance, legal, marketing and academia at 0.027 to 0.037; within health and trivia at under 0.02; and within translation not at all. The hypothesis of a common decay rate is rejected (Wald statistic 103 on eleven degrees of freedom). The category share and the platform share have the same kind of denominator, the category's or the platform's total tokens including free tiers, unmatched models and the long tail, so the platform row of the table, the same versions' decay in platform share, is directly comparable: 0.063, between programming and the rest, as the token weights imply. The capability coefficient is ordered the same way. Our main estimates use the archived weeks, in which every model in a category is itemized: $\theta_c$ is 3.0 (0.5) in programming, between 1.5 and 2.3 in the ten other categories, and 0.5 (0.5) in trivia; the common price coefficient, 0.45 (0.17), is identified within categories and is not the platform-level $\lambda$. On all weeks the estimates are uniformly lower, 1.9 in programming down to zero in trivia, with the same ordering. Appendix~\ref{app:robust} shows that the difference comes from two sources: the selection of product-weeks into the data-API category tables, which itemize only the fifty largest models, and a genuine decline in the platform-level $\theta$ after May 2025 (3.19 before, 2.56 after, with $\lambda$ unchanged), which coincides with a concentration of traffic on a few frontier versions that differ little in capability. The archived-week estimates are free of the first source and are the ones we use; the ordering across uses is the same in both periods.

The third panel of the figure is the full test of the Red Queen formula across uses. Corollary~\ref{cor:rq} says that a fixed version's share in use $c$ decays at $\theta_c(1-\bar s_c)g_{A,c}$, where $\bar s_c$ is the mean category share of frontier versions and $g_{A,c}$ the growth of the token-weighted mean capability of paid versions in the category (0.84 to 1.20 standard deviations a year). With $\theta_c$ from the archived category tables, the prediction has no free parameter. A regression of measured on predicted decay through the origin has a slope of 0.96 (95 percent interval 0.77 to 1.10), and the median ratio of measured to predicted decay across categories is 1.02. The test also holds when the decay is measured out of sample: $\theta_c$ comes from the tables that end in May 2025, and when decay is measured only on frontier versions that enter after that date, the slope is 1.03 (0.63 to 1.38) and the correlation across the twelve uses 0.74. The clear exception is translation, where the formula predicts a decay of 0.032 a week and the data show none; translation, and to a lesser degree role-play and health, are uses in which adoption continues for most of a version's first year, so the age window averages a long ramp with the subsequent decay and the measured rate is a lower bound (Appendix~\ref{app:decomp}). Programming has both the most capability-sensitive demand and the fastest depreciation, and because it is half of the platform's traffic, the platform-wide rate of Fact~1 is close to programming's. A nested logit with an upper nest across uses would add a parameter for substitution between uses that these data cannot identify; what they do identify is that the within-use capability-sensitivity that drives depreciation is heterogeneous in the way the model predicts, and by the amount it predicts.

\begin{table}[htbp]
\centering
\caption{Depreciation and the capability coefficient by use}\label{tab:category}
\scriptsize\setlength{\tabcolsep}{3pt}
\resizebox{\textwidth}{!}{%
\begin{tabular}{lrcccccrrr}
\toprule
 & Token & Decay per & Half-life & $\theta_c$ (IV), & $\theta_c$ (IV), & $g_{A,c}$ & Predicted & \multicolumn{2}{c}{Measured/predicted}\\
Category & share (\%) & week & (weeks) & archived & all weeks & per year & decay & age window & peak window\\
\midrule
Programming & 50.7 & -0.081 (0.009) & 9 & 3.03 (0.46) & 1.93 (0.27) & 1.04 & 0.059 & 1.38 & 1.70\\
Technology & 21.7 & -0.049 (0.007) & 14 & 2.06 (0.44) & 1.02 (0.26) & 1.08 & 0.042 & 1.17 & 1.70\\
Science & 10.3 & -0.047 (0.006) & 15 & 1.73 (0.48) & 0.69 (0.24) & 1.20 & 0.039 & 1.20 & 1.87\\
Roleplay & 9.7 & -0.024 (0.009) & 29 & 2.23 (0.47) & 0.62 (0.29) & 0.92 & 0.038 & 0.62 & 1.57\\
Finance & 2.5 & -0.037 (0.006) & 19 & 2.20 (0.39) & 0.65 (0.24) & 1.09 & 0.045 & 0.82 & 1.00\\
Academia & 1.2 & -0.027 (0.006) & 25 & 2.07 (0.43) & 0.63 (0.25) & 0.93 & 0.036 & 0.76 & 1.19\\
Legal & 0.9 & -0.031 (0.006) & 22 & 1.51 (0.42) & 0.39 (0.26) & 0.95 & 0.027 & 1.15 & 1.72\\
Translation & 0.9 & 0.000 (0.008) & -- & 1.83 (0.46) & 0.22 (0.24) & 0.94 & 0.032 & -0.00 & 0.35\\
Health & 0.8 & -0.018 (0.005) & 38 & 1.59 (0.43) & 0.33 (0.24) & 1.07 & 0.032 & 0.57 & 0.84\\
Marketing & 0.8 & -0.032 (0.007) & 22 & 1.67 (0.43) & 0.31 (0.26) & 0.84 & 0.026 & 1.21 & 1.87\\
Marketing/SEO & 0.3 & -0.027 (0.009) & 25 & 1.82 (0.46) & 0.30 (0.26) & 0.91 & 0.031 & 0.88 & 1.19\\
Trivia & 0.1 & -0.019 (0.003) & 37 & 0.54 (0.46) & 0.03 (0.25) & 1.08 & 0.011 & 1.72 & 2.80\\
\midrule
Platform shares, same versions & 100 & -0.063 (0.006) & 11 & & 2.42 (0.32) & & & \\
Pooled category shares & 100 & & & & 0.50 (0.25) & & & \\
\bottomrule

\end{tabular}}
\begin{minipage}{0.95\textwidth}\footnotesize\emph{Notes:} Share of tokens: the category's share of classified tokens over the data-API weeks. Decay: within-version coefficient of log category share on age, frontier versions (78 to 100 per category), ages 4 to 56, standard errors clustered by version; the platform row uses the same 103 versions' shares of platform tokens. $\theta_c$: capability coefficient per standard deviation of item-response capability from \eqref{eq:estc} with developer and category-by-week fixed effects, the controls of Table~\ref{tab:demand}, a common log-price coefficient instrumented with the window instrument, standard errors clustered by version. Archived: April 2024 to May 2025, every model itemized (10,655 observations, first-stage F 26); all weeks: to October 2026 (45,450 observations, F 123). The common within-category log-price coefficient is 0.45 (0.17) on the archived weeks and 0.11 (0.14) on all weeks; it is identified from price variation within category-week cells and is not the platform-level $\lambda$ of Table~\ref{tab:demand}. Restricting the data-API weeks to versions itemized in every week of their presence (25,937 observations) gives $\theta_c$ within 0.2 of the all-weeks column in every category. $g_{A,c}$: growth of the token-weighted mean capability of paid versions in the category, standard deviations per year. Predicted decay: $\theta_c(1-\bar s_c)g_{A,c}$ per week with the archived $\theta_c$ and $\bar s_c$ the mean category share of frontier versions; the last two columns divide the measured decay by it, measured over ages 4 to 56 (age window) and over the 52 weeks after the version's peak category share within its first 30 weeks (peak window). The platform row reports the IV specification of Table~\ref{tab:demand} on the product-weeks that appear in the category panel; the pooled row imposes a common $\theta$ on category shares.\end{minipage}
\end{table}

\begin{figure}[htbp]
\centering
\includegraphics[width=\textwidth]{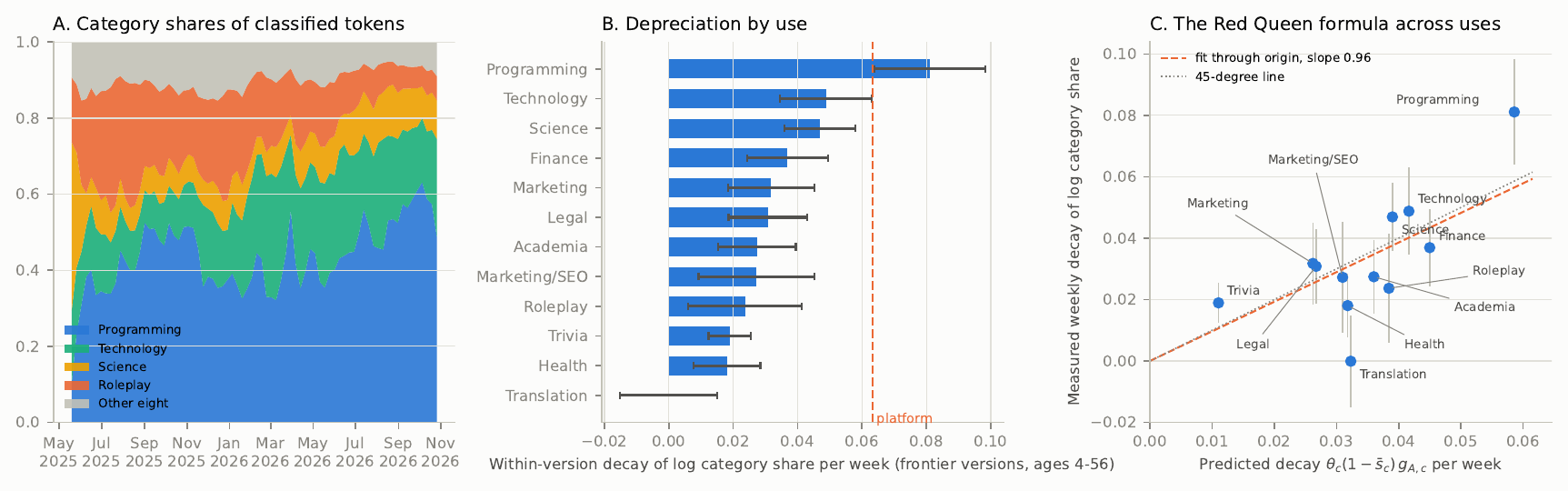}
\caption{Usage, depreciation, and the capability coefficient by use}\label{fig:category}
\begin{minipage}{0.95\textwidth}\footnotesize\emph{Notes:} (A) Shares of classified platform tokens by category, data-API weeks. (B) Within-version weekly decay of log category share, frontier versions, with 95 percent intervals; the dashed line is the same versions' decay in platform share. (C) The measured decay against the Red Queen prediction $\theta_c(1-\bar s_c)g_{A,c}$ with no free parameter; the dashed line is the fit through the origin and the dotted line the 45-degree line.\end{minipage}
\end{figure}

\begin{table}[htbp]
\centering
\caption{Demand estimates (item-response capability, archived weeks)}\label{tab:demand}
\footnotesize
\resizebox{\textwidth}{!}{%
\begin{tabular}{lccccc}
\toprule
 & (1) & (2) & (3) & (4) & (5)\\
 & OLS & IV: window & IV: window, & IV: window, capability & Nested logit,\\
 & & & developer$\times$quarter & instrumented by & developer nests\\
 & & & effects & chained index & \\
\midrule
Capability, $\theta$ & 2.69 & 2.74 & 2.88 & 2.78 & 3.32\\
 & (0.30) & (0.37) & (0.38) & (0.41) & (0.56)\\
Log price, $-\lambda$ & $-$0.66 & $-$0.69 & $-$0.78 & $-$0.70 & $-$0.94\\
 & (0.10) & (0.19) & (0.18) & (0.19) & (0.22)\\
Log within-developer share, $\sigma$ &  &  &  &  & 0.21\\
 &  &  &  &  & (0.09)\\
Log output speed & 0.17 & 0.17 & 0.18 & 0.17 & 0.13\\
 & (0.07) & (0.08) & (0.08) & (0.08) & (0.08)\\
Log context window & 0.22 & 0.21 & 0.17 & 0.20 & $-$0.03\\
 & (0.11) & (0.11) & (0.11) & (0.12) & (0.10)\\
Reasoning model & $-$1.95 & $-$1.96 & $-$2.13 & $-$1.98 & $-$0.84\\
 & (0.37) & (0.37) & (0.41) & (0.38) & (0.43)\\
Open weights & 0.56 & 0.53 & 0.48 & 0.53 & $-$0.64\\
 & (0.43) & (0.42) & (0.41) & (0.42) & (0.42)\\
\midrule
First-stage $F$: log price &  & 102.8 & 102.1 & 80.5 & 18.8\\
First-stage $F$: capability / within-developer share &  &  &  & 613.5 & 14.6\\
Fixed effects & dev., week & dev., week & dev.$\times$qtr, week & dev., week & dev., week\\
Observations & 11,004 & 11,004 & 11,004 & 11,004 & 5,491\\
Versions & 261 & 261 & 261 & 261 & 163\\
\bottomrule

\end{tabular}}
\begin{minipage}{0.95\textwidth}\footnotesize\emph{Notes:} Dependent variable $\ln s_{it}-\ln s_{0t}$, weekly markets, February 2024 to May 2026 (archived weeks); capability is the item-response scale in standard deviations across priced models. Standard errors clustered by version in parentheses. Columns (2) to (5) instrument the log price with the window instrument; Section~\ref{sec:iv} explains why it measures a conditional cross-sectional gradient. Column (4) also instruments item-response capability with the chain-linked composite index, whose errors come from benchmark revisions rather than from the item-response fit. Column (5) is \eqref{eq:nested} of Section~\ref{sec:decres}: the log within-developer share is instrumented with the log number of the developer's versions on sale and the mean and maximum capability of its other versions, and the sample is the version-weeks of developers with at least two versions on sale; its second first-stage $F$ is for the nest instruments. First-stage $F$ statistics are clustered Wald statistics on the excluded instruments divided by their number. Appendix Table~\ref{tab:demandweak} reports the specifications with the weak alternative instruments.\end{minipage}
\end{table}

\section{Decomposing Depreciation}\label{sec:decomp}

This section turns Fact~1 into a measurement. Under the demand system, every change in an incumbent version's share is the mirror image of changes in the appeal of its alternatives, and the logit structure makes the accounting exact. We use it to decompose each version's first-year decline into the share gained by newer rival versions, the share gained by the firm's own newer versions, the price changes of older versions, movements of the outside option, and the version's own residual loss of appeal. The one assumption the accounting needs, a group of products whose appeal is constant, is examined across its full range; we then estimate the speed at which developers migrate after a release, and the response to individual releases.

\subsection{Exact accounting under the demand system}\label{sec:terms}

Write $s_{mt}=e^{\delta_{mt}}/D_t$ with $D_t=\sum_m e^{\delta_{mt}}$ the inclusive value over all products including the outside option. For any version $i$,
\begin{equation}
\Delta\ln s_{it}=\Delta\delta_{it}-\Delta\ln D_t,
\label{eq:acct}
\end{equation}
and $D_t/D_{t-1}-1$ decomposes exactly into $\sum_{m\in\mathrm{cont}}s_{m,t-1}(e^{\Delta\delta_{mt}}-1)$ over continuing products, $\sum_{m\in\mathrm{new}}s_{mt}D_t/D_{t-1}$ over entrants, and the outside option's term. The only unobserved object is $\Delta\ln D_t$. We identify it from an anchor group of versions whose $\delta$ is constant by the model's assumptions: versions at least 26 weeks old whose list price did not change in $t$ (at least 16 weeks old in the first months of the sample, when few older versions exist). Their aggregate share moves only through $D_t$, so $\Delta\ln D_t=-\Delta\ln S^{\mathrm{anchor}}_t$. Given $\Delta\ln D_t$, every continuing product's $\Delta\delta_{mt}$ is observed from its share change, and the price part of it is $-\lambda\Delta\ln p_{mt}$ with $\lambda$ from Section~\ref{sec:demand}. Each contribution is rescaled to log units so that the pieces sum to $\Delta\ln D_t$.

For incumbent $i$ we group the contributions to $-\Delta\ln D_t$ by who gains: versions released after $i$ by other firms (entrants in $t$ and younger versions still ramping up), versions released after $i$ by $i$'s own firm, older versions' price changes, older versions' other shifts, and the outside option. The version's own $\Delta\delta_{it}$ is split into its price change and an own residual, which is positive during its launch ramp and negative if the version keeps losing appeal relative to the anchor group. The parameter $\theta$ does not enter the accounting: the share that newer rivals take is observed. The demand estimate of $\theta$ is what explains \emph{why} they take it, through their capability, and it is used in Section~\ref{sec:dynamics}.

\subsection{Results}\label{sec:decres}

Figure~\ref{fig:stack} stacks the cumulative components by age for frontier versions, with the direct cumulative decline in log share since age four as the black line, and Table~\ref{tab:decomp} reports the first-year totals. For the median frontier version observed at age 52, the token share has fallen by 2.36 log points, 91 percent, since age four. These 28 versions are the ones still sold a year after entry; versions withdrawn earlier fall faster, and the within-version age slope of all 54 frontier versions in the accounting, over ages 4 to 56, is 0.066 log points a week, 3.4 a year, close to Fact~1's. Newer rival versions account for 1.99 log points of the decline, four fifths of it on medians and on means alike (84 percent). The firm's own newer versions account for 0.23 log points, 10 to 14 percent; a nested logit that lets a developer's versions substitute more closely for one another raises this to about a third, with rivals still at 79 percent (below). The price changes of older versions account for essentially nothing, consistent with Fact~3, and older versions' other shifts for 1 to 2 percent. The outside option takes 6 to 9 percent. The version's own residual is 1 percent of the mean decline and indistinguishable from zero: once the inclusive value is anchored on versions older than 26 weeks, an incumbent's appeal relative to the anchor group barely drifts after its launch ramp, and its decline is accounted for by what the alternatives gain. Measured directly for the same versions, the market's paid revenue grows by 2.53 log points, about twelvefold, between their fourth and fifty-second weeks, and the median version's revenue \emph{level} falls by 0.09 log points, or 8 percent. The rental value of an asset whose weights have not changed falls by 91 to 96 percent relative to the market in a year and is roughly unchanged in dollars on a platform whose revenue grows twelvefold; measured against the slower revenue growth of the largest developers, it falls by 67 to 78 percent (Section~\ref{sec:payback}). Defining the frontier instead as the top decile among the week's entrants gives 34 versions observed at age 52, whose share falls by 3.29 log points, two thirds of it to newer rivals, and whose revenue level falls by 39 percent (the shares are in Appendix Table~\ref{tab:robust}): which versions survive to age 52 matters for the size of the first-year decline, while newer rivals are the main channel under either definition.
\begin{figure}[htbp]
\centering
\includegraphics[width=0.9\textwidth]{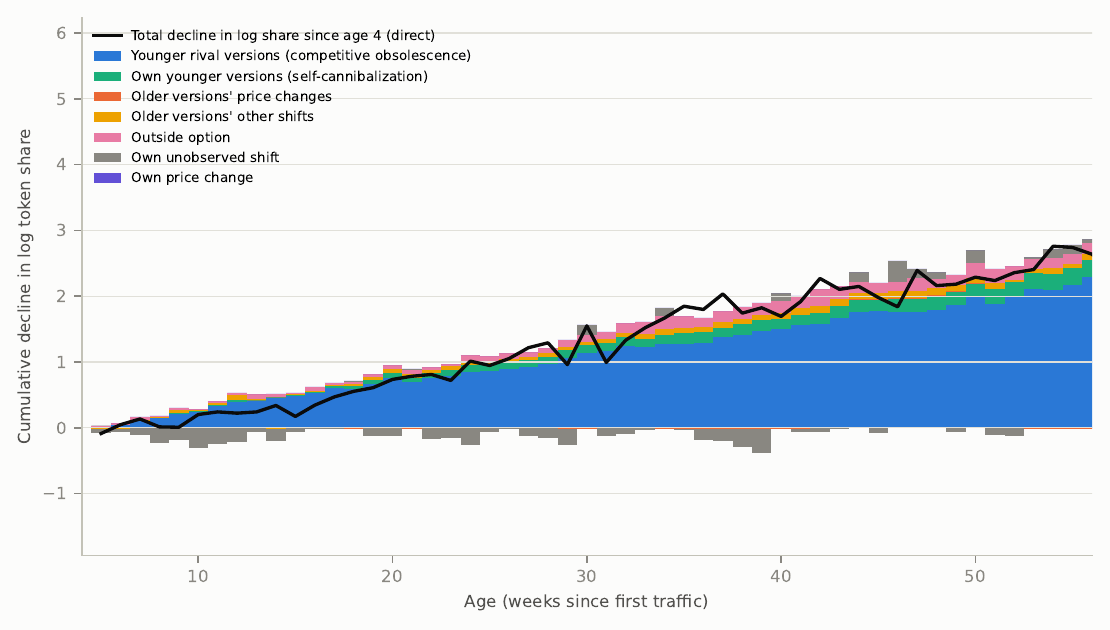}
\caption{Decomposition of the age profile, frontier versions}\label{fig:stack}
\begin{minipage}{0.9\textwidth}\footnotesize\emph{Notes:} Median cumulative contributions since age four, by age, over frontier versions (capability at entry in the top decile of versions available in the entry week, as in Fact~1) entering after March 2024 and observed for at least 30 weeks; positive values are contributions to the decline. Anchor: versions at least 26 weeks old with unchanged list price. The black line is the direct median of the cumulative change in log token share, which also includes weeks in which the version is not observed and which the components therefore do not cover.\end{minipage}
\end{figure}

\begin{table}[htbp]
\centering
\caption{First-year decomposition, frontier versions (age 4 to age 52)}\label{tab:decomp}
\scriptsize
\begin{tabular}{lrr}
\toprule
Component & Log points (median) & Share of mean total\\
\midrule
Newer rival versions (competitive obsolescence) & $-1.99$ & 0.84\\
Own newer versions (self-cannibalization) & $-0.23$ & 0.14\\
Older versions' price changes & $0.00$ & 0.00\\
Older versions' other shifts & $-0.03$ & 0.02\\
Outside option & $-0.21$ & 0.06\\
Own price change & $0.00$ & $-0.05$\\
Own residual & $+0.12$ & 0.01\\
Weeks in which the version is not observed & $+0.01$ & $-0.02$\\
\midrule
Total change in log token share & $-2.36$ & 1.00\\
Change in log revenue share (direct) & $-2.45$ & \\
Market paid-revenue growth (direct) & $+2.53$ & \\
Change in the version's log revenue (direct) & $-0.09$ & \\
\bottomrule
\end{tabular}
\begin{minipage}{0.95\textwidth}\footnotesize\emph{Notes:} 28 frontier versions observed at age 52 (frontier: capability at entry in the top decile of versions available in the entry week, as in Fact~1). Medians of cumulative contributions do not add to the median total; shares of the mean total ($-2.41$ log points) add to one. The two calendar weeks missing from the panel (late January and the end of March 2025) are bridged as two-week steps. The not-observed component is the part of the direct total that falls in weeks in which the version has no traffic in week $t$ or in week $t-1$, so that the week pair enters no accounting; it is reported rather than allocated. Direct rows: medians of the change between ages 4 and 52; medians of different quantities do not add up, so the median changes in log revenue share and in market revenue do not sum to the median change in the version's log revenue.\end{minipage}\end{table}

\paragraph{Sampling uncertainty.} The shares are averages over 28 frontier versions, and versions are the natural unit of resampling. A bootstrap over versions (1,000 draws) gives a 95 percent interval of 63 to 112 percent for the share of newer rival versions, 10 to 20 percent for own successors, and $-36$ to $+22$ percent for the own residual: the share of rivals can exceed one because the version's own terms can offset part of the decline, its lower bound is well above one half, and the own residual is not distinguishable from zero.
\paragraph{The anchor.} The accounting's one assumption is the group whose appeal is constant, so Figure~\ref{fig:anchor} reports the first-year shares for every anchor threshold from 12 to 78 weeks. The figure and the bootstrap interval above use the accounting without the fallback to younger anchors in weeks that lack one; at 26 weeks it gives 82 percent rather than the baseline's 84. For thresholds between 12 and 45 weeks the share of newer rivals is between 76 and 90 percent and the own residual lies between $-0.27$ and $+0.10$. At 52 weeks the share of rivals falls to 53 percent and the own residual peaks at 48 percent, before falling again at longer thresholds; with an anchor of a year or more no anchor exists in the panel's first year, so the attribution rests on a selected set of later version-weeks. We take 26 weeks as the baseline because it is the longest threshold at which an anchor exists in nearly every week after mid-2024, and the range 76 to 90 percent is the uncertainty attached to the anchor choice. Using all price-unchanged versions older than eight weeks gives 76 percent (Appendix~\ref{app:robust}). The anchor's appeal could also fall for reasons of its own, since old versions lose hosts and speed and developers announce their retirement. Excluding such versions from the anchor week by week leaves the share of newer rivals at 81 percent and raises the own residual to 16 percent, which bounds the part of the decline that reflects a version's own service; directly, time-varying controls for speed, latency, hosts and price change the within-version decay from 0.069 to 0.066 log points a week. Defining the frontier as the top decile of the week's entrants, which admits 34 versions observed at age 52, gives 68 percent. Appendix Table~\ref{tab:robustsummary} collects these variants.
\begin{figure}[htbp]
\centering
\includegraphics[width=0.8\textwidth]{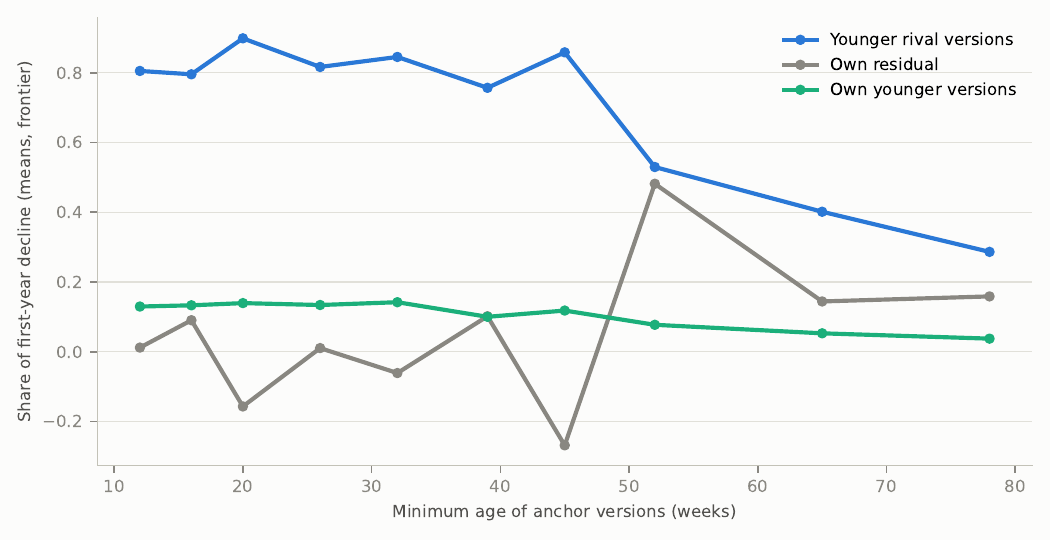}
\caption{Attribution as a function of the anchor threshold}\label{fig:anchor}
\begin{minipage}{0.9\textwidth}\footnotesize\emph{Notes:} Shares of the mean first-year decline in log token share (frontier versions at age 52) attributed to newer rival versions, own newer versions, and the version's own residual, when the inclusive value is anchored on price-unchanged versions at least $a$ weeks old, for $a$ from 12 to 78. The baseline uses $a=26$.\end{minipage}
\end{figure}

\paragraph{Weeks not covered.} The accounting is defined on consecutive week pairs in which the version is observed in both weeks; the two calendar weeks missing from the panel are bridged as two-week steps, which the identity \eqref{eq:acct} allows. For the frontier versions observed from age 4 to 52, 47 of the 48 week pairs are covered on average. In the remainder the version has no traffic in one of the two weeks; the change across such gaps is a small net gain, 2 percent of the first-year decline, against a weekly fall of 0.053 in covered weeks, and we report it separately rather than allocate it. Allocating it in proportion to the covered components would change the share of newer rivals from 84 to 83 percent.

\paragraph{Substitution within developers.} The logit accounting imposes that a developer's new version draws share from the developer's older versions in proportion to their shares, as it does from any other product. Buyers may instead integrate with a developer's interface, tooling and terms rather than with a single version, so that a successor substitutes more closely for its predecessor than a rival does. We relax this with a nested logit \citep{Berry1994} in which each developer's versions form a nest,
\begin{equation}
\ln s_{it}-\ln s_{0t}=\theta\,k_i-\lambda\,\ln p_{it}+\sigma\ln s_{it|f}+\beta'x_{it}+\xi_{f(i)}+\xi_t+\xi_{it},
\label{eq:nested}
\end{equation}
where $s_{it|f}$ is $i$'s share of its developer's tokens and $\sigma\in[0,1)$ is the correlation of tastes within a developer. The within-developer share is instrumented with the log number of the developer's versions on sale and the mean and maximum capability of its other versions, together with the window instrument for price. The estimate is $\sigma=0.21$ (0.09), with a first-stage $F$ of 15 for the within-developer share, and $\theta=3.32$ (0.56) and $\lambda=0.94$ (0.22) on the 5,491 version-weeks for which the nest instruments are defined (Appendix Table~\ref{tab:nested}). Under \eqref{eq:nested} the identity \eqref{eq:acct} becomes $\Delta\ln s_{it}=\Delta\delta_{it}/(1-\sigma)-\sigma\,\Delta\ln D_{f(i)t}-\Delta\ln\check D_t$, where $D_{ft}=\sum_{j\in f}e^{\delta_{jt}/(1-\sigma)}$ is the developer's inclusive value and $\check D_t=1+\sum_hD_{ht}^{1-\sigma}$ the market's; a successor now lowers its predecessor's share through both terms, and each term remains exact, with the anchor normalized as before. At $\sigma=0$ this accounting returns the logit attribution, 87 percent to newer rival versions and 12 percent to own successors on this implementation. At $\hat\sigma=0.21$ newer rival versions account for 79 percent of the first-year decline, own successors for 36 percent and the outside option for 5 percent, and the version's own term turns to $-25$ percent: part of what the logit reads as an incumbent holding its appeal is, in the nested model, the incumbent gaining relative to its developer's other versions while a successor draws on the nest. Across the 95 percent interval for $\sigma$, 0.03 to 0.39, the share of newer rivals is 70 to 86 percent and that of own successors 15 to 62 percent; own successors would match rivals at $\sigma$ of about 0.42, just above the interval. Self-cannibalization is therefore larger than the logit implies, about a third of the decline at the point estimate, but competition from other developers remains the main channel.
\subsection{How fast developers migrate}\label{sec:migrate}

A static logit says that a release moves shares at once. The event paths of Section~\ref{sec:event} and the launch ramps of Fact~1 say that developers migrate over weeks. To measure the speed, we estimate a partial-adjustment version of \eqref{eq:est}, $\ln(s_{it}/s_{0t})=\rho\ln(s_{i,t-1}/s_{0,t-1})-\lambda\ln p_{it}+\beta'x_{it}+\xi_i+\xi_t+u_{it}$, in which the version effects absorb capability, with version and week fixed effects, the lagged share instrumented by its second lag \citep{AndersonHsiao1981} and the price by the window instrument of Section~\ref{sec:iv}. The adjustment parameter is $\rho=0.88$ (standard error 0.02), a half-life of adjustment of five and a half weeks. The same parameter governs the response to a version's own price, which after $k$ weeks is the fraction $1-\rho^{k+1}$ of the full response; Section~\ref{sec:lambdas} uses this to recover the full price coefficient from list-price changes. A third of the response of an incumbent's share to a change in the alternatives therefore arrives within three weeks and two thirds within about two months, which is the ramp visible in the event paths; it is also why the accounting credits the gains of younger rival versions \emph{as they ramp}, not only in their entry week.

Migration takes weeks rather than days because switching is costly: a buyer must validate a new version on its own workload and adapt its integration before the per-token savings or capability gains pay off. The budget for switching scales with a buyer's volume, so high-volume buyers recover the cost of integration within weeks and low-volume buyers may never do so. That is consistent with smooth adoption ramps, with the fastest depreciation in programming, where the largest buyers are coding agents processing trillions of tokens a week, and with slower depreciation among buyers whose switching costs the marketplace does not see.

\subsection{Depreciation inside applications}\label{sec:apps}

The platform-wide decline of a version's share could come from two sources: each application replacing the version with newer ones, or applications that rely on the version losing traffic to applications that use newer models. Only the first is depreciation in the sense of the model; the second would make the decline a by-product of turnover among downstream products. The marketplace's page for each application shows, for the preceding thirty days, the application's daily tokens by its ten most-used models, with the rest pooled \citep{OpenRouterApps2026}. Internet Archive captures of these pages, one every four weeks or more for each application with at least three captures, give a panel of 277 applications and 5,182 application-weeks; the captures run from May 2025 to November 2026 and their 30-day windows reach back to April 2025; 89 percent of the tokens attributed to named models map to panel versions. The data have one limitation that matters: a version's in-application share is observed only while it is among the application's top ten, so it drops out of view as it declines, and estimated decay rates inside applications are lower bounds.

Three results follow, for frontier versions at ages 4 to 52 weeks. First, versions lose share inside applications about as fast as on the platform. With application-by-version fixed effects, the in-application log share falls by 0.056 per week (standard error 0.005; 3,241 application-versions); setting the share in weeks after a version leaves an application's top ten to the smallest share displayed that week gives 0.051 (0.004); the platform share of the same versions over the same weeks falls by 0.074 (0.007), and the paper's within-version estimate is 0.070. Figure~\ref{fig:apps}A shows the two profiles. Second, the decline happens inside applications. Splitting the weekly change in a version's share of the sampled applications' tokens into the change in its share within each application, weighted by application size, and the change in application sizes, weighted by the version's shares, the within-application term accounts for 116 percent of the first-half-year decline of 77 frontier versions and the between term for $-16$ percent: applications that relied on older versions grew slightly faster than the rest. Neither result depends on which applications the archive happened to capture: in the largest quarter of applications by tokens the in-application decay is 0.066 (0.007) and the within term 133 percent; in applications with at least six captures, 0.056 (0.006) and 113 percent; weighting applications by their tokens gives 0.120 (0.014), dominated by the largest coding agents. Third, the speed differs by use as the category panel says it should. Table~\ref{tab:apps} reports the in-application decay by the application's category: on the observed shares, coding agents and IDE extensions replace versions fastest and role-play and creative-writing applications among the slowest, the ordering of Section~\ref{sec:byuse}; the imputed shares and the per-application medians compress the differences, and role-play remains the slowest on both.

The application panel also tests the substitution pattern the accounting assumes. Within an application, logit demand implies that an entering version lowers every incumbent's log share by the same amount. With application-by-week effects, which absorb the common response, incumbents from the entrant's own developer lose 0.36 log points more per unit of the entrant's share than incumbents from other developers (standard error 0.28, 2,319 application-weeks with entry), and incumbents within half a standard deviation of the entrant's capability lose no more than those farther away ($-0.08$, 0.29). Substitution inside applications is close to proportional; to the extent that it is nested by developer, the accounting understates the share taken by own successors and overstates that taken by rivals, by an amount the imprecision of the contrast does not allow us to pin down.

Adoption is fast. A new frontier version enters an application's top ten a median of one week after it enters the platform (zero weeks in coding applications), 62 percent do so within two weeks, and its share within an application peaks after about four weeks (three in coding applications, six in role-play, seven in gateways that route on behalf of others). This is the launch ramp of Fact~1 seen from inside the buyer, and it is consistent with the partial-adjustment estimate above. The application panel does not identify the capability coefficient inside applications: because only the top ten versions are visible, a regression of in-application shares on capability and price with application-by-week fixed effects gives $\theta=0.72$ (0.12), attenuated by selection in the same way as the category panel in its data-API weeks, so we do not use it to test the Red Queen formula in levels. Its ranking across groups agrees with the decay ranking, highest in coding (1.05) and lowest in role-play (0.18, not distinguishable from zero).

\begin{table}[htbp]\centering\small
\caption{Depreciation of frontier versions inside applications}\label{tab:apps}
\resizebox{\textwidth}{!}{\begin{tabular}{lccccc}
\toprule
Application group & Apps & App-versions & Decay (observed) & Decay (imputed) & Median app slope\\
\midrule
Coding (CLI, IDE, cloud agents) & 30 & 670 & $-0.098$ (0.008) & $-0.049$ (0.007) & $-0.063$\\
General chat & 13 & 285 & $-0.067$ (0.010) & $-0.046$ (0.006) & $-0.046$\\
Personal and other agents & 15 & 187 & $-0.054$ (0.013) & $-0.053$ (0.007) & $-0.093$\\
Gateways & 8 & 201 & $-0.031$ (0.009) & $-0.039$ (0.007) & $-0.044$\\
Role-play and creative writing & 28 & 328 & $-0.042$ (0.006) & $-0.035$ (0.005) & $-0.041$\\
All applications & 277 & 3,241 & $-0.056$ (0.005) & $-0.051$ (0.004) & $-0.067$\\
\midrule
Platform, same versions and weeks & & 102 versions & $-0.074$ (0.007) & & \\
\bottomrule
\end{tabular}}
\begin{minipage}{0.95\textwidth}\footnotesize\emph{Notes:} Weekly slope of the log share of a frontier version in an application's tokens on age (weeks 4 to 52), with application-by-version fixed effects; standard errors clustered by version. Observed: weeks in which the version is among the application's ten most-used models. Imputed: adds the application's later observed weeks with the share set to the smallest displayed share that week. Median app slope: median over applications with at least five frontier versions (227 applications). Groups follow the marketplace's category tags; applications without tags are included in the total only. Source: Internet Archive captures of openrouter.ai application pages, April 2025 to November 2026.\end{minipage}
\end{table}

\begin{figure}[htbp]
\centering
\includegraphics[width=0.95\textwidth]{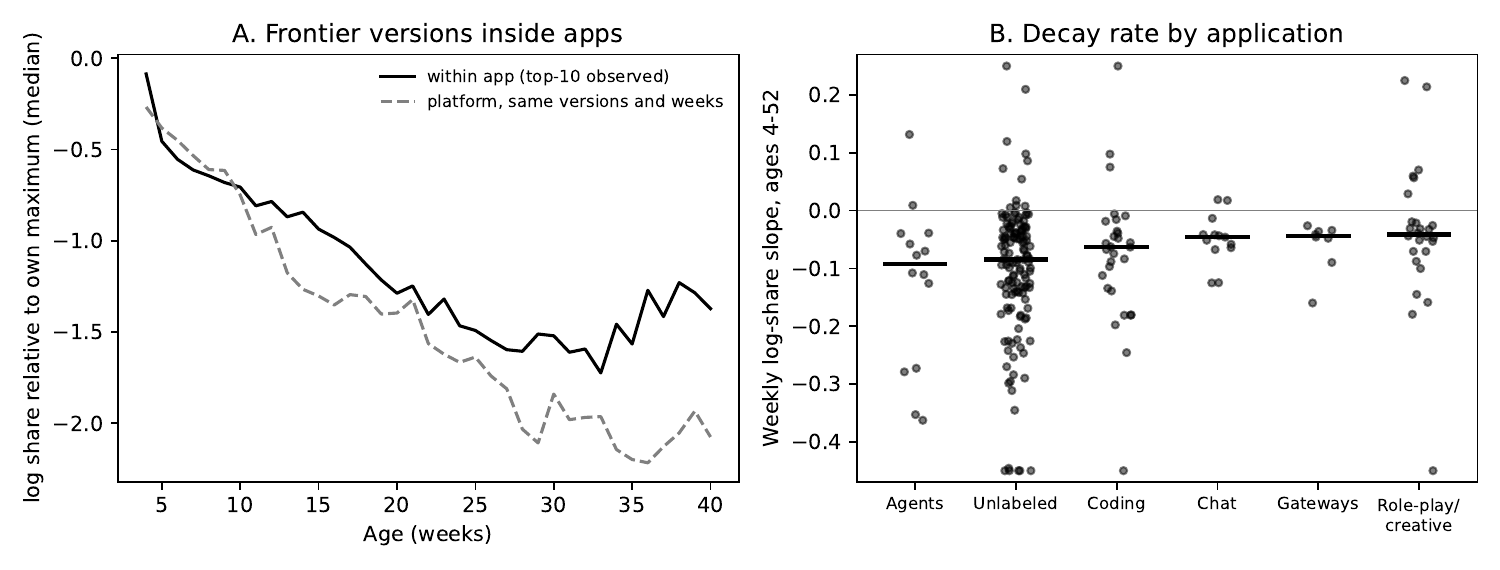}
\caption{Depreciation inside applications}\label{fig:apps}
\begin{minipage}{0.92\textwidth}\footnotesize\emph{Notes:} (A) Median log share of frontier versions relative to their own maximum, by age, inside applications (while among the application's top ten) and on the platform over the same version-weeks. (B) Weekly log-share slope of frontier versions at ages 4 to 52 for each application with at least five frontier versions, by application group; horizontal bars are medians; slopes are clipped at $-0.45$ and $0.25$ for display.\end{minipage}
\end{figure}

\subsection{Age profiles and individual releases}\label{sec:apcshort}

Appendix~\ref{app:decomp} reports two complements. The accounting resolves the collinearity of age, cohort and calendar time that makes an age profile with version and week effects unidentified: the cumulative accounting terms lower the within-version age coefficient from $-0.061$ to $-0.035$ log points a week. And event paths around individual releases show an incumbent's share falling by about 0.8 log points over the twelve weeks after a comparable rival release, with own-successor releases taking somewhat more; releases overlap so densely that the event design cannot separate entrants. The attribution therefore rests on the accounting, not on the event paths.

\section{Depreciation and the Measurement of AI Capital}\label{sec:measure}

This section turns the estimates into the objects capital measurement uses: a depreciation rate for each asset that model development creates, its comparison with the rates and lives in use, and what it implies for net investment and payback.

\subsection{Which rate, for which asset}\label{sec:whichrate}

National accounts deduct depreciation at the rate at which an asset's value declines with age at a given date and record changes in the price of new assets as revaluation \citep{HultenWykoff1981,Hulten1992}. All versions alive at a date sell into the same market, so the ratio of their rents is the ratio of their revenue shares, and market growth raises every vintage's rent alike. A version's depreciation rate is therefore the decline of its share: $\delta_v=3.2$ log points a year, 96 percent, a central value in the range of 2.4 to 3.6 that the measures below span; the median frontier path of Fact~1 falls by 2.6 log points between ages 4 and 52, and by more over the following weeks. The range is bounded as follows: the 28 frontier versions of the accounting that are still sold at age 52 lose 2.4 to 2.5 log points of token and revenue share, a lower bound because versions withdrawn earlier are excluded, and the within-version age slope implies 3.4 to 3.6. The Hulten--Wykoff profile is strictly the ratio of rents across vintages at a date; in the cross-section, with week effects, the log revenue share of frontier versions falls by 0.057 log points per week of age (0.038 to 0.044 for all versions), 2.0 to 2.9 a year. Under steady capability growth the two coincide; the cross-section is flatter because only surviving versions are observed at older ages, and it sits in the lower part of the range or below it.

A version is not the only asset model development produces. Research and training also produce capability the developer carries into its next versions: data, methods, infrastructure and pretrained bases. That asset loses value only as rivals erode the developer's position, since the developer advances with them. After a developer's own release puts it ahead of the other leading developers, its lead closes at 0.5 to 0.7 log points a year, while the revenue it gains persists through its later versions (Appendix~\ref{app:capshocks}). Developer capability capital therefore depreciates at a fifth to a sixth of the version rate, and in revenue possibly more slowly still.

Table~\ref{tab:lives} sets both rates against those in use. The BEA depreciates business R\&D at 16 to 22 percent a year in the industries where developers would be classified and own-account software at 33 percent \citep{BEA2024}; the largest firms depreciate servers and accelerators over five to six years \citep{Microsoft10K2023,Alphabet10Q2023,Meta10K2025,CoreWeave10K2025}, and Amazon shortened part of its fleet's life in 2025, citing the pace of progress in AI \citep{Amazon10K2025}. Versions depreciate ten times as fast as the servers that train and serve them; developer capability depreciates two to four times as fast as R\&D in the official accounts. In the user cost of capital \citep{Jorgenson1963}, the rental per dollar of a version's value is $r+\delta_v$, 3.3 a year, so a version must earn its value within about four months; the cheapening of new capability is not added, because it is what lowers the version's share and is already in $\delta_v$. Proposals for AI satellite accounts that capitalize training as model capital \citep{KorinekMcKelvey2026} need this rate; official rates imply a user cost an order of magnitude too low.

\begin{table}[htbp]\centering\small
\caption{Depreciation rates of model capital, official rates and hardware lives}\label{tab:lives}
\resizebox{\textwidth}{!}{\begin{tabular}{lccc}
\toprule
Asset and source & Life & Rate (log points/yr) & Version rate / this rate\\
\midrule
\emph{Model capital (this paper)} & & & \\
\quad Versions: revenue-share decay & -- & 3.2 & 1.0\\
\quad Versions: developers' service lives (Section~\ref{sec:retire}) & 16 months & 1.24 & 2.6\\
\quad Developer capability: decay of capability leads & -- & 0.5--0.7 & 4.6--6.4\\
\emph{Hardware (company filings)} & & & \\
\quad Servers: Microsoft, Alphabet, \mbox{CoreWeave} & 6 years & 0.28 & 11.6\\
\quad Servers: Meta (2025) & 5.5 years & 0.30 & 10.7\\
\quad Servers: Amazon (2025, part of fleet) & 5 years & 0.33 & 9.7\\
\emph{Official rates (BEA)} & & & \\
\quad R\&D, scientific R\&D services & -- & 0.17 & 18.4\\
\quad R\&D, software publishers & -- & 0.25 & 12.9\\
\quad Custom and own-account software & 5 years & 0.40 & 8.0\\
\quad Prepackaged software & 3 years & 0.80 & 4.0\\
\bottomrule
\end{tabular}}
\begin{minipage}{0.95\textwidth}\footnotesize\emph{Notes:} Lives $T$ are converted to continuous rates $R/T$ with the BEA's declining-balance multiple for equipment, $R=1.65$; annual geometric rates $d$ as $-\ln(1-d)$. Developer capability: decay of a developer's capability lead over the five leading developers after a jump caused by its own release (Appendix~\ref{app:capshocks}); the developer's revenue gain from such a jump shows no measurable decay within a year. Sources: \citet{BEA2024}; company filings cited in the text.\end{minipage}
\end{table}

\subsection{Retirement and service lives}\label{sec:retire}

National accounts infer most depreciation from service lives \citep{BEA2024}. For models the service life is observable: OpenAI, Anthropic, Google and Mistral publish when each model will stop being served \citep{OpenAIDeprecations2026,AnthropicDeprecations2026,GeminiDeprecations2026,MistralModels2026}. Across 141 dated text-model snapshots the median service life is 16 months by the Kaplan--Meier estimator \citep{KaplanMeier1958} (Figure~\ref{fig:lifecycle}), and Microsoft's cloud sets the retirement date of a generally available model version at 18 months after launch, with exceptions for some models \citep{AzureRetirements2026}. The economic life is far shorter. When the final deprecation notice is issued, a median of 11 months after release, the version retains 2.5 percent of its peak share and has earned 88 percent of its lifetime revenue on the platform (medians over the 46 versions with an observed notice and the 34 with a complete post-notice window); among retired versions the median is served for eight more months after its share has fallen below a tenth of its peak, and at shutdown it retains about one percent. Retirement ends a long tail of service after the value has gone, so a schedule built from service lives understates depreciation even when the lives are right: the declining-balance method applied to 16 months gives 1.24 log points a year, because it spreads over the whole life a loss that occurs in its first third.

\begin{figure}[htbp]
\centering
\includegraphics[width=0.95\textwidth]{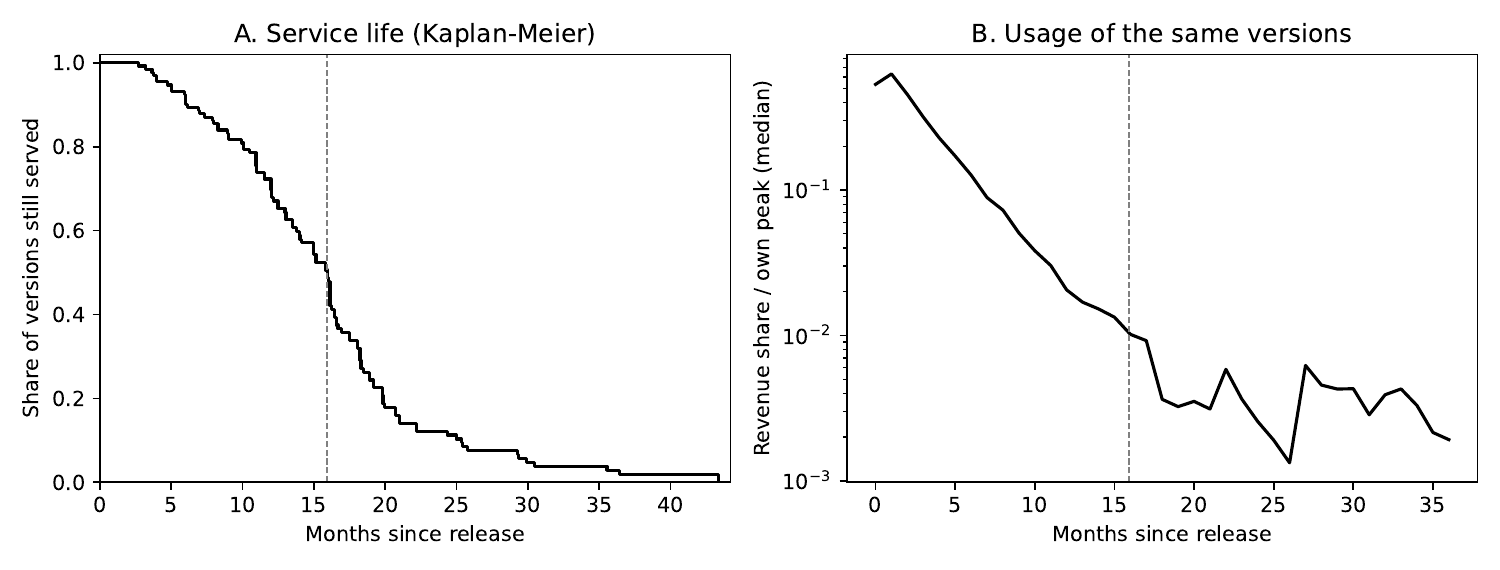}
\caption{Service lives and economic lives of model versions}\label{fig:lifecycle}
\begin{minipage}{0.92\textwidth}\footnotesize\emph{Notes:} (A) Kaplan--Meier estimate of the share of dated API snapshots still served, by months since release; 141 text-model snapshots from OpenAI (median 18 months), Anthropic (16), Mistral (11) and Google (10), with announced shutdowns as events and snapshots still served censored at November 2026. (B) Median revenue share relative to the version's own peak, by months since release, for the 89 panel versions matched to the schedules; it falls to half of its peak after 2.3 months, a tenth after 5.5 and a hundredth after 12. The dashed line marks the median service life. Sources: developers' deprecation pages; OpenRouter.\end{minipage}
\end{figure}

\subsection{Net capital formation}\label{sec:netcap}

With gross investment growing at rate $g$ and depreciation at $\delta$, net investment on a steady path is the fraction $g/(g+\delta)$ of gross investment. For versions, with $g$ between 0.5 and 1.0, the fraction is 14 to 24 percent: three quarters or more of what is spent on versions in a year replaces capability that competitors have made obsolete. If new capability cheapens at the hedonic rate, the constant-price fraction, $(g+\pi)/(g+\pi+\delta_v)$ with $\pi=1.25$, is at most 35 to 41 percent.

For total development spending the fraction depends on its split between versions and carry-forward capability, which no developer reports. Appendix Table~\ref{tab:firmacct} illustrates the range for the two closed-model developers whose finances have been reported \citep{EpochCompanies2026}, treating research compute as gross investment: if all of it builds versions, net investment in 2025 is a fifth to a quarter of gross; if half does, about two fifths. With capability capital depreciating at 0.65 to 1.0 log points, the upper end of the measured 0.5 to 0.7 and above it, revenue net of inference compute fell short of the depreciation of model capital at both firms for every split, before any other cost, unless OpenAI's inference for non-paying users is set aside; at official software rates both show a surplus. Amortizing all model development over software or R\&D lives overstates these two firms' 2025 net investment by up to \$5.8 to \$7.7 billion.

\subsection{Payback}\label{sec:payback}

Let $R_4$ be a version's annualized revenue at age four, after its launch ramp, and $g_M$ the growth of market revenue. At discount rate $r$ the present value of the version's revenue is $M(g_M)\times R_4$, where
\begin{equation}
M(g_M)=\frac{1}{52}\sum_{a\ge0}\rho(a)\,e^{(g_M-r)(a-4)/52},
\label{eq:lifetime}
\end{equation}
where $\rho(a)$ is the median frontier version's share at age $a$ relative to age four. With $r=0.10$, $M$ is 0.40 years if the market does not grow, 0.53 at one log point of growth, 0.65 at 1.5 and 1.20 at the platform's 2.53. The relevant growth is probably below the platform's: the platform's token volume grew no faster than the tokens Google reports processing across its products, but its paid revenue grew at 2.3 to 3.4 log points a year against 1.1 and 1.5 for OpenAI's and Anthropic's total revenue \citep{EpochCompanies2026}, so measured against developers' revenue the median version's revenue level falls by 67 to 78 percent in its first year rather than being roughly unchanged, as it is on the platform. Allocating each closed-model developer's reported revenue across its versions by their platform shares gives a median lifetime revenue of 0.41 years of the age-four run-rate, at the no-growth end of the range.

A version repays its development cost $C$ if $\varpi MR_4\ge C$; at a 60 percent margin this requires an age-four run-rate of 2.6 to 3.1 times $C$ at developers' revenue growth of about 1 to 1.5 log points (1.4 to 4.2 times across all growth assumptions). With frontier training runs costing \$0.1 to \$1 billion \citep{CottierEtAl2024}, a frontier version must reach a run-rate of roughly \$0.3 to \$3 billion within a month of launch to repay its own training run, a large fraction of the largest developer's total annualized revenue in mid-2025 \citep{CNBC2025OpenAI}. The arithmetic is consistent with continued investment because a version is not the unit of investment: developers release several versions from one pretrained base and carry capability forward, so the payback condition binds on the stream of versions, and much of the return accrues to the slowly depreciating capability. The condition is for contestable revenue; enterprise buyers with switching costs migrate more slowly, while developers' own consumer products, which switch their default model at release, do not.

\section{Implications for the Innovation Race}\label{sec:dynamics}

The depreciation rate is also the rate at which the innovation race destroys private value. This section uses the estimates to compare the model's rate of competitive obsolescence with the measured one, and to ask whether the race involves too much or too little investment, holding rivals' strategies fixed; Appendix~\ref{app:race} treats the strategic term and the diffusion of open-weight models.

\paragraph{Red Queen rates.}\label{sec:calib}\label{sec:rq} With $\theta=2.74$ per standard deviation of capability, frontier growth of $g_A=1.03$ standard deviations a year, and the market structure of the last archived quarter (three developers above five percent of paid revenue, an outside-option share $s_0=0.19$ and a symmetric inside share $s=0.27$; Appendix Table~\ref{tab:calib}), Corollary~\ref{cor:rq} gives the rate at which a version whose rivals keep pace with the frontier loses share to all other products, $\theta(1-s_i)g_A$, where $s_i$ is the version's own share, about 1.6 percent for frontier versions in their first year: 2.78 log points a year. The decomposition, which uses neither $\theta$ nor $g_A$, measures 2.56 log points taken by all other products (newer rivals, own successors, older versions and the outside option) in a frontier version's first year, four fifths of them by newer rivals: the formula overstates the measured rate by 8 percent. For a firm, whose inside share is 0.27, the same formula gives 2.07, the rate at which a developer that stopped releasing would lose its position. If every proprietary firm stopped improving, each would still lose share at $\theta s_0g_A=0.53$ log points a year through the open-weight option.

\paragraph{The price coefficient.}\label{sec:lambdas} The instrumented coefficient, $\lambda=0.69$, is a cross-sectional gradient conditional on capability and the controls (Section~\ref{sec:iv}), and it implies an own-price elasticity below one, which the pricing condition \eqref{eq:lerner} rules out. Two other sources bear on $\lambda$, and the model of Section~\ref{sec:env} says how they should relate to it. The first is the response of a version's share to a change in its own list price. With developers migrating at the weekly adjustment coefficient $\rho=0.88$ (Section~\ref{sec:migrate}), the response measured $k$ weeks after a change is the fraction $1-\rho^{k+1}$ of the full response, so the responses of closed-weight versions from the week of the change to sixteen weeks later identify the full coefficient: fitting $\lambda_A(1-\rho^{k+1})$ to them gives $\lambda_A=0.87$ (0.27), with a 95 percent interval of 0.36 to 1.49 once the uncertainty in $\rho$ is included (Appendix Figure~\ref{fig:pricing}). The partial-adjustment regression's own long-run coefficient, $\lambda/(1-\rho)$, is too imprecise to add to this (6.3, standard error 5.2). The second source is prices. A developer that sets a version's launch price once, for buyers who all adopt at that price, sets the margin \eqref{eq:lerner}. At the median observed gross margin of 0.73, this implies $\lambda=1.39$ for a version priced on its own and 1.87 for a developer that takes into account the substitution among its versions at its inside share of 0.27; the hedonic price gradient gives 1.62 to 2.05. The single-version value lies inside the interval from price changes, and the portfolio value lies above it. The nested logit of Section~\ref{sec:decres}, which lets a developer's versions substitute more closely for one another, gives 0.94 (0.22). The sources thus span 0.69 to about 2, with the cross-sectional gradient at the bottom, as attenuation from unobserved quality would predict, and every value in that range gives the comparison below the same sign.

\paragraph{The welfare comparison.}\label{sec:welfare} Proposition~\ref{prop:welfare} gives the ratio of business stealing to unappropriated developer surplus as $\lambda\varpi(1-s_0-s)$, and under static pricing as $(1-s_0-s)/(1-s)$, which depends on market shares alone. Table~\ref{tab:welfare} reports it. Under static pricing the ratio is 0.74: business stealing is three quarters of the surplus developers receive and do not pay for, and the shortfall is the outside option, whose users gain from a capability improvement without any proprietary rival losing profit. Across the price coefficients of the table and observed margins of 0.4 to 0.8 the ratio is 0.15 to 0.89; it would exceed one only for a coefficient above 2.3 at a margin of 0.8, beyond the upper end of every source. Holding rivals' strategies fixed, the race involves too little investment, not too much, and under static pricing the planner values a marginal improvement in capability 17 percent more than the firm does. The ratio rose from 0.64 to 0.79 as the outside option's share fell after May 2025, and the narrowing of the gap between the frontier and the best open-weight version would lower it to 0.51. In a discrete release game calibrated to these estimates (Appendix~\ref{app:race}), the planner releases more often than the equilibrium and a cartel less often, at every release cost at which a pure equilibrium exists and with either the cross-sectional or the portfolio-pricing coefficient, so the strategic term does not reverse the direct comparison.

\begin{table}[htbp]
\centering
\caption{Business stealing relative to unappropriated surplus, $B_f/U_f$}\label{tab:welfare}
\footnotesize
\begin{tabular}{lcccc}
\toprule
 & \multicolumn{3}{c}{Gross margin $\varpi$} & Critical\\
Closure & 0.4 & 0.6 & 0.8 & margin\\
\midrule
Cross-sectional gradient $\lambda=0.69$ (0.19) & 0.15 (0.04) & 0.23 (0.06) & 0.30 (0.08) & 2.66\\
Nested logit $\lambda=0.94$ (0.22) & 0.20 (0.05) & 0.30 (0.07) & 0.41 (0.09) & 1.97\\
Own list-price changes $\lambda_A=0.87$ (0.27) & 0.19 (0.06) & 0.28 (0.09) & 0.38 (0.12) & 2.12\\
\quad upper end of 95\% interval, $\lambda_A=1.49$ & 0.32 & 0.48 & 0.64 & 1.24\\
Lerner, single version $\lambda=1.39$ & 0.30 & 0.45 & 0.60 & 1.34\\
Lerner, developer portfolio $\lambda=1.87$ & 0.40 & 0.61 & 0.81 & 0.99\\
Hedonic $\lambda=1.62$ / 2.05 & 0.35 / 0.44 & 0.53 / 0.66 & 0.70 / 0.89 & 1.14 / 0.90\\
\midrule
Static pricing ($\varpi$ from \eqref{eq:lerner}) & \multicolumn{3}{c}{0.74} & --\\
\quad to May 2025 ($s_0=0.27$) / after ($s_0=0.15$) & \multicolumn{3}{c}{0.64 / 0.79} & \\
\bottomrule
\end{tabular}
\begin{minipage}{0.92\textwidth}\footnotesize\emph{Notes:} $B_f/U_f=\lambda\varpi(1-s_0-s)$ with $N=3$, $s_0=0.19$, $s=0.27$ (last archived quarter); delta-method standard errors in parentheses. The critical margin equates the two terms. Cross-sectional gradient: Table~\ref{tab:demand}; nested logit: Appendix Table~\ref{tab:nested}; own list-price changes: the adjustment-path fit of Section~\ref{sec:lambdas}, with the interval including the uncertainty in $\rho$; Lerner: \eqref{eq:lerner} at the median observed margin of 0.73 with the version's share (0.016) or the developer's inside share (0.27). Static pricing sets $\varpi=1/(\lambda(1-s))$, at which the ratio is $(1-s_0-s)/(1-s)$ for any $\lambda$; the portfolio Lerner coefficient gives 0.74 at the median margin by construction.\end{minipage}
\end{table}

\section{Conclusion}\label{sec:conclusion}

A trained model is an asset that does not wear out and whose rental value relative to the market nonetheless disappears within months. The median frontier version loses half of its revenue share within twelve weeks of its launch ramp and 96 percent of it in its first year. Newer rival versions take \rivalshare\ of the lost share, the firm's own successors a tenth to a third, and price competition almost nothing, and the decline happens inside each application that uses the version, at about the platform rate. In the demand system we derive, this is the property of an asset whose value is relative: a fixed model's share decays at the capability-sensitivity of demand times the capability growth of its rivals. The formula holds on average across uses, in sample and out of sample; the deterioration of old versions' service does not account for the decline, and substitution inside applications is close to proportional; and the model-implied rate of competitive obsolescence is within a tenth of the rate the decomposition measures without using either of its inputs.

For capital measurement the paper supplies a concept and two numbers. The concept is that model depreciation is measured on revenue shares, because market growth is a revaluation common to every vintage, and that model development creates two assets: the version, whose capability is frozen, and the capability the developer carries into later versions. Versions depreciate at about 96 percent a year (2.4 to 3.6 log points); developers' carry-forward capability depreciates at 0.5 to 0.7 log points, the rate at which a developer's capability lead closes as rivals catch up, while the revenue it earns through later versions persists. Developers' own retirement schedules show that versions are withdrawn long after their value has gone: retirement is announced after a median of eleven months, when a version retains a fortieth of its peak share, and takes effect after sixteen, so schedules built from service lives, or from the five-to-six-year lives of the servers that train and serve models, understate the depreciation of versions several-fold. Net investment in versions is a seventh to a quarter of gross investment; for total development spending it depends on the undisclosed split between the two assets. The growth-accounting contribution of AI should be computed from net stocks that separate versions, carry-forward capability and hardware.

For the innovation race, the private gain from stealing rivals' business and the social gain the developer does not capture are of the same order, and the second is larger: under static pricing their ratio is 0.74 and depends only on market shares, and the price coefficients from the cross-section, from list-price changes, from margins and from the hedonic gradient all give the same sign. The relevant question about the market for frontier models is therefore not only whether it is concentrated at a point in time, which it is, but how fast incumbents' versions lose value. They lose it within months, and leading developers hold their positions only by releasing continually.

Three limitations frame the next steps. Our rates are rates on contestable revenue, from the most price- and capability-sensitive segment of the market; enterprise contracts will depreciate more slowly, developers' own consumer products, which switch their default model at release, probably not, and migration is likely faster for high-volume buyers, who recover a fixed cost of switching sooner. The price coefficient is not identified by a price experiment: the instrument we use measures a conditional cross-sectional gradient, and the coefficient is bracketed by that gradient, the response to list-price changes and the pricing condition rather than estimated. And the welfare statement holds rivals' strategies fixed; the release game of Appendix~\ref{app:race} agrees with it, but has pure equilibria only for some release costs. The version panel, extended as it accumulates, can answer all three.

\section*{Acknowledgments}
We used generative AI tools, primarily large language models, to edit and polish the writing of this paper. The authors reviewed all text, checked every number against the underlying estimates and every reference against its source, and take full responsibility for the content.

\clearpage
\bibliographystyle{plainnat}
\bibliography{references}

\begin{thebibliography}{64}
\providecommand{\natexlab}[1]{#1}
\providecommand{\url}[1]{\texttt{#1}}
\expandafter\ifx\csname urlstyle\endcsname\relax
  \providecommand{\doi}[1]{doi: #1}\else
  \providecommand{\doi}{doi: \begingroup \urlstyle{rm}\Url}\fi

\bibitem[Acemoglu(2025)]{Acemoglu2025}
Daron Acemoglu.
\newblock The simple macroeconomics of {AI}.
\newblock \emph{Economic Policy}, 40\penalty0 (121):\penalty0 13--58, 2025.

\bibitem[Aghion and Howitt(1992)]{AghionHowitt1992}
Philippe Aghion and Peter Howitt.
\newblock A model of growth through creative destruction.
\newblock \emph{Econometrica}, 60\penalty0 (2):\penalty0 323--351, 1992.

\bibitem[Aghion et~al.(2001)Aghion, Harris, Howitt, and
  Vickers]{AghionHarrisHowittVickers2001}
Philippe Aghion, Christopher Harris, Peter Howitt, and John Vickers.
\newblock Competition, imitation and growth with step-by-step innovation.
\newblock \emph{Review of Economic Studies}, 68\penalty0 (3):\penalty0
  467--492, 2001.

\bibitem[Aghion et~al.(2005)Aghion, Bloom, Blundell, Griffith, and
  Howitt]{AghionBloomBlundellGriffithHowitt2005}
Philippe Aghion, Nick Bloom, Richard Blundell, Rachel Griffith, and Peter
  Howitt.
\newblock Competition and innovation: An inverted-{U} relationship.
\newblock \emph{Quarterly Journal of Economics}, 120\penalty0 (2):\penalty0
  701--728, 2005.

\bibitem[Agrawal et~al.(2018)Agrawal, Gans, and
  Goldfarb]{AgrawalGansGoldfarb2018}
Ajay Agrawal, Joshua Gans, and Avi Goldfarb.
\newblock \emph{Prediction Machines: The Simple Economics of Artificial
  Intelligence}.
\newblock Harvard Business Review Press, Boston, 2018.

\bibitem[{Alphabet Inc.}(2023)]{Alphabet10Q2023}
{Alphabet Inc.}
\newblock Form 10-{Q} for the quarter ended {September} 30, 2023.
\newblock U.S. Securities and Exchange Commission, 2023.

\bibitem[{Amazon.com, Inc.}(2026)]{Amazon10K2025}
{Amazon.com, Inc.}
\newblock Form 10-{K} for the fiscal year ended {December} 31, 2025.
\newblock U.S. Securities and Exchange Commission, 2026.

\bibitem[Anderson and Hsiao(1981)]{AndersonHsiao1981}
T.~W. Anderson and Cheng Hsiao.
\newblock Estimation of dynamic models with error components.
\newblock \emph{Journal of the American Statistical Association}, 76\penalty0
  (375):\penalty0 598--606, 1981.

\bibitem[Anderson and Rubin(1949)]{AndersonRubin1949}
T.~W. Anderson and Herman Rubin.
\newblock Estimation of the parameters of a single equation in a complete
  system of stochastic equations.
\newblock \emph{Annals of Mathematical Statistics}, 20\penalty0 (1):\penalty0
  46--63, 1949.

\bibitem[{Anthropic}(2026)]{AnthropicDeprecations2026}
{Anthropic}.
\newblock Model deprecations.
\newblock platform.claude.com/docs/en/about-claude/model-deprecations, Internet
  Archive capture of November 13, 2026, 2026.

\bibitem[Argente et~al.(2018)Argente, Lee, and Moreira]{ArgenteLeeMoreira2018}
David Argente, Munseob Lee, and Sara Moreira.
\newblock Innovation and product reallocation in the {Great Recession}.
\newblock \emph{Journal of Monetary Economics}, 93:\penalty0 1--20, 2018.

\bibitem[{Artificial Analysis}(2026)]{ArtificialAnalysis2026}
{Artificial Analysis}.
\newblock Model leaderboards, benchmark scores and prices.
\newblock artificialanalysis.ai and its API; monthly Internet Archive captures
  from April 2024, accessed November 2026, 2026.

\bibitem[Barnett and Hansen(1996)]{BarnettHansen1996}
William~P. Barnett and Morten~T. Hansen.
\newblock The {Red Queen} in organizational evolution.
\newblock \emph{Strategic Management Journal}, 17\penalty0 (S1):\penalty0
  139--157, 1996.

\bibitem[Berry et~al.(1995)Berry, Levinsohn, and
  Pakes]{BerryLevinsohnPakes1995}
Steven Berry, James Levinsohn, and Ariel Pakes.
\newblock Automobile prices in market equilibrium.
\newblock \emph{Econometrica}, 63\penalty0 (4):\penalty0 841--890, 1995.

\bibitem[Berry(1994)]{Berry1994}
Steven~T. Berry.
\newblock Estimating discrete-choice models of product differentiation.
\newblock \emph{RAND Journal of Economics}, 25\penalty0 (2):\penalty0 242--262,
  1994.

\bibitem[{Bloomberg Opinion}(2025)]{Bloomberg2025Capex}
{Bloomberg Opinion}.
\newblock {DeepSeek} fails to deter {Silicon Valley}'s \$325 billion gamble of
  the century, February 2025.
\newblock Column of February 7, 2025, compiling the 2025 capital-expenditure
  guidance of Alphabet, Amazon, Meta and Microsoft.
  \url{https://www.bloomberg.com/opinion/articles/2025-02-07/deepseek-fails-to-deter-silicon-valley-s-325-billion-gamble-of-the-century}.

\bibitem[Bresnahan and Greenstein(1999)]{BresnahanGreenstein1999}
Timothy~F. Bresnahan and Shane Greenstein.
\newblock Technological competition and the structure of the computer industry.
\newblock \emph{Journal of Industrial Economics}, 47\penalty0 (1):\penalty0
  1--40, 1999.

\bibitem[Bresnahan and Trajtenberg(1995)]{BresnahanTrajtenberg1995}
Timothy~F. Bresnahan and Manuel Trajtenberg.
\newblock General purpose technologies: `engines of growth'?
\newblock \emph{Journal of Econometrics}, 65\penalty0 (1):\penalty0 83--108,
  1995.

\bibitem[Brynjolfsson et~al.(2021)Brynjolfsson, Rock, and
  Syverson]{BrynjolfssonRockSyverson2021}
Erik Brynjolfsson, Daniel Rock, and Chad Syverson.
\newblock The productivity {J}-curve: How intangibles complement general
  purpose technologies.
\newblock \emph{American Economic Journal: Macroeconomics}, 13\penalty0
  (1):\penalty0 333--372, 2021.

\bibitem[{Bureau of Economic Analysis}(2024)]{BEA2024}
{Bureau of Economic Analysis}.
\newblock {BEA} depreciation estimates, 2024.
\newblock Rates and service lives of private nonresidential fixed assets; U.S.
  Department of Commerce.

\bibitem[{CNBC}(2025)]{CNBC2025OpenAI}
{CNBC}.
\newblock {OpenAI} hits \$10 billion in annual recurring revenue fueled by
  {ChatGPT} growth, 2025.
\newblock June 9, 2025.

\bibitem[{CoreWeave, Inc.}(2026)]{CoreWeave10K2025}
{CoreWeave, Inc.}
\newblock Form 10-{K} for the fiscal year ended {December} 31, 2025.
\newblock U.S. Securities and Exchange Commission, 2026.

\bibitem[Corrado et~al.(2009)Corrado, Hulten, and
  Sichel]{CorradoHultenSichel2009}
Carol Corrado, Charles Hulten, and Daniel Sichel.
\newblock Intangible capital and {U.S.} economic growth.
\newblock \emph{Review of Income and Wealth}, 55\penalty0 (3):\penalty0
  661--685, 2009.

\bibitem[Cottier et~al.(2024)Cottier, Rahman, Fattorini, Maslej, Besiroglu, and
  Owen]{CottierEtAl2024}
Ben Cottier, Robi Rahman, Loredana Fattorini, Nestor Maslej, Tamay Besiroglu,
  and David Owen.
\newblock The rising costs of training frontier {AI} models, 2024.
\newblock arXiv:2405.21015.

\bibitem[de~Rassenfosse and Jaffe(2018)]{deRassenfosseJaffe2018}
Ga{\'e}tan de~Rassenfosse and Adam~B. Jaffe.
\newblock Econometric evidence on the depreciation of innovations.
\newblock \emph{European Economic Review}, 101:\penalty0 625--642, 2018.

\bibitem[Demirer et~al.(2025)Demirer, Fradkin, Tadelis, and
  Peng]{Demirer2025Market}
Mert Demirer, Andrey Fradkin, Nadav Tadelis, and Sida Peng.
\newblock The emerging market for intelligence: Pricing, supply, and demand for
  {LLMs}.
\newblock Working Paper 34608, National Bureau of Economic Research, 2025.

\bibitem[{Epoch AI}(2026{\natexlab{a}})]{EpochCompanies2026}
{Epoch AI}.
\newblock Data on {AI} companies: Revenue, compute spending and usage reports.
\newblock epoch.ai/data/ai-companies, CC BY 4.0, accessed November 2026,
  2026{\natexlab{a}}.

\bibitem[{Epoch AI}(2026{\natexlab{b}})]{EpochECI2026}
{Epoch AI}.
\newblock {Epoch Capabilities Index}.
\newblock epoch.ai/data/benchmark-data, CC BY 4.0, accessed November 2026,
  2026{\natexlab{b}}.

\bibitem[{Epoch AI}(2026{\natexlab{c}})]{EpochModels2026}
{Epoch AI}.
\newblock Data on {AI} models.
\newblock epoch.ai/data/ai-models, CC BY 4.0, accessed November 2026,
  2026{\natexlab{c}}.

\bibitem[Ericson and Pakes(1995)]{EricsonPakes1995}
Richard Ericson and Ariel Pakes.
\newblock Markov-perfect industry dynamics: A framework for empirical work.
\newblock \emph{Review of Economic Studies}, 62\penalty0 (1):\penalty0 53--82,
  1995.

\bibitem[Fradkin(2025)]{Fradkin2025}
Andrey Fradkin.
\newblock Demand for {LLMs}: Descriptive evidence on substitution, market
  expansion, and multihoming.
\newblock arXiv:2504.15440, 2025.

\bibitem[Gandhi and Houde(2019)]{GandhiHoude2019}
Amit Gandhi and Jean-Fran{\c{c}}ois Houde.
\newblock Measuring substitution patterns in differentiated-products
  industries.
\newblock Working Paper 26375, National Bureau of Economic Research, 2019.

\bibitem[Garcia-Macia et~al.(2019)Garcia-Macia, Hsieh, and
  Klenow]{GarciaMaciaHsiehKlenow2019}
Daniel Garcia-Macia, Chang-Tai Hsieh, and Peter~J. Klenow.
\newblock How destructive is innovation?
\newblock \emph{Econometrica}, 87\penalty0 (5):\penalty0 1507--1541, 2019.

\bibitem[Goettler and Gordon(2011)]{GoettlerGordon2011}
Ronald~L. Goettler and Brett~R. Gordon.
\newblock Does {AMD} spur {Intel} to innovate more?
\newblock \emph{Journal of Political Economy}, 119\penalty0 (6):\penalty0
  1141--1200, 2011.

\bibitem[{Goldman Sachs}(2025)]{GoldmanCapex2026}
{Goldman Sachs}.
\newblock Why {AI} companies may invest more than \$500 billion in 2026.
\newblock Goldman Sachs Insights,
  \url{https://www.goldmansachs.com/insights/articles/why-ai-companies-may-invest-more-than-500-billion-in-2026},
  2025.

\bibitem[{Google}(2026)]{GeminiDeprecations2026}
{Google}.
\newblock Gemini deprecations.
\newblock ai.google.dev/gemini-api/docs/deprecations, accessed November 2026,
  2026.

\bibitem[Hoffmann et~al.(2022)Hoffmann, Borgeaud, Mensch,
  et~al.]{HoffmannEtAl2022}
Jordan Hoffmann, Sebastian Borgeaud, Arthur Mensch, et~al.
\newblock Training compute-optimal large language models.
\newblock In \emph{Advances in Neural Information Processing Systems},
  volume~35, 2022.

\bibitem[Hottman et~al.(2016)Hottman, Redding, and
  Weinstein]{HottmanReddingWeinstein2016}
Colin~J. Hottman, Stephen~J. Redding, and David~E. Weinstein.
\newblock Quantifying the sources of firm heterogeneity.
\newblock \emph{Quarterly Journal of Economics}, 131\penalty0 (3):\penalty0
  1291--1364, 2016.

\bibitem[Hulten(1992)]{Hulten1992}
Charles~R. Hulten.
\newblock Growth accounting when technical change is embodied in capital.
\newblock \emph{American Economic Review}, 82\penalty0 (4):\penalty0 964--980,
  1992.

\bibitem[Hulten and Wykoff(1981)]{HultenWykoff1981}
Charles~R. Hulten and Frank~C. Wykoff.
\newblock The measurement of economic depreciation.
\newblock In Charles~R. Hulten, editor, \emph{Depreciation, Inflation, and the
  Taxation of Income from Capital}, pages 81--125. Urban Institute Press,
  Washington, DC, 1981.

\bibitem[Igami(2017)]{Igami2017}
Mitsuru Igami.
\newblock Estimating the innovator's dilemma: Structural analysis of creative
  destruction in the hard disk drive industry, 1981--1998.
\newblock \emph{Journal of Political Economy}, 125\penalty0 (3):\penalty0
  798--847, 2017.

\bibitem[Jord{\`a}(2005)]{Jorda2005}
{\`O}scar Jord{\`a}.
\newblock Estimation and inference of impulse responses by local projections.
\newblock \emph{American Economic Review}, 95\penalty0 (1):\penalty0 161--182,
  2005.

\bibitem[Jorgenson(1963)]{Jorgenson1963}
Dale~W. Jorgenson.
\newblock Capital theory and investment behavior.
\newblock \emph{American Economic Review}, 53\penalty0 (2):\penalty0 247--259,
  1963.

\bibitem[Kaplan and Meier(1958)]{KaplanMeier1958}
E.~L. Kaplan and Paul Meier.
\newblock Nonparametric estimation from incomplete observations.
\newblock \emph{Journal of the American Statistical Association}, 53\penalty0
  (282):\penalty0 457--481, 1958.

\bibitem[Korinek and McKelvey(2026)]{KorinekMcKelvey2026}
Anton Korinek and Patrick McKelvey.
\newblock Measuring the {AI} economy.
\newblock Working Paper 26-9, Peterson Institute for International Economics,
  2026.

\bibitem[Kshirsagar(2025)]{CITP2025Chips}
Mihir Kshirsagar.
\newblock Lifespan of {AI} chips: The \$300 billion question.
\newblock Freedom to Tinker (blog of the Princeton Center for Information
  Technology Policy), October 15, 2025, 2025.

\bibitem[Li and Hall(2020)]{LiHall2020}
Wendy C.~Y. Li and Bronwyn~H. Hall.
\newblock Depreciation of business {R\&D} capital.
\newblock \emph{Review of Income and Wealth}, 66\penalty0 (1):\penalty0
  161--180, 2020.

\bibitem[Lisk(2026)]{Lisk2026}
David Lisk.
\newblock Substitution across {LLM} {API} products in the {OpenRouter} market.
\newblock Honors thesis, Glucksman Institute for Research in Securities
  Markets, NYU Stern School of Business, April 2026, 2026.

\bibitem[Lord(1980)]{Lord1980}
Frederic~M. Lord.
\newblock \emph{Applications of Item Response Theory to Practical Testing
  Problems}.
\newblock Lawrence Erlbaum Associates, Hillsdale, NJ, 1980.

\bibitem[Mankiw and Whinston(1986)]{MankiwWhinston1986}
N.~Gregory Mankiw and Michael~D. Whinston.
\newblock Free entry and social inefficiency.
\newblock \emph{RAND Journal of Economics}, 17\penalty0 (1):\penalty0 48--58,
  1986.

\bibitem[{Meta Platforms, Inc.}(2026)]{Meta10K2025}
{Meta Platforms, Inc.}
\newblock Form 10-{K} for the fiscal year ended {December} 31, 2025.
\newblock U.S. Securities and Exchange Commission, 2026.

\bibitem[{Microsoft}(2026)]{AzureRetirements2026}
{Microsoft}.
\newblock Azure {OpenAI} in {Foundry Models}: Model deprecations and
  retirements.
\newblock learn.microsoft.com, accessed November 2026, 2026.

\bibitem[{Microsoft Corporation}(2023)]{Microsoft10K2023}
{Microsoft Corporation}.
\newblock Form 10-{K} for the fiscal year ended {June} 30, 2023.
\newblock U.S. Securities and Exchange Commission, 2023.

\bibitem[{Mistral AI}(2026)]{MistralModels2026}
{Mistral AI}.
\newblock Models: Deprecated and retired models.
\newblock docs.mistral.ai/models, accessed November 2026, 2026.

\bibitem[Nevo(2003)]{Nevo2003}
Aviv Nevo.
\newblock New products, quality changes, and welfare measures computed from
  estimated demand systems.
\newblock \emph{Review of Economics and Statistics}, 85\penalty0 (2):\penalty0
  266--275, 2003.

\bibitem[{OpenAI}(2026)]{OpenAIDeprecations2026}
{OpenAI}.
\newblock Deprecations.
\newblock developers.openai.com/api/docs/deprecations, accessed November 2026,
  2026.

\bibitem[{OpenRouter}(2026{\natexlab{a}})]{OpenRouterApps2026}
{OpenRouter}.
\newblock Application pages.
\newblock openrouter.ai/apps, Internet Archive captures, May 2025 to November
  2026, 2026{\natexlab{a}}.

\bibitem[{OpenRouter}(2026{\natexlab{b}})]{OpenRouterRankings2026}
{OpenRouter}.
\newblock {LLM} rankings and data {API}.
\newblock openrouter.ai/rankings, CC BY 4.0; Internet Archive captures,
  November 2023 to October 2026, accessed November 2026, 2026{\natexlab{b}}.

\bibitem[Oster(2019)]{Oster2019}
Emily Oster.
\newblock Unobservable selection and coefficient stability: Theory and
  evidence.
\newblock \emph{Journal of Business \& Economic Statistics}, 37\penalty0
  (2):\penalty0 187--204, 2019.

\bibitem[Pakes and McGuire(1994)]{PakesMcGuire1994}
Ariel Pakes and Paul McGuire.
\newblock Computing {M}arkov-perfect {N}ash equilibria: Numerical implications
  of a dynamic differentiated product model.
\newblock \emph{RAND Journal of Economics}, 25\penalty0 (4):\penalty0 555--589,
  1994.

\bibitem[Petrin(2002)]{Petrin2002}
Amil Petrin.
\newblock Quantifying the benefits of new products: The case of the minivan.
\newblock \emph{Journal of Political Economy}, 110\penalty0 (4):\penalty0
  705--729, 2002.

\bibitem[Sargan(1958)]{Sargan1958}
J.~D. Sargan.
\newblock The estimation of economic relationships using instrumental
  variables.
\newblock \emph{Econometrica}, 26\penalty0 (3):\penalty0 393--415, 1958.

\bibitem[Spence(1976)]{Spence1976}
Michael Spence.
\newblock Product selection, fixed costs, and monopolistic competition.
\newblock \emph{Review of Economic Studies}, 43\penalty0 (2):\penalty0
  217--235, 1976.

\bibitem[{United Nations, European Commission, IMF, OECD and World
  Bank}(2025)]{SNA2025}
{United Nations, European Commission, IMF, OECD and World Bank}.
\newblock {System of National Accounts 2025}.
\newblock Adopted by the United Nations Statistical Commission, March 2025,
  2025.

\end{thebibliography}

\clearpage
\appendix
\section{Microfoundation and proofs}\label{app:proofs}

\subsection{Developers}\label{app:micro}

Developer $j$ served by version $i$ produces $Y_j=O_j\,[K_i\,h(d_j)]^{\psi}L_j^{1-\psi}$, where $O_j$ is its orchestration capital, $d_j\ge0$ the architectural complexity of its workload (reasoning depth, tool calls, context), $h(d)=1+\phi d^{\nu}$ with $\nu\in(0,1)$ the efficiency of that architecture, and $L_j$ labor hired at wage $w$. Token use is $x_{ji}=\tau(d_j)Y_j$ with $\tau(d)=1+\xi_\tau d$, and profit is $(P_Y-p_i\tau(d_j))Y_j-wL_j$. Write $m(d,p)\equiv P_Y-p\tau(d)$.

\begin{lemma}[Indirect profit]\label{lem:indirect}
Given $(K_i,p_i)$, the developer's maximized profit is $\Pi^*_{ji}=\Gamma_jK_i\Psi(p_i)$ with $\Gamma_j\equiv\psi O_j^{1/\psi}((1-\psi)/w)^{(1-\psi)/\psi}$ and $\Psi(p)\equiv\max_{d\ge0}h(d)\,m(d,p)^{1/\psi}$. The optimal architecture $d^*(p)$ is unique and decreasing in $p/P_Y$.
\end{lemma}

\noindent\emph{Proof.} The labor first-order condition gives $wL=(1-\psi)mY$; substituting,
\[
Y=O_j^{1/\psi}K_ih(d)\big((1-\psi)m/w\big)^{(1-\psi)/\psi},\qquad \Pi=\psi mY=\Gamma_jK_ih(d)m^{1/\psi}.
\] The derivative of $\ln h(d)+\psi^{-1}\ln m(d,p)$ in $d$, $h'/h-\psi^{-1}p\tau'/m$, is strictly decreasing and positive as $d\to0$, so the maximizer is unique; its first-order condition,
\[
\psi\phi\nu d^{\nu-1}\big(1-\pi(1+\xi_\tau d)\big)=\xi_\tau\pi(1+\phi d^\nu),\qquad \pi=p/P_Y,
\]
shows that $d^*$ falls with $\pi$. \qed

\smallskip
By the envelope theorem $\dd\ln\Psi/\dd\ln p=-p\tau(d^*)/(\psi m(d^*,p))$, so $\lambda=\theta p\tau(d^*)/(\psi m)$, and $\partial\ln s_i/\partial\ln p_i=\theta(1-s_i)\dd\ln\Psi_i/\dd\ln p_i$ gives \eqref{eq:lambda}. The elasticity of a developer's token use with respect to its supplier's price, $\varepsilon_\chi=\varepsilon_{\tau,d}\varepsilon_d+\varepsilon_{h,d}\varepsilon_d+\frac{1-\psi}{\psi}\frac{\pi\tau(d^*)}{1-\pi\tau(d^*)}$ with $\varepsilon_d\equiv-\dd\ln d^*/\dd\ln p$, is the intensive margin; aggregating over the developers that choose $i$ gives $Q_i=\Omega_ts_iK_i\chi(p_i)$ with $\Omega_t\equiv\int O_j^{1/\psi}\dd j$. A further envelope result is used below: $\partial\Pi^*_{ji}/\partial p_i=-x_{ji}$, so that $\Pi^*_{ji}=-x_{ji}p_i/(\dd\ln\Psi/\dd\ln p)=\theta\,p_ix_{ji}/\lambda_i$: a developer's profit is $\theta/\lambda$ times its spending on tokens.

\subsection{Proofs}

\paragraph{Proof of equation~\eqref{eq:lerner}.} $\partial\ln Q_i/\partial\ln p_i=-\lambda_i(1-s_i)-\varepsilon_\chi$; the first-order condition $Q_i+(p_i-MC_i)\partial Q_i/\partial p_i=0$ rearranges to the stated Lerner index. Equation~\eqref{eq:pi_elast} follows from \eqref{eq:logit} and $Q_i\propto s_iK_i$. \qed

\paragraph{Proof of Proposition~\ref{prop:release}.} With rivals fixed, the released version earns $\Pi(\hat k)$ until the release and $\Pi(\hat ke^{\Delta})$ after it; the next release is ignored, which is exact if subsequent gains are discounted enough. Releasing when the gap is $\Delta$, at $T=\Delta/g$, is worth $e^{-rT}[G(\Delta)/r-F]$. Its derivative in $T$ is $e^{-rT}\{-r[G/r-F]+gG'(\Delta)/r\}$, positive where $gG'>r(G-rF)$ and negative where the inequality is reversed. With $G$ S-shaped, $gG'$ is single-peaked and $r(G-rF)$ increasing, so they cross at most twice; the second crossing, on the concave part, is the maximum $\Delta^*$, which exists when $G$ eventually exceeds $rF$ and $G'\to0$. Differentiating \eqref{eq:release}, $\dd\Delta^*/\dd F=-r^2/[gG''(\Delta^*)-rG'(\Delta^*)]>0$ on the concave part. \qed

\paragraph{Proof of Proposition~\ref{prop:decomp}.} $\pi^*_i=\varpi_ip_iQ_i$. With $K_i$ fixed, $\dd\ln Q_i/\dd t=g_\Omega+\dd\ln s_i/\dd t-\varepsilon_\chi\hat p_i$, and from \eqref{eq:logit} $\dd\ln s_i/\dd t=\theta[(\hat k_i+\hat\Psi_i)-\sum_{m\ge0}s_m(\hat k_m+\hat\Psi_m)]$. With $\hat k_i=-g_A$, $\hat k_m=g_m-g_A$ for rivals and $\hat k_0=0$, the capability terms collect to $-s_0g_A-\sum_{m\ne i,0}s_mg_m$; $\theta\hat\Psi_m=-\lambda_m\hat p_m$ gives the price terms; the $\hat p_i$ in $p_iQ_i$ combines with the intensive-margin term, and the margin term is the remainder. Corollary~\ref{cor:rq} sets prices and margins constant and $g_m=g_A$ for rivals, or $g_m=0$ for all proprietary versions. \qed

\paragraph{Proof of Proposition~\ref{prop:cross}.} From \eqref{eq:pi_elast}, $\partial\pi^*_i/\partial k_i=\pi^*_i(1+\theta(1-s_i))/k_i$. Differentiating in $k_m$ with $\partial\ln\pi^*_i/\partial\ln k_m=-\theta s_m$ and $\partial s_i/\partial\ln k_m=-\theta s_is_m$ gives $\partial^2\pi^*_i/\partial k_i\partial k_m=\frac{\pi^*_i\theta s_m}{k_ik_m}[2\theta s_i-(1+\theta)]$, negative if and only if $s_i<(1+\theta)/(2\theta)$. \qed

\paragraph{Proof of Proposition~\ref{prop:welfare}.} Developers' surplus in money is $\mathcal C=\sum_j\E\max_i\{\Pi^*_{ji}e^{\epsilon_{ji}/\theta}\}$ up to the normalization of the match term; its derivative with respect to $\ln k_f$ is the profit of the developers who choose $f$, $\sum_{j\to f}\Pi^*_{jf}$, because $\Pi^*$ is proportional to $K$ (Lemma~\ref{lem:indirect}) and the choice is optimal. By the envelope result of Appendix~\ref{app:micro}, $\Pi^*_{jf}=\theta p_fx_{jf}/\lambda$, so $\partial\mathcal C/\partial\ln k_f=\theta R_f/\lambda$. Let the planner maximize the present value of $\mathcal W=\sum_m\pi^*_m+\mathcal C-\sum_m[p_cc_m+\Xi(c_m)]$ subject to the same accumulation equations. Differentiating the firm's and the planner's Hamilton--Jacobi--Bellman equations with respect to $k_f$ at the symmetric state and holding rivals' policies fixed (so that the strategic terms $\sum_{m\ne f}(\partial V_f/\partial k_m)(\partial\dot k_m/\partial k_f)$ vanish), the terms in $\partial\dot k_f/\partial k_f=-(\delta_0+g_A)$ move to the left and the own-policy terms drop by the envelope theorem, leaving $(r+\delta_0+g_A)(q^S_f-q_f)=\partial(\mathcal W-\pi^*_f)/\partial k_f=\sum_{m\ne f}\partial\pi^*_m/\partial k_f+\partial\mathcal C/\partial k_f=-B_f+U_f$, with $B_f=\theta s_f\sum_{m\ne f}\pi^*_m/k_f$ by \eqref{eq:pi_elast} and $U_f=\theta R_f/(\lambda k_f)$. The firm's first-order condition $p_c+\Xi'(c_f)=q_f\mathcal F'(c_f)$ and the planner's analogue imply that equilibrium compute exceeds the planner's if and only if $q_f>q^S_f$. With a common margin $\pi^*_m=\varpi R_m$ (the open-weight option earns none) and symmetric prices, $R_m/R_f=s_m/s_f$, so $B_f/U_f=\lambda\varpi\sum_{m\ne f,0}s_m$. Under \eqref{eq:lerner} with $\varepsilon_\chi=0$, $\varpi=1/(\lambda(1-s_f))$, which gives \eqref{eq:BUbertrand}. \qed

\section{Data construction}\label{app:data}

\paragraph{Captures and windows.} The OpenRouter panel uses one Internet Archive capture of \texttt{openrouter.ai/rankings} per calendar week, preferring captures whose embedded table reports a trailing seven-day window ending in that week. Window classification compares each table's totals for the five largest models with the 52-week weekly chart series that the same page embeds; a ratio near one identifies a weekly table, near one seventh a daily table. Captures embedding only the twenty largest models (from August 15, 2026) are retained with a flag and are not used for market totals.

\paragraph{Data API.} For weeks after May 11, 2026 the panel uses the \texttt{rankings-daily} dataset of OpenRouter's data API, which reports the fifty largest model variants per UTC day and an aggregate \emph{other} row. Days are summed to calendar weeks with seven observed days. Validation on the 91 overlapping weeks from January 2025 to May 2026 gives a median ratio of API to archived weekly totals of 1.02 and a within-week log-share correlation across models of 0.92.

\paragraph{Chain-linking of the composite index.} As a robustness check on the item-response capability scale we chain-link AA's composite index across its revisions by regressing the log index on model and capture fixed effects; capture effects absorb every rescaling and the model effects give quality on a common scale. Re-basings are detected as discrete shifts in the median ratio of the index between consecutive captures and occur in September 2025, October 2025, February 2026, July 2026, and October 2026. The chained index and the item-response scale have a rank correlation of \irtcorr\ across the versions in the panel.

\paragraph{Crosswalk.} OpenRouter identifiers are matched to AA versions by developer and by the tokens of the model name after removing dates and variant tags; AA reasoning-effort variants are collapsed to their base version. Matches among the three hundred largest OpenRouter models are verified by hand and the remainder resolved with an override file. Stealth pre-release models that OpenRouter listed under its own name and later identified publicly are mapped to the released version as free tiers.

\paragraph{Category panel.} The data API's \texttt{rankings-daily} dataset with a category parameter and a weekly period returns, for each Monday-start week, the tokens of the largest models in the category and an aggregate remainder; we retrieved the twelve categories for 2025-01-01 onward (data begin May 19, 2025). The archived category pages (\texttt{openrouter.ai/}\allowbreak\texttt{rankings/<category>}, 262 captures) embed, until mid-2025, a table of weekly totals for every model in the category over the preceding year, anchored on the weekday of the capture; each row is assigned to the Monday-start week containing its date, and where several captures cover the same category-week-model the latest capture is kept. API weeks replace archived weeks where both exist. Model identifiers are mapped to versions with the crosswalk, free tiers kept as separate products. Twelve categories are used; two sub-categories archived only on the pages (scripting, web) are dropped.

\paragraph{Archived listings.} OpenRouter's model-listing endpoint (\texttt{api/v1/models}) was captured on 234 distinct days between July 2023 and October 2026 (captures with a JSON body, one per day); each capture lists every model with the per-token price at which the marketplace sold it on that day, which is the developer's list price for first-party models and the price of the host the marketplace routed to for multi-host open-weight models. A price regime is a run of consecutive captures with identical prices; a change is dated to the first capture that shows it and its exposure is the interval since the previous capture, capped at eight weeks. The provider listing (\texttt{api/frontend/all-providers}, 760 captures from December 2024) gives the set of inference providers on each day. Artificial Analysis's provider pages (\texttt{artificialanalysis.ai/}\allowbreak\texttt{models/<slug>/providers}) embed, in each capture until February 2026, AA's complete model table with every host's input and output price; we use one capture per day, 366 days from March 2024, to measure the number of hosts per model.

\paragraph{Release calendar.} A version's birth date is the earlier of AA's release date and its first week with traffic on OpenRouter. Comparable rivals for the event studies are versions from other developers whose capability at release is within 0.25 log points of the chained index below, or above, the incumbent's, about 0.28 standard deviations of the item-response scale.

\begin{table}[htbp]
\centering
\caption{Summary statistics of the version panel (paid versions)}\label{tab:summary}
\footnotesize
\begin{tabular}{lrrrrr}
\toprule
 & Mean & SD & P10 & Median & P90\\
\midrule
Tokens per version-week (billions) & 140.77 & 942.67 & 0.04 & 1.85 & 172.82\\
Revenue per version-week (\$ thousand) & 227.5 & 1484.4 & 0.0 & 1.3 & 144.9\\
Revenue share among paid versions (percent) & 1.073 & 5.013 & 0.001 & 0.040 & 1.636\\
Price, \$ per million tokens (blended) & 4.02 & 16.67 & 0.13 & 0.76 & 6.00\\
Capability (item-response scale, SD across priced models) & 0.03 & 0.78 & -0.91 & -0.08 & 1.22\\
Capability (log chained index) & 2.06 & 0.72 & 1.18 & 2.00 & 3.02\\
Tokens per request (thousands) & 5.2 & 8.0 & 0.8 & 2.7 & 11.3\\
Age (weeks since first traffic) & 30 & 24 & 4 & 25 & 65\\
Open weights (0/1) & 0.55 & 0.50 & 0.00 & 1.00 & 1.00\\
Reasoning model (0/1) & 0.17 & 0.38 & 0.00 & 0.00 & 1.00\\
\bottomrule

\end{tabular}
\begin{minipage}{0.95\textwidth}\footnotesize\emph{Notes:} Version-weeks of paid versions in the linked panel, February 2024 to October 2026 (12,900 observations). Capability on the item-response scale is in standard deviations across priced models, the units of Table~\ref{tab:demand}; the chained-index row is the log of the chain-linked composite index. Tokens per request from the archived weeks (the data-API weeks do not report requests).\end{minipage}
\end{table}

\begin{table}[htbp]
\centering
\caption{Coverage of the linked panel by half-year}\label{tab:coverage}
\footnotesize
\begin{tabular}{lrrrrrrr}
\toprule
Period & Weeks & Products & Versions & Developers & Tokens (T) & Free-tier share & Matched share\\
\midrule
2024H1 & 20 & 46 & 40 & 15 & 1.3 & 0.03 & 0.73\\
2024H2 & 27 & 97 & 81 & 18 & 5.3 & 0.01 & 0.85\\
2025H1 & 24 & 205 & 154 & 21 & 44.6 & 0.33 & 0.96\\
2025H2 & 26 & 258 & 207 & 28 & 124.0 & 0.12 & 0.97\\
2026H1 & 26 & 282 & 247 & 32 & 600.9 & 0.09 & 0.90\\
2026H2 & 15 & 107 & 99 & 20 & 1599.7 & 0.09 & 0.81\\
\bottomrule

\end{tabular}
\begin{minipage}{0.95\textwidth}\footnotesize\emph{Notes:} Tokens are platform totals over the half-year in trillions; free-tier share is the share of matched tokens served by free tiers; matched share is the share of platform tokens carried by versions with AA characteristics. 2026H2 covers July to October 12 and consists of data-API weeks, in which only the daily top fifty variants are itemized.\end{minipage}
\end{table}

\section{Additional results for Section~\ref{sec:facts}}\label{app:facts}

\paragraph{Peak-normalized profiles.} Normalizing each version's path by its largest share within its first twelve weeks, rather than by its share at age four, gives a median half-life of four weeks for frontier versions and three weeks for all versions, and a within-version decay of $0.075$ log points per week; the difference from the age-four normalization is selection on the noisy maximum, which is why the main text uses a fixed age.

\paragraph{Price changes behind the weighted matched-model index.} Table~\ref{tab:pricechanges} lists the ten largest token-weighted contributions to the chained matched-model index of Fact~3. Two increases, Gemini 2.5 Flash in July 2025 and DeepSeek V4 Flash in August 2026, account for most of the weighted index's rise; the unweighted index falls at 14 percent per year.

\begin{table}[htbp]
\centering
\caption{Largest weighted contributions to within-version price changes}\label{tab:pricechanges}
\scriptsize
\begin{tabular}{llrr}
\toprule
Version & Week & Log price change & Weighted contribution\\
\midrule
deepseek-v4-flash & 2026-08-31 & +1.33 & +0.303\\
gemini-2-5-flash & 2025-07-28 & +1.17 & +0.155\\
deepseek-v4-pro & 2026-06-01 & -1.39 & -0.068\\
mistral-7b-instruct & 2024-03-18 & +0.21 & +0.052\\
mimo-v2-5-pro & 2026-06-01 & -1.01 & -0.049\\
deepseek-v4-pro & 2026-08-31 & +1.29 & +0.048\\
deepseek-v3-2 & 2026-06-01 & +0.90 & +0.036\\
llama-3-instruct-8b & 2024-05-06 & +0.26 & +0.031\\
gemini-1-5-pro & 2024-06-17 & -0.69 & -0.031\\
gemini-pro & 2024-04-08 & -0.51 & -0.026\\
\bottomrule
\end{tabular}

\begin{minipage}{0.9\textwidth}\footnotesize\emph{Notes:} Log change in the version's blended list price in the week shown, multiplied by the version's share of platform tokens in the previous week.\end{minipage}
\end{table}

\paragraph{Marketplace price changes.} Table~\ref{tab:priceregimes} summarizes price changes in the archived model listings by year and Figure~\ref{fig:repricing} plots the monthly hazard. Changes become more frequent over time, from 1.8 per model-year in 2024 to 4.8 in 2026, and the share of rises grows from a quarter to a half, because an increasing number of models are open-weight models served by several hosts, so that the marketplace's price moves whenever it routes to a different host. Splitting models by whether their developer prices them itself (the closed-weight developers and the open-weight developers with their own API) or whether they are served only by third-party hosts, the hazard is 1.0, 2.1 and 4.0 per model-year in 2024, 2025 and 2026 for the first group and 2.6, 3.7 and 5.9 for the second; the version panel uses the developers' posted list prices, whose changes are the first group's and are rarer still at the level of the developer's own listing (4.8 percent of version-months). Among the 318 listed models matched to the version panel the hazard is 2.5 per model-year. The DeepSeek window of January 20 to March 16, 2025 is not a period of unusual repricing: the hazard inside it is 1.1 per model-year against 3.7 outside, with 17 cuts among 246 listed paid models, five of them DeepSeek's own and three Meta's. Among the 172 models listed both in the month before the window and in the month after it, those in the top quartile of the pre-event log price gap to DeepSeek V3 lowered their price by 9 percent on average and those in the bottom quartile raised it by 5 percent; the instrument of Section~\ref{sec:iv} therefore works through the entrant's effect on relative prices and this modest differential repricing, not through a general price war. The number of inference providers listed on the marketplace rose from 41 in December 2024 to 74 in July 2026 (128 ever listed), while the median number of hosts serving a model on AA stayed at two, with between 55 and 71 percent of models served by more than one host in each half-year through 2026H1.

\paragraph{A bound on the margin from provider prices.} For open-weight models that the developer sells through its own API while third-party hosts serve the same weights, the gap between the developer's price and the cheapest third-party price bounds the developer's gross margin, if the third-party floor approximates the marginal cost of serving. AA's provider pages give 2,833 model-days on 43 such models with at least two third-party hosts (Alibaba, Mistral, DeepSeek, Cohere and Moonshot account for most of them). The developer's blended price exceeds the cheapest third-party price in 86 percent of model-days. The implied gross margin, $1-p^{\min}_{\text{host}}/p_{\text{developer}}$, has a median of 0.73 (the developer's price is about 3.7 times the cheapest host's) and an interquartile range from 0.20 to 0.83; by developer the median margin is 0.83 for Alibaba, 0.66 for Mistral, 0.19 for Moonshot, 0.10 for DeepSeek and zero for Cohere, and the log dispersion among third-party hosts of the same weights has a median of 1.1. The grid of Section~\ref{sec:welfare}, 0.4 to 0.8, spans the interquartile range above its lower quartile; the closed-weight developers' margins are not observed.

\begin{table}[htbp]
\centering
\caption{Price changes in the archived marketplace listings, by year}\label{tab:priceregimes}
\scriptsize\setlength{\tabcolsep}{3.5pt}
\begin{tabular}{lrrrrrrrr}
\toprule
Year & Paid models & Captures & Changes & Cuts & Rises & Per model-year & Median cut (log) & Median rise (log)\\
\midrule
2023 & 65 & 5 & 20 & 17 & 3 & 3.33 & -0.47 & 0.29\\
2024 & 271 & 15 & 180 & 134 & 46 & 1.76 & -0.29 & 0.59\\
2025 & 496 & 93 & 709 & 389 & 320 & 2.78 & -0.20 & 0.28\\
2026 & 794 & 121 & 1375 & 717 & 658 & 4.81 & -0.19 & 0.25\\
\bottomrule

\end{tabular}
\begin{minipage}{0.95\textwidth}\footnotesize\emph{Notes:} Blended (3 input : 1 output) per-token prices from Internet Archive captures of \texttt{openrouter.ai/api/v1/models}; free tiers excluded. A change is counted when a model's price differs between consecutive captures; changes per model-year is the exposure-weighted hazard, changes divided by model-weeks between consecutive captures (intervals capped at eight weeks), times 52.\end{minipage}
\end{table}

\begin{figure}[htbp]
\centering
\includegraphics[width=\textwidth]{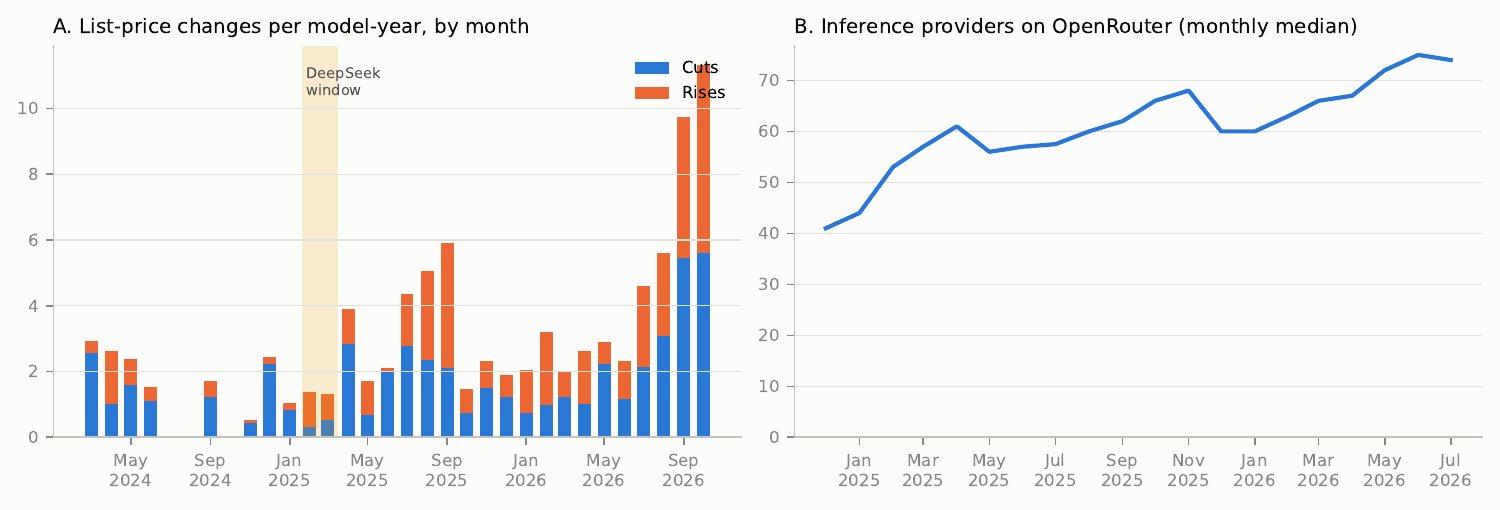}
\caption{Marketplace repricing and provider entry}\label{fig:repricing}
\begin{minipage}{0.95\textwidth}\footnotesize\emph{Notes:} (A) Price changes per model-year by month of the capture that first shows them, split into cuts and rises; the shaded band is the instrument window of Section~\ref{sec:iv}. (B) Number of inference providers in OpenRouter's archived provider listing (\texttt{api/frontend/all-providers}), monthly median across captures; the listing's format is unchanged over the period, unlike AA's provider pages.\end{minipage}
\end{figure}

\paragraph{Flagship statistics.} Over its life GPT-4o (entry May 2024) faced 74 comparable rival releases and 19 own successors with one list-price change; Claude 3.5 Sonnet (June 2024) 68 and 9 with none; Gemini 1.5 Pro (March 2024) 48 and 7 with two; DeepSeek V3 (November 2024) 63 and 10 with five. The mean four-week log change in revenue share after comparable rival releases is $-0.17$, $-0.29$, $-0.29$ and $-0.16$ respectively.

\paragraph{Role-play traffic.} Role-play and character-chat applications are below five percent of top-twenty application tokens from 2026, and the versions they favor are small community fine-tunes that are mostly unmatched and therefore already in the outside option rather than in the version panel. Because the application panel does not report which model each application uses, we cannot remove role-play traffic from a matched version's tokens directly; the frontier versions that drive Facts~1 and 2 are proprietary and open-weight flagships whose traffic, by the application rankings, comes predominantly from coding and general-purpose agents.

\section{The Item-Response Capability Scale}\label{app:irt}

AA's composite index cannot be used directly in a panel: it has had three named versions since 2024 and has been re-based under an unchanged name at least five more times, and comparing the scores of the same models before and after a re-base gives median ratios between 0.56 and 0.88. We therefore place every version on one scale with an item-response model of AA's benchmark results. Each benchmark field in each methodology version is an item; a field starts a new version when, between consecutive captures, more than half of the models scored in both change score and the median change is at least one percentage point. Keeping items with at least 15 scored models leaves 40 items from 33 fields, from MMLU and the Chatbot Arena rating in 2024 to agentic and terminal-use benchmarks in 2026. Accuracies are transformed to the logit scale after clipping at 0.5 and 99.5 percent, rating-type fields are used in levels, and each item is standardized. The model is
\[
y_{mb}=\alpha_b+\lambda_b\theta_m+e_{mb},\qquad e_{mb}\sim N(0,\sigma_b^2),\quad \lambda_b>0,
\]
where $y_{mb}$ is the transformed score of version $m$ on item $b$, $\theta_m$ its capability, $\lambda_b$ the item's discrimination and $\sigma_b$ its noise: a continuous-response version of the two-parameter logistic model \citep{Lord1980}, or a one-factor model with item-specific noise. Versions scored on items from several periods link the items of those periods to one scale (GPT-4o of May 2024 has 17 items spanning 2024 to 2026), and a benchmark that saturates loses discrimination at the top and receives less weight there. The identifying assumption is that a version's capability does not change, which is why versions are identified by AA's persistent identifiers and aliases that point to updated weights are separate versions.

We estimate the model by Gibbs sampling with a standard normal prior on $\theta_m$, normal priors on $(\alpha_b,\lambda_b)$ with means $(0,0.8)$ and variances $(4,1)$ truncated to $\lambda_b>0$, and Inverse-Gamma$(1,0.05)$ priors on $\sigma_b^2$. Two chains of 4,000 iterations, with 1,500 discarded and every 25th draw kept, give 200 draws; the largest potential scale reduction factor across the 733 versions is 1.04. Each draw is re-expressed so that GPT-4o (May 2024) has $\theta=0$ and capability has unit standard deviation across priced models. The median posterior standard deviation of a version's capability is 0.12, and versions are scored on a median of 17 items. The correlation between observed and fitted scores has a median of 0.88 across items and a minimum of 0.70, and within each version of AA's composite index the rank correlation between $\theta$ and the index is between 0.91 and 0.99. The posterior means are the capability measure of the paper; Section~\ref{sec:demres} uses the posterior standard deviations to bound attenuation, and the knowledge-item and agentic-item scales of Section~\ref{sec:demres} re-estimate the same model on those subsets of items.

\section{Robustness}\label{app:robust}

Table~\ref{tab:robustsummary} collects the robustness checks of the paper's main estimates in one place; Table~\ref{tab:robust} reports the earlier checks on sample rules, capability measures, anchors and instrument placebos in more detail.

\begin{table}[htbp]\centering\small
\caption{Robustness of the main estimates: summary}\label{tab:robustsummary}
\resizebox{\textwidth}{!}{\begin{tabular}{llc}
\toprule
Estimate & Variant & Value\\
\midrule
\emph{Version decay} & Within-version slope of log revenue share, frontier, ages 4--56 (per week; Fact 1, 62 versions) & $-0.070$\\
 & \quad service-quality sample, without / with speed, latency, hosts and price & $-0.069$ / $-0.066$\\
 & \quad validation sample (105 versions) / dropping weak identifier matches & $-0.070$ / $-0.076$\\
 & \quad validation sample, dropping weeks reconstructed from ranking tables only & $-0.062$ (0.016)\\
 & \quad inside applications (top-ten observed; all / large applications) & $-0.056$ / $-0.066$\\
 & Half-life across frontier cutoffs and sample rules (Fact 1) & 10--14 weeks\\
\midrule
\emph{Rivals' share of the decline} & Baseline (anchor: price-unchanged versions $\geq$ 26 weeks; 28 versions) & 84\% [63, 112]$^\ast$\\
 & \quad anchor thresholds 12 to 45 weeks & 76--90\%\\
 & \quad anchor excluding announced, losing-host and slowing versions & 81\%\\
 & \quad all price-unchanged versions $\geq$ 8 weeks & 76\%\\
 & \quad frontier as top decile of the week's entrants (34 versions) & 68\%\\
 & \quad nested logit, developer nests, $\hat\sigma=0.21$ (rivals / own successors) & 79\% / 36\%\\
 & \quad nested logit across the 95\% interval for $\sigma$ (rivals / own successors) & 70--86\% / 15--62\%\\
\midrule
\emph{Price coefficient} $\lambda$ & Window IV, baseline / OLS & 0.69 (0.19) / 0.66 (0.10)\\
 & \quad same construction in 26 eight-week windows, median [10th, 90th pct.] & 0.68 [0.49, 0.85]\\
 & \quad with gap trend controls (first-stage F 3.3) & uninformative\\
 & \quad cost shifter, developer effects / version effects (F 52 / 0.7) & 0.52 (0.31) / uninformative\\
 & Nested logit, developer nests & 0.94 (0.22)\\
 & Own list-price changes, adjustment path $\lambda_A$ & 0.87 [0.36, 1.49]\\
 & Lerner at median margin, single version / developer portfolio & 1.39 / 1.87\\
\midrule
\emph{Capability coefficient} $\theta$ & Baseline (IRT, archived weeks, window IV; Table~\ref{tab:demand}, col.\ 2) & 2.74 (0.37)\\
 & \quad developer-by-quarter fixed effects & 2.88 (0.38)\\
 & \quad capability from knowledge items / agentic items only (same-sample baseline 2.75 / 2.86) & 2.48 / 2.00\\
 & \quad all-item capability instrumented by agentic-item score (same-sample baseline 2.86) & 3.07 (0.46)\\
 & \quad item-response capability instrumented by the chained index & 2.78 (0.41)\\
 & \quad Oster bound, $\delta=1$, $R_{\max}=1.3R^2$ (controlled OLS 2.69) & 2.86\\
 & \quad Epoch Capabilities Index vs IRT, matched versions & 2.05 vs 1.85\\
\midrule
\emph{Red Queen formula} & Slope of measured on predicted decay across 12 uses, in sample & 0.96 [0.77, 1.10]\\
 & \quad out of sample ($\theta_c$ to May 2025, decay of later entrants) & 1.03 [0.63, 1.38]\\
 & \quad fully ex ante (growth and shares also from the earlier period) & 1.38 [0.48, 2.22]\\
 & Rate of competitive obsolescence, IRT vs ECI (matched versions) & 2.27 vs 2.34\\
\midrule
\emph{Developer capability} & Decay of capability leads after own jumps, $\geq$ 0.15 SD / $\geq$ 0.25 SD / continuous & 0.74 / 0.47 / 0.67\\
\bottomrule
\end{tabular}}
\begin{minipage}{0.95\textwidth}\footnotesize\emph{Notes:} Standard errors in parentheses, 95 percent intervals (90 percent for developer capability) in brackets. $^\ast$The interval is from the bootstrap of the accounting without the fallback to younger anchors, whose point estimate is 82 percent; a share above one is possible because the version's own terms can be negative. $\theta$ rows are per standard deviation of the item-response scale, the units of Table~\ref{tab:demand}; the knowledge-item and agentic-item measures are standardized across versions in the estimation sample and expressed in the same units, and the Epoch row is per standard deviation of each measure on the matched versions. Oster: adding the observable controls raises $\theta$ (2.57 to 2.69 by OLS), so the $\delta$ that would set it to zero is negative. Sources: Sections~\ref{sec:iv}, \ref{sec:decomp}, \ref{sec:byuse}, \ref{sec:lambdas} and \ref{sec:whichrate}; Appendices~\ref{app:validation}, \ref{app:eci} and \ref{app:capshocks}.\end{minipage}
\end{table}

Three conclusions follow. First, the age profile of Fact~1 is not sensitive to how the frontier is defined or how long a version must be observed: the half-life of the median frontier version's revenue share is between 10 and 14 weeks and the within-version decay between 0.064 and 0.074 log points per week across the variants, and dropping the data-API weeks changes nothing. Second, the demand estimates are stable across capability measures and samples. The main specification measures capability on the item-response scale of Appendix~\ref{app:irt} and uses the archived weeks; on the chained index, on the full sample, excluding DeepSeek's own versions, or from 2025 onward, $\theta$ ranges from 2.50 to 3.00 and $\lambda$ from 0.66 to 0.84, and adding the recent-rival-capability instrument to the DeepSeek-window instrument passes the Sargan test (statistic 0.06 on one degree of freedom) with estimates unchanged. The placebo instrument, built by the same recipe for February to April 2024, is not inert in the first stage (F 17.5; $\lambda=1.46$, standard error 0.64), and Table~\ref{tab:ivdiag} shows why: the recipe gives a strong first stage in every eight-week window of the panel and a coefficient between 0.49 and 0.85 (10th and 90th percentiles across the 26 windows), and it loses its first stage once the pre-event gap is interacted with a trend and with age. The instrument therefore measures the conditional cross-sectional price gradient, not the response to the DeepSeek event, and the paper treats it so (Section~\ref{sec:iv}). Third, the attribution of about four fifths of first-year depreciation to newer rival versions is stable across frontier cutoffs, minimum lives and the exclusion of data-API weeks (74 to 85 percent; 68 percent with the frontier defined among the week's entrants), and across anchor thresholds between 12 and 45 weeks (76 to 90 percent, Figure~\ref{fig:anchor}); anchoring on all price-unchanged versions older than eight weeks gives 76 percent, and the version-bootstrap interval at the baseline anchor is 63 to 112 percent. Only anchors of a year or more give lower rival shares, and at those thresholds no anchor exists in the panel's first year. Own successors account for 9 to 14 percent across the variants and price channels for essentially nothing. In Panel B the rows labelled chained index measure capability with the chain-linked composite index of Appendix~\ref{app:data}, an alternative to the item-response scale. In Panel C the rows for minimum lives of 20 and 40 weeks coincide with the baseline because they leave the 28 frontier versions observed at age 52 unchanged.

\paragraph{Period and selection in the category estimates.} The category estimates on all weeks are lower than on the archived tables, with the same ordering and the equality of the $\theta_c$ again rejected. The difference has two sources, and we can separate them. The first is the period. Splitting the main sample of Table~\ref{tab:demand} at May 2025, the platform-level specification gives $\theta=3.19$ (0.68) before and 2.56 (0.30) after by OLS, with $\lambda$ unchanged at 0.68 and 0.66. The second is the selection of product-weeks into the category panel, which in the data-API weeks contains only the fifty largest models of each category: the platform-level specification estimated on the category panel's product-weeks gives 2.97 (0.68) on the archived period, 1.84 (0.29) on the same later archived weeks and 0.71 (0.32) on the data-API weeks after May 2026, in which the outside option is also measured differently. Restricting the data-API weeks to versions itemized in every week of their presence does not remove the selection, because the versions that are always itemized are the large ones. The category estimates on the archived weeks, where every model is itemized, are free of the second source and are the ones we use; the first source is a genuine change in the market over the sample. The decline is not a compression of the capability scale, since the dispersion of capability across traded versions rose from 0.54 to 0.77 standard deviations between the periods while the dispersion within the top quartile was unchanged, and it is not a change in price sensitivity. One explanation consistent with the data is a change in the composition of demand. The later period is dominated by agentic coding workloads whose buyers concentrate on a few frontier versions, and the capability dispersion among the versions that carry the traffic did fall: the standard deviation of capability among the five largest paid versions each week is 0.40 before May 2025 and 0.28 after, and in programming the gap between the three largest versions' mean capability and the category median fell from 0.37 to 0.22 standard deviations. A flatter share response across all traded versions is what one expects when the traffic sits on versions that differ little. We offer this as the likely reason, not as an established one. The within-period ordering across uses, which is the object of this section, is the same in both.

\begin{table}[htbp]
\centering
\caption{Demand estimates with the alternative instruments (item-response capability, archived weeks)}\label{tab:demandweak}
\scriptsize
\resizebox{\textwidth}{!}{%
\begin{tabular}{lccccc}
\toprule
 & OLS & IV: rival char. & IV: hosts & IV: window & IV: all\\
\midrule
Capability $\theta$ (per SD) & 2.693 (0.303) & 3.590 (0.862) & -77.183 (290.318) & 2.742 (0.370) & 3.092 (0.493)\\
Log price $\lambda$ & 0.663 (0.100) & 1.248 (0.618) & -51.496 (189.884) & 0.695 (0.190) & 0.923 (0.346)\\
Log output speed & 0.174 &  &  &  & \\
Log context window & 0.216 &  &  &  & \\
Reasoning model & -1.952 &  &  &  & \\
Open weights & 0.559 &  &  &  & \\
First-stage F & & 7.3 & 0.1 & 102.8 & 32.1\\
Sargan over-identification statistic & & & & & 1724.5\\
Within-version $\lambda$ (version FE) & 0.360 (0.216) & & & -1.701 (2.523) & \\
Anderson--Rubin 95\% set for $\lambda$ & & & & [0.320, 1.070] & \\
Observations & 11004 & 11004 & 11004 & 11004 & 11004\\
\bottomrule

\end{tabular}}
\begin{minipage}{0.95\textwidth}\footnotesize\emph{Notes:} Specification of Table~\ref{tab:demand}, column (2), with each set of instruments. Point estimates of the controls are reported for OLS only. The rival-characteristics and host-count instruments are weak on this sample (first-stage $F$ 7.3 and 0.1), so their columns are uninformative; the host-count column's price coefficient is not meaningful. With all instruments the over-identification test rejects decisively (Sargan statistic in the table): the instruments do not identify a common coefficient, which is why the paper does not use them. The within-version row replaces developer with version fixed effects (OLS, and IV with the window instrument). The Anderson--Rubin set inverts the reduced-form test of the window instrument over a grid for $\lambda$.\end{minipage}
\end{table}

\begin{table}[htbp]
\centering
\caption{First stages for the log price (full sample, chained-index controls)}\label{tab:firststage}
\footnotesize
\resizebox{\textwidth}{!}{%
\begin{tabular}{lcccc}
\toprule
Excluded instruments & Rival char. & Hosts & DeepSeek event & All\\\midrule
Rivals within 0.25 log-capability & 0.010 (0.005) &  &  & 0.010 (0.005)\\
Mean capability of recent rival entrants & -0.854 (0.419) &  &  & -0.922 (0.400)\\
Log number of hosts &  & 0.060 (0.088) &  & 0.058 (0.086)\\
DeepSeek window $\times$ pre-event price gap &  &  & 0.445 (0.045) & 0.439 (0.045)\\
First-stage F (clustered) & 3.8 & 0.5 & 98.5 & 29.1\\
Observations & 12856 & 12856 & 12856 & 12856\\
\bottomrule

\end{tabular}}
\begin{minipage}{0.95\textwidth}\footnotesize\emph{Notes:} Dependent variable: log list price. All columns include developer and week fixed effects and the controls of Table~\ref{tab:demand}. Standard errors clustered by version in parentheses. The F-statistic is the clustered Wald statistic on the excluded instruments divided by their number.\end{minipage}
\end{table}

\begin{table}[htbp]
\centering
\caption{Calibrated parameters}\label{tab:calib}
\scriptsize
\begin{tabular}{lp{3.6cm}p{1.6cm}p{5.0cm}}
\toprule
Parameter & Meaning & Value & Source\\
\midrule
$\theta$ & capability coefficient (per SD) & 2.74 & Table~\ref{tab:demand}, main specification\\
$\lambda$ & log-price coefficient, cross-sectional & 0.69 & Table~\ref{tab:demand}, main specification\\
$\lambda$ & from own price changes; pricing & 0.87; 1.4--2.1 & adjustment path (Section~\ref{sec:lambdas}); Lerner at the median margin, single version / portfolio (1.39 / 1.87); hedonic gradient (1.62--2.05)\\
$\varepsilon_\chi$ & intensive-margin elasticity & 0 & Section~\ref{sec:intensive}: not identified with a sign\\
$g_A$ & frontier growth, SD per year & 1.03 & maximum item-response capability among traded versions\\
$N$ & frontier firms & 3 & developers with $\ge 5$ percent of paid revenue, last quarter\\
$s_0$ & outside-option share & 0.19 & archived panel, last quarter\\
$\varpi$ & gross margin & 0.4--0.8 & gross margin over the cheapest third-party host of the same weights, $1-p^{\min}_{\text{host}}/p_{\text{developer}}$: median 0.73, IQR 0.20--0.83 (Appendix~\ref{app:facts})\\
$\Delta\ln\kappa_0$ & open-weight gap closure & 0.36 SD & archived panel, first vs.\ last quarter\\
\bottomrule
\end{tabular}
\end{table}

\begin{table}[htbp]
\centering
\caption{Diagnostics for the window instrument and the cost shifter (item-response capability, archived weeks)}\label{tab:ivdiag}
\scriptsize
\resizebox{\textwidth}{!}{%
\begin{tabular}{lcccr}
\toprule
Specification & $\theta$ & $\lambda$ & First-stage F & Obs.\\
\midrule
IV, baseline & 2.74 (0.37) & 0.69 (0.19) & 103 & 11,004\\
IV, control: window $\times$ capability & 2.85 (0.38) & 0.77 (0.19) & 103 & 11,004\\
IV, controls: window $\times$ capability, reasoning, open weights & 2.79 (0.37) & 0.73 (0.18) & 122 & 11,004\\
IV, non-reasoning versions only & 2.69 (0.40) & 0.73 (0.20) & 88 & 9,395\\
IV, excluding DeepSeek's versions & 2.46 (0.41) & 0.64 (0.19) & 102 & 10,343\\
OLS, all archived weeks & 2.69 (0.30) & 0.66 (0.10) & & 11,004\\
OLS, excluding the window weeks & 2.68 (0.30) & 0.66 (0.10) & & 10,425\\
OLS, window weeks only & 3.09 (0.72) & 0.70 (0.13) & & 579\\
\midrule
IV, controls: gap $\times$ trend, gap $\times$ age, launch price $\times$ age & 2.85 & 0.79 (1.85) & 3.3 & 11,004\\
IV, same controls, version effects & & 0.05 (1.76) & 8.0 & 11,004\\
IV, placebo window February--April 2024 & 3.54 & 1.22 (0.59) & 19.8 & 11,004\\
IV, placebo window, with the trend controls & & uninformative & 0.1 & 11,004\\
IV, each of 26 eight-week windows: median [10th, 90th pct.] & & 0.68 [0.49, 0.85] & 87 (all $>10$) & 11,004\\
IV, cost shifter, developer effects & 2.42 & 0.52 (0.31) & 51.6 & 5,115\\
IV, cost shifter, version effects & & 5.6 (6.2) & 0.7 & 5,115\\
IV, window instrument on the cost-shifter sample & 2.80 & 0.75 (0.30) & 53.7 & 5,115\\
\midrule
First stage, developer and week effects, no controls & \multicolumn{3}{l}{0.62 (0.06)} & \\
First stage, version and week effects & \multicolumn{3}{l}{0.028 (0.011)} & \\
\bottomrule
\end{tabular}}
\begin{minipage}{0.9\textwidth}\footnotesize\emph{Notes:} Dependent variable $\ln s_{it}-\ln s_{0t}$; developer and week fixed effects and the controls of Table~\ref{tab:demand}; standard errors clustered by version. The window is January 20 to March 16, 2025; the instrument is the window indicator times the version's pre-window log price gap to DeepSeek V3. Trend controls: the gap times a linear week trend and times age, and the log launch price times age. Rolling windows: the same construction for every eight-week window starting every four weeks, with the gap measured in the week before each window. Cost shifter: mean log price of the cheapest third-party host in the same week over other developers' models in the version's quartile of parameter count (active parameters where reported), on versions with a reported size. The last two rows report the coefficient of the log price on the window instrument with no controls.\end{minipage}
\end{table}

\paragraph{The price coefficient and substitution within developers.} Figure~\ref{fig:pricing} shows the two diagnostics of Section~\ref{sec:lambdas}. Panel A applies the construction of the window instrument to every eight-week window in the archived weeks: the coefficient is about the same in every window, and the DeepSeek window is not special, which is what one expects if the instrument picks up the cross-sectional relation between a version's price and its pre-window price. Panel B plots the responses of closed-weight versions' shares to their own list-price changes, expressed as implied coefficients, against the path $\lambda_A(1-\rho^{k+1})$ that slow migration predicts; the fitted $\lambda_A$ weights each horizon by its inverse sampling variance and the interval adds the sampling uncertainty in $\rho$ by simulation. Table~\ref{tab:nested} reports the nested-logit estimates of Section~\ref{sec:decres} and the accounting for a grid of $\sigma$. The estimation sample is the version-weeks of developers with at least two versions on sale, on which the nest instruments are defined. In the accounting, the developer's inclusive value $D_f$ and the market's $\check D$ are both updated from observed shares and $\sigma$, the anchor normalization is the logit one applied to $\check D$, and the identity holds to machine precision in every week.

\begin{figure}[htbp]
\centering
\includegraphics[width=\textwidth]{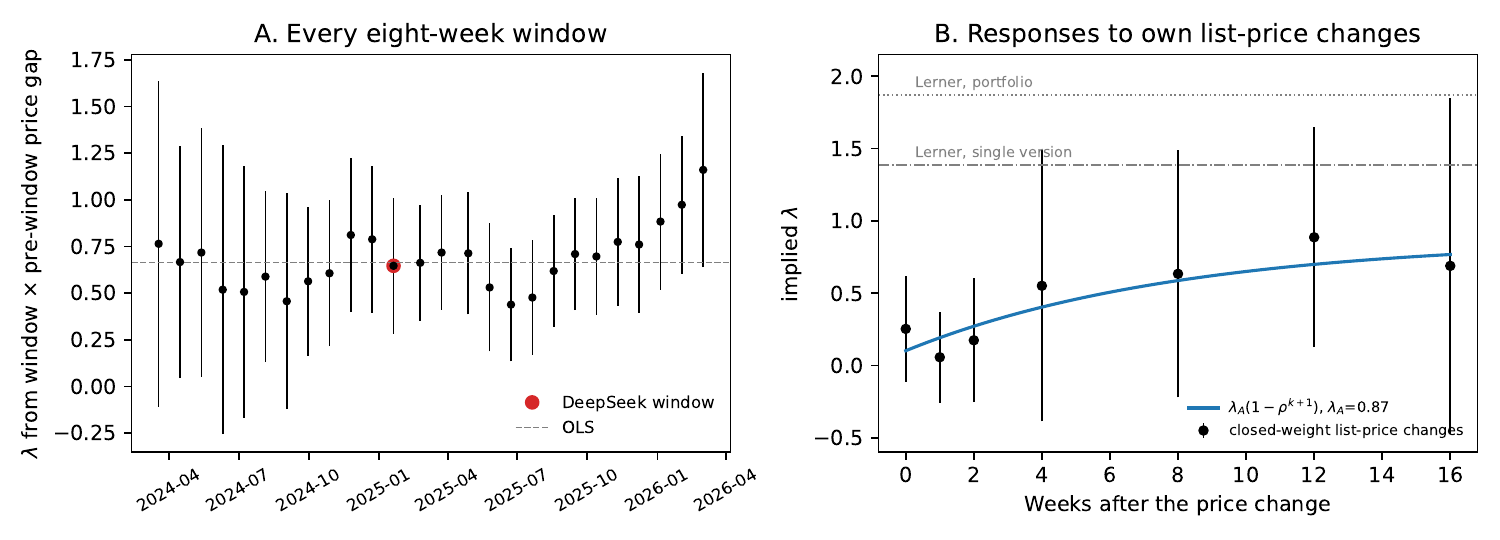}
\caption{The window instrument in every window, and responses to own price changes}\label{fig:pricing}
\begin{minipage}{0.95\textwidth}\footnotesize\emph{Notes:} (A) Price coefficient from the window instrument built for each eight-week window starting every four weeks, archived weeks, with 95 percent intervals; the red point is the window that starts on January 20, 2025, and the dashed line is OLS. (B) Implied coefficients $-b_k/(1-\bar s)$ from local projections of closed-weight versions' log shares on their own list-price changes (Appendix~\ref{app:priceevents}), with 95 percent intervals; the curve is $\lambda_A(1-\rho^{k+1})$ with $\rho=0.88$ and $\lambda_A=0.87$; horizontal lines are the coefficients that the Lerner condition implies at the median margin with a single version's share and with the developer's inside share.\end{minipage}
\end{figure}

\begin{table}[htbp]
\centering
\caption{Nested logit with developer nests}\label{tab:nested}
\footnotesize
\begin{tabular}{lccccc}
\toprule
\multicolumn{6}{l}{\emph{A. Estimates}}\\
 & $\sigma$ & $\theta$ & $\lambda$ & First-stage F, $\ln s_{i|f}$ & Obs.\\
\midrule
Nested logit, IV & 0.21 (0.09) & 3.32 (0.56) & 0.94 (0.22) & 14.6 & 5,491\\
\midrule
\multicolumn{6}{l}{\emph{B. Shares of the mean first-year decline, frontier versions at age 52}}\\
$\sigma$ & 0 & 0.20 & 0.21 & 0.40 & 0.60\\
\midrule
Newer rival versions & 0.87 & 0.80 & 0.79 & 0.69 & 0.56\\
Own newer versions & 0.12 & 0.34 & 0.36 & 0.64 & 1.05\\
Outside option & 0.06 & 0.05 & 0.05 & 0.03 & $-0.01$\\
Own price and residual & $-0.03$ & $-0.23$ & $-0.25$ & $-0.45$ & $-0.73$\\
Older versions, own and rival & 0.01 & 0.06 & 0.06 & 0.11 & 0.15\\
Weeks not observed & $-0.02$ & $-0.02$ & $-0.02$ & $-0.02$ & $-0.02$\\
\bottomrule
\end{tabular}
\begin{minipage}{0.92\textwidth}\footnotesize\emph{Notes:} Panel A: \eqref{eq:nested} with developer and week fixed effects and the controls of Table~\ref{tab:demand}; the log price and the log within-developer share are instrumented with the window instrument, the log number of the developer's versions on sale, and the mean and maximum capability of its other versions; standard errors clustered by version. Panel B: exact nested accounting for the 28 frontier versions of Table~\ref{tab:decomp} (mean total $-2.41$ log points); columns add to one up to rounding. At $\sigma=0$ the accounting is the logit one; the differences of two to three points from Table~\ref{tab:decomp} come from computing the anchor normalization on the inclusive value directly rather than from rescaled contributions.\end{minipage}
\end{table}

\begin{table}[htbp]
\centering
\caption{Robustness of the main results}\label{tab:robust}
\scriptsize
\resizebox{\textwidth}{!}{\begin{tabular}{lcccc}
\toprule
\multicolumn{5}{l}{\emph{A. Fact 1: age-revenue profile of frontier versions}}\\
Variant & Versions & Half-life (weeks) & Weekly decay & \\\midrule
baseline & 62 & 12 & $-$0.070 & \\
top 5\% at entry & 42 & 10 & $-$0.074 & \\
top 20\% at entry & 91 & 12 & $-$0.064 & \\
min life 20 weeks & 74 & 10 & $-$0.071 & \\
min life 40 weeks & 46 & 14 & $-$0.066 & \\
archived weeks only & 53 & 12 & $-$0.069 & \\
\midrule
\multicolumn{5}{l}{\emph{B. Demand: window-instrumented specification}}\\
Variant & $\theta$ (s.e.) & $\lambda$ (s.e.) & First-stage F & N\\\midrule
Chained index, full sample (alternative capability measure) & 2.83 (0.41) & 0.80 (0.19) & 98.5 & 12856\\
Placebo: same window in 2024 & 3.81 (0.80) & 1.46 (0.64) & 17.5 & 12856\\
Excluding DeepSeek's own versions & 2.50 (0.44) & 0.74 (0.19) & 96.7 & 12059\\
2025 onward & 2.86 (0.39) & 0.82 (0.20) & 106.6 & 10852\\
Main: item-response scale, archived weeks & 2.74 (0.37) & 0.69 (0.19) & 102.8 & 11004\\
\quad + recent-rival instrument (Sargan in F column) & 2.75 (0.36) & 0.70 (0.19) & 56.9 (J=0.06) & 11004\\
Item-response scale, full sample & 2.54 (0.35) & 0.66 (0.19) & 112.8 & 11962\\
Chained index, archived weeks & 3.00 (0.45) & 0.84 (0.19) & 87.5 & 11815\\
\midrule
\multicolumn{5}{l}{\emph{C. Decomposition at age 52, frontier versions (shares of the total decline, means)}}\\
Variant & Total (log, median) & Younger rivals & Own successors & Own residual\\\midrule
Baseline (anchor: age $\ge$ 26) & $-$2.36 & 0.84 & 0.14 & 0.01\\
Anchor: age $\ge$ 52 (with fallback) & $-$2.36 & 0.56 & 0.11 & 0.56\\
Anchor: all price-unchanged versions & $-$2.36 & 0.76 & 0.12 & 0.08\\
Anchor excluding announced, losing-host and slowing versions & $-$2.36 & 0.81 & 0.11 & 0.16\\
Frontier: top 5\% at entry & $-$2.61 & 0.74 & 0.12 & 0.17\\
Frontier: top 20\% at entry & $-$2.19 & 0.85 & 0.12 & $-$0.09\\
Frontier: top decile of the week's entrants & $-$3.29 & 0.68 & 0.09 & 0.06\\
Minimum life 20 weeks & $-$2.36 & 0.84 & 0.14 & 0.01\\
Minimum life 40 weeks & $-$2.36 & 0.84 & 0.14 & 0.01\\
Archived weeks only & $-$2.58 & 0.77 & 0.14 & 0.10\\
\bottomrule
\end{tabular}
}
\begin{minipage}{0.95\textwidth}\footnotesize\emph{Notes:} Panel A: frontier versions' revenue-share profile normalized at age four; weekly decay is the within-version coefficient of log share on age over ages 4 to 56. Panel B: the window-instrumented demand specification; in the item-response row $\theta$ is per standard deviation of capability rather than per log point of the chained index; the placebo applies the same instrument construction to February to April 2024 (Table~\ref{tab:ivdiag}). Panel C: mean shares of the first-year decline in log token share, frontier versions at age 52.\end{minipage}
\end{table}

\section{Additional results for Sections~\ref{sec:demand} and~\ref{sec:decomp}}\label{app:decomp}

\paragraph{Windows for the use-level test.} In Section~\ref{sec:byuse} the interval for the slope is from redrawing category decay rates and capability coefficients from their sampling distributions (0.82 to 1.18 without translation); the median ratio has a 95 percent interval from 0.72 to 1.23 and an interquartile range across categories from 0.73 to 1.21. A fully ex ante version of the out-of-sample test, with growth rates and shares also taken from the earlier period, is uninformative: one year of category growth rates is too noisy (it predicts negative decay for finance), and the slope is 1.38 with an interval from 0.48 to 2.22 over ten categories. In the uses with fast depreciation a frontier version's share peaks soon after launch: in programming its mean category share is 0.043 at ages four to eight and 0.005 at ages forty to fifty-six. In translation it is 0.016 and then 0.030, and in legal 0.021 and then 0.028; in 2026 the token-weighted mean age of the versions serving translation is 25 weeks and 12 percent of its tokens are served by versions more than a year old, against 9 weeks and 0.2 percent in programming. Role-play (ratio 0.62) has the same profile, with a mean age in use of 24 weeks, and its free tiers carry only 2 percent of its matched tokens; health's ratio is 0.57. The last column of Table~\ref{tab:category} bounds the other way, measuring decay from each version's peak share within its first thirty weeks over the following year, a window steepened by selection on the noisy maximum: the median ratio rises to 1.6 and translation's to 0.35, and in every use but translation and trivia both windows put the decay within a factor of two of the prediction.

\subsection{The intensive margin}\label{sec:intensive}

The elasticity $\varepsilon_\chi$ governs how much a developer's token use per unit of output responds to its supplier's price. Tokens per request, observed for every version-week in the archived panel, proxy the complexity of the workloads a version serves, and we regress their log on the log price with the controls and fixed effects of \eqref{eq:est}. The elasticity is not identified with a sign: $-0.12$ (standard error 0.06) with the DeepSeek-window instrument, $-0.11$ (0.06) adding the recent-rival instrument, and $+0.19$ (0.11) with all four instruments, a set the Sargan test rejects (134 on three degrees of freedom); OLS gives $-0.04$. Capability enters with a coefficient of about 1.0, the selection of complex workloads onto capable versions, which is what makes a version-level price instrument unable to isolate the architectural response. The aggregate response is visible, tokens per request among coding agents rising from 41 thousand in March 2026 to 118 thousand in October 2026 as the quality-adjusted price fell, but we set $\varepsilon_\chi=0$ in what follows, which makes the markups of \eqref{eq:lerner} upper bounds.

\subsection{Age, cohort, and calendar time}\label{sec:apc}

An age profile estimated with version and calendar-week fixed effects is not identified, because age equals calendar time minus cohort. Fact~1 sidestepped the problem by reporting an unconditional profile in event time. The accounting resolves it: the model says that the calendar-time effect on a fixed version's share is not a free parameter but the cumulative share gained by newer rivals and own successors and the cumulative price changes of older versions. Regressing the cumulative log share since age four on age alone, with version fixed effects, gives $-0.061$ log points per week (standard error 0.003); adding the cumulative accounting terms lowers the age coefficient to $-0.035$ (0.006). The attenuation is partial because the weekly accounting terms are noisy regressors; the accounting itself, which imposes the identity rather than estimating it, is the cleaner statement. We do not interpret the remaining age coefficient as deterioration, since the asset cannot deteriorate.

\subsection{The response to individual releases}\label{sec:event}

The accounting attributes share by who gains it in a given week. As a check that does not depend on the anchor, Figure~\ref{fig:es} plots the average path of an incumbent's log token share around comparable rival releases (versions from other developers with capability within 0.25 log points below or above the incumbent's) and around own-successor releases, relative to the week before the release, with version-clustered confidence bands. In the twelve weeks after a comparable rival release an incumbent's log share falls by 0.84 on average, across 9,324 version-events; after an own-successor release it falls by 0.91 across 1,405 events. The paths are smooth rather than step-like, which is the adoption ramp measured by $\rho$ above. They also slope downward before the release, at about 0.05 per week, because comparable releases arrive more than once every two weeks and the pre-period of one event is the post-period of the previous one. Calendar windows in which no frontier incumbent faces a comparable release within two weeks are so rare, three in the panel, that they cannot serve as a control group, which is Fact~2 restated; the 310 windows of Fact~2 are version-specific, defined relative to each incumbent's own capability, and mostly belong to incumbents far from the frontier.

A dose-response version of the event study, in which the incumbent's eight-week log-share change is regressed on the entrant's capability gap with incumbent and release-week fixed effects, finds no gradient within the comparable band for rival releases (0.005 per log point, standard error 0.011, over 18,347 incumbent-release pairs), and none either in the entrant's realized share eight weeks after entry when it is added as a regressor ($-0.001$ per log point, standard error 0.001): an incumbent's loss over the window is unrelated to how good the entrant is and to how much share it takes. Under the logit structure the incumbent's loss should scale with the entrant's share, so this is not evidence for IIA; nor is it evidence against it, because the eight-week windows contain, on average, three comparable releases, and the loss in any one window is the sum of the treadmill's effect over all of them. The within-window event design cannot separate the contributions of overlapping entrants, which is why the accounting of Section~\ref{sec:terms}, which assigns gains to every gainer in every week, is the instrument of attribution and the event paths are its description. Own-successor releases take 0.22 log points more than rival releases at the same gap and share (standard error 0.06). The capability coefficient $\theta$ is identified from which versions developers choose in the cross-section, as in Section~\ref{sec:demand}, rather than from the size of the loss per release.

\begin{figure}[htbp]
\centering
\includegraphics[width=0.85\textwidth]{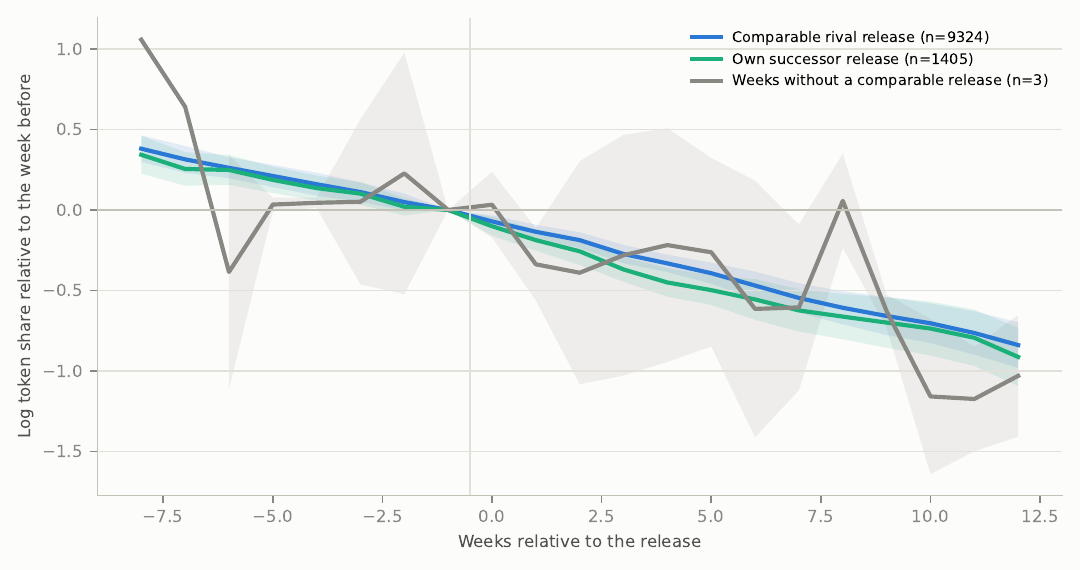}
\caption{Incumbent versions' log share around rival and own-successor releases}\label{fig:es}
\begin{minipage}{0.9\textwidth}\footnotesize\emph{Notes:} Mean of $\ln s_{i,t_e+h}-\ln s_{i,t_e-1}$ over version-events, with 95 percent bands from version-clustered standard errors. Events are releases at least four weeks after the incumbent's own entry. The placebo series averages the three windows in the panel with no comparable release within two weeks.\end{minipage}
\end{figure}

\section{The release race and open-weight diffusion}\label{app:race}

\subsection{The strategic term: the release race}\label{sec:race}\label{sec:mpe}

The comparison holds rivals' investment fixed. In step-by-step models a rival's response can go either way: a firm that falls behind may invest more to escape or less because the prize has shrunk \citep{AghionHarrisHowittVickers2001}. Here the response has a definite structure. With the estimated $\theta$ and a symmetric inside share of 0.27, each firm's profit is locally convex in its own capability: its second derivative, $\varpi\theta^2s(1-s)(1-2s)$, is positive whenever $s<1/2$, and a firm that moves 0.3 standard deviations ahead of symmetric rivals raises its inside share from 0.27 to \sharelead, while one that falls 0.3 behind drops to \sharelag. A firm facing rivals that keep pace therefore prefers to jump ahead or to fall back, never to keep pace, so the continuous investment game has no smooth symmetric equilibrium (Appendix~\ref{app:curvature}). The race is a preemption game.

The discrete-release version of the game of Section~\ref{sec:releases} turns this into predictions (Appendix~\ref{app:curvature}, Table~\ref{tab:release}, Figure~\ref{fig:release}). Two symmetric developers face a frontier that advances by $g_A$ a year. Each month each chooses between releasing a version at the frontier, at a fixed cost $F$, and letting its current version age; flow profit follows the estimated logit, and strategies are pure Markov strategies in the two gaps to the frontier, as in \citet{EricsonPakes1995} and \citet{PakesMcGuire1994}. Frontier training runs cost of the order of \$0.1 to \$1 billion \citep{CottierEtAl2024} and the largest developer reported annualized revenue of \$10 billion in mid-2025 \citep{CNBC2025OpenAI}. Measured against the frontier developers' combined revenue, which is the market the symmetric game describes, a release costs between a few tenths of a percent and a few percent of annual revenue. Three results follow. First, a pure Markov equilibrium exists for release costs of 0.4 to 0.8 percent of annual revenue, where each firm releases four to six times a year as the leading developers in the panel do, and again from 4 percent upward; between 1 and 3 percent the best responses cycle. Second, where the equilibrium exists it has alternating leads: the firms never release in the same month; at release costs of 0.4 to 0.8 percent each releases one to two months after its rival and the lead changes eight to twelve times a year, and at 4 to 6 percent each releases five to six months after its rival and the lead changes two to three times a year. In the panel the leading developer's median frontier release comes one week after a rival's (see the paragraph on leadership below). Third, the ordering of release frequencies matches the direct-effect comparison: the planner releases more often than the equilibrium at every cost, and the cartel that maximizes joint profit releases less often. The gap between equilibrium and cartel is business stealing; the gap between planner and equilibrium is the unappropriated surplus, which is the larger. The planner's weight on developers' surplus is $1/\lambda$. The baseline game uses the cross-sectional coefficient, 0.69; with the pricing-consistent coefficient of 1.87, which weights that surplus less, the planner still releases more often than the equilibrium at every cost, six times a year against four at a cost of 0.6 percent and three against one and a half at 4 percent. The strategic term therefore does not reverse the direct comparison in the release game. What the convexity adds is that the private return to a lead is larger than a local linearization suggests, which favors the business-stealing side and is the reason the welfare statement above is about the direct effect.

\paragraph{Leadership and turnover.}\label{sec:flywheel} The panel shows persistent oligopoly without tipping. Three developers hold more than five percent of paid revenue throughout the last year of the archived sample, the revenue Herfindahl index across developers averages 0.35 in its last quarter, and although the identity of the developer with the largest weekly revenue changed eleven times, Anthropic led in 106 of 116 weeks with a share that oscillated between a third and two thirds. At the version level turnover is complete. The leader keeps its position by releasing: of the 30 frontier versions that entered the archived sample, 8 were Anthropic's, and the median one arrived one week after a rival's frontier release. Part of what the decomposition records as rivals' gains is therefore the leader's own response to rivals, a distinction the accounting cannot make.

\subsection{Open-weight diffusion}\label{sec:cf}

Over the archived sample the capability gap between the frontier and the best open-weight version fell from 0.76 to 0.39 standard deviations. Treating the 0.36 closure as a change in $\kappa_0$ raises the outside option's share from 0.19 to 0.39 through \eqref{eq:logit}, and the frontier pressure $\theta s_0g_A$ faced by every proprietary incumbent from 0.53 to 1.09 log points a year. Developers gain the increase in the log-sum, $\ln(1-s_0+s_0e^{\theta x})=0.28$ utils, worth 18 percent of annual paid revenue at the pricing-consistent coefficient and 49 percent at the cross-sectional one. Under static pricing the welfare ratio falls from 0.74 to 0.51: diffusion of open-weight capability raises incumbents' depreciation and widens the gap between the social and the private return to proprietary capability.

\begin{figure}[htbp]
\centering
\includegraphics[width=\textwidth]{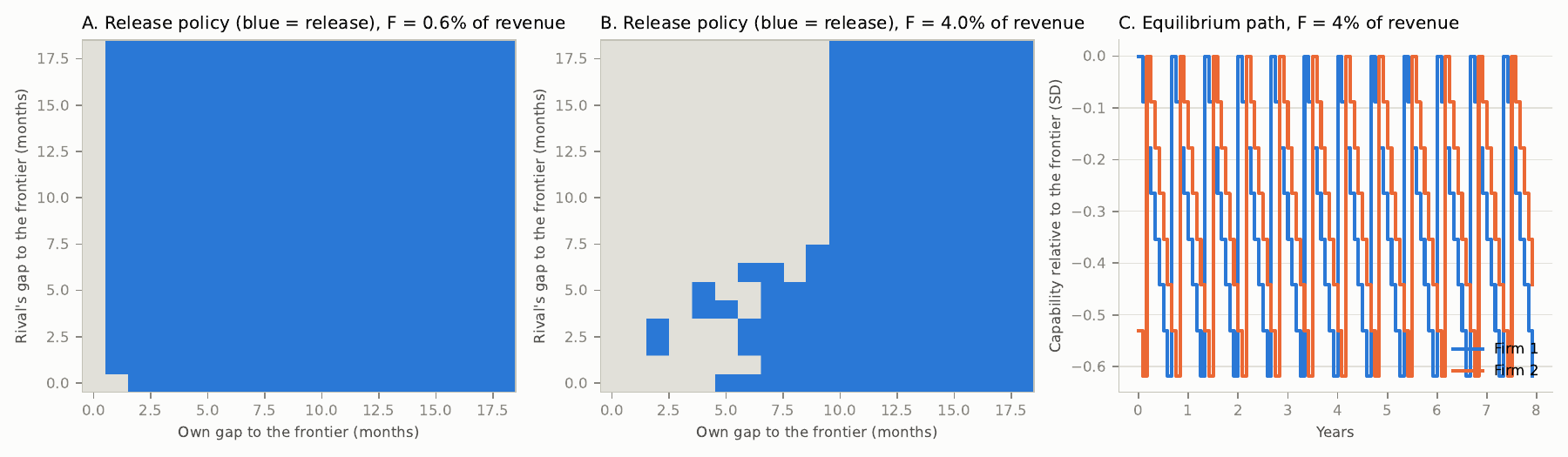}
\caption{The discrete release game}\label{fig:release}
\begin{minipage}{0.95\textwidth}\footnotesize\emph{Notes:} (A, B) Firm 1's equilibrium release policy as a function of its own and its rival's gap to the frontier, in months, at release costs of 0.6 and 4 percent of annual market revenue; the isolated cells in (B) are artifacts of the monthly grid. (C) Eight years of the simulated equilibrium path from the symmetric state at the 4 percent cost: each firm's capability relative to the frontier.\end{minipage}
\end{figure}

\section{Local curvature and the release game}\label{app:curvature}

At the symmetric point of the logit demand system with inside share $s$, firm $i$'s flow profit $\varpi s_i$ has own second derivative $\varpi\theta^2s(1-s)(1-2s)$ with respect to its capability, positive whenever $s<1/2$. With $\theta=2.74$ and $s=0.27$ it equals $\varpi\times0.68$, against a first derivative of $\varpi\times0.54$: profit is locally convex in own capability, a firm that moves 0.3 standard deviations ahead of symmetric rivals raises its inside share from 0.27 to \sharelead, and one that falls 0.3 behind sees it fall to \sharelag. Two computations confirm what this implies for a game with continuous investment. A linear-quadratic approximation around the symmetric point has no finite feedback Nash solution, because each player's Riccati equation has an indefinite state cost. Value-function iteration on the nonlinear problem with quadratic adjustment costs produces best responses that cycle at the symmetric point, because each firm's best response to rivals keeping pace is to jump ahead or to fall back; away from the symmetric point the iteration converges, and a firm 0.3 standard deviations behind invests almost nothing while one 0.3 ahead invests at the frontier rate. The structure is that of a preemption game, whose equilibria involve alternating leads, and the discrete game below has that form.
\paragraph{The discrete release game.} The state is the pair of gaps $(d_1,d_2)$ between each firm's current version and the research frontier, in months, on a grid from 0 to 36. Each month the frontier advances by $g_A/12$ standard deviations of capability ($g_A=1.06$ per year; the growth of the maximum item-response capability measured directly is 1.03); a firm that releases pays $F$ and resets its gap to zero, one that waits sees its gap grow by one. Flow profit is $\varpi s_i/12$ with $\varpi=0.6$, market revenue normalized to one per year and $s_i$ the logit share \eqref{eq:logit} at $\theta=2.74$, with the open-weight option's constant set so that $s_0=0.19$ when both firms are one month behind the frontier; the discount factor is 0.9 per year. Moves are simultaneous and strategies are pure Markov policies $\sigma_i(d_i,d_j)\in\{\text{wait},\text{release}\}$. Each firm's Bellman equation is solved by value iteration given the rival's policy and the policies are updated alternately until neither changes; the planner's problem is solved by the same iteration with the flow $\varpi(s_1+s_2)/12+\ln\big(\sum_j e^{\theta k_j}+A_0\big)/(12\lambda)$, which adds developers' surplus, and both releases chosen jointly. The planner's weight on developers' surplus, $1/\lambda$, uses the cross-sectional coefficient in the baseline; the last row of Table~\ref{tab:release} repeats the planner's problem with the pricing-consistent coefficient $\lambda=1.87$ of Section~\ref{sec:lambdas}.

\begin{table}[htbp]
\centering
\caption{Equilibrium of the discrete release game by release cost}\label{tab:release}
\scriptsize
\begin{tabular}{lcccccc}
\toprule
Release cost $F$ (\% of annual revenue) & 0.4 & 0.6 & 0.8 & 1--3 & 4 & 6\\
\midrule
Pure Markov equilibrium exists & yes & yes & yes & no (cycles) & yes & yes\\
Releases per firm-year & 6 & 4 & 4 & -- & 1.5 & 1.2\\
Simultaneous releases (share) & 0 & 0 & 0 & -- & 0 & 0\\
Lead changes per year & 12 & 8 & 8 & -- & 3 & 2.4\\
Median months from rival's release to own & 1 & 2 & 2 & -- & 6 & 5\\
Planner's releases per firm-year & 12 & 12 & 12 & 6 & 4 & 4\\
Planner's releases, $\lambda=1.87$ & 12 & 6 & 6 & 3 to 6 & 3 & 2.4\\
Cartel's releases per firm-year & 4 & 3 & 3 & 2 to 1.5 & 1.3 & 1.1\\
\bottomrule
\end{tabular}
\begin{minipage}{0.9\textwidth}\footnotesize\emph{Notes:} Statistics of the simulated equilibrium path from the symmetric state after a 20-year burn-in. The planner maximizes industry profit plus developers' surplus (the log-sum over $\lambda$); the cartel maximizes joint profit alone. Where no pure equilibrium exists the planner's and cartel's rows are still defined.\end{minipage}
\end{table}

\section{Additional public data}\label{app:v9}

\subsection{Developers' retirement schedules}\label{app:lifecycle}

The schedules come from OpenAI's deprecations page (live copy, November 2026), Anthropic's model-deprecations page (Internet Archive capture of November 13, 2026, since the live page could not be retrieved directly), Google's Gemini API deprecations page, and Mistral's table of deprecated and retired models \citep{OpenAIDeprecations2026,AnthropicDeprecations2026,GeminiDeprecations2026,MistralModels2026}. Each row gives a model or snapshot identifier, the date of the deprecation notice (not reported by Google), the shutdown date and a recommended replacement. Identifiers that refer to speech, image, search, embedding and long-context variants are dropped. A snapshot's release date is the date embedded in its identifier, Google's release-date column, or the panel version's release date. An API identifier is matched to a panel version when its name tokens equal those of one of the version's OpenRouter identifiers or exceed them only by a date or snapshot suffix within 75 days of the version's release. A version's retirement is that of its last-served snapshot. Google's models retired before November 2025, including the Gemini 1.5 family, no longer appear on its page and are missing.

\subsection{Application panel}\label{app:apps}

The Internet Archive holds 47,116 captures of openrouter.ai application pages. We keep captures of the page itself (not its client-side data requests) of at least 15 kB, one per application every 28 days or more, for applications with at least three such captures: 1,380 captures of 633 page addresses, which resolve to 277 applications. Each capture embeds the application's daily tokens for the preceding 30 days by its ten most-used models, with the rest pooled; where captures overlap, the latest capture's value is used. Days are aggregated to calendar weeks with at least five observed days. Model identifiers are mapped to panel versions with the paper's crosswalk; 11 percent of tokens go to identifiers outside the crosswalk and 16 percent are pooled as other models. Application groups come from the category that the marketplace embeds in its application rankings (51 applications), a manual classification of 44 further applications from their names and descriptions, and the page's own category tags for the rest.

\subsection{Firm finances and training costs}\label{app:epoch}

Annualized revenue, research and inference compute spending and token-volume disclosures are from Epoch AI's compilation of company statements and press reports \citep{EpochCompanies2026}. Annual revenue is the calendar-year integral of the log-linearly interpolated annualized revenue. Training compute and the cost of final training runs are from Epoch AI's model database \citep{EpochModels2026}; where it gives compute but no cost, the cost is compute times the median dollar cost per FLOP of large training runs in the same year, held non-increasing over time. Developer revenue is allocated to versions in proportion to their shares of the developer's closed-weight revenue on the platform, since the platform revenue of open-weight versions accrues mostly to third-party hosts; for 52 versions observed from launch the median lifetime revenue is 0.41 years of the age-four run-rate (interquartile range 0.23 to 0.84). For the ten closed-weight versions with an estimate of the final training run, mostly of low confidence, lifetime revenue at a 60 percent margin is 0.1 to 28 times that cost (median 2.6); each of the two firms' research compute in 2025 was more than ten times the estimated final-run cost of the versions it released that year.

\begin{table}[htbp]\centering\small
\caption{Model capital accounting for two developers, 2025: an illustration (billions of dollars)}\label{tab:firmacct}
\resizebox{\textwidth}{!}{\begin{tabular}{lccccccc}
\toprule
 & \multicolumn{4}{c}{OpenAI} & \multicolumn{3}{c}{Anthropic}\\
\cmidrule(lr){2-5}\cmidrule(lr){6-8}
Version share $\phi$, capability-capital rate $\delta_K$ & Depr. & Surplus & Surplus$^\dagger$ & Net/gross & Depr. & Surplus & Net/gross\\
\midrule
$\phi=1$ (all spending builds versions) & 6.8 & $-3.2$ & $+0.3$ & 0.19 & 3.1 & $-1.6$ & 0.24\\
$\phi=0.75$, $\delta_K=0.65$ & 6.0 & $-2.5$ & $+1.0$ & 0.27 & 2.7 & $-1.2$ & 0.33\\
$\phi=0.5$, $\delta_K=0.65$ & 5.3 & $-1.8$ & $+1.7$ & 0.36 & 2.4 & $-0.8$ & 0.42\\
$\phi=0.1$, $\delta_K=1.0$ & 5.0 & $-1.4$ & $+2.1$ & 0.40 & 2.2 & $-0.6$ & 0.48\\
$\phi=0.1$, $\delta_K=0.65$ & 4.2 & $-0.6$ & $+2.9$ & 0.49 & 1.8 & $-0.2$ & 0.57\\
$\phi=0.1$, $\delta_K=0.25$ & 2.6 & $+1.0$ & $+4.5$ & 0.69 & 1.0 & $+0.5$ & 0.74\\
\midrule
All at the BEA software rate (0.40) & 2.9 & $+0.6$ & $+4.1$ & 0.65 & 1.2 & $+0.4$ & 0.72\\
\bottomrule
\end{tabular}}
\begin{minipage}{0.95\textwidth}\footnotesize\emph{Notes:} Gross investment is research-and-development compute (OpenAI \$8.3bn, Anthropic \$4.1bn in 2025), growing at $g$ (0.73 and 1.01 log points from 2024). A share $\phi$ builds versions, depreciating at $\delta_v=3.2$; the rest builds capability capital, depreciating at $\delta_K$. On a steady path the depreciation flow is $I[\phi\delta_v/(\delta_v+g)+(1-\phi)\delta_K/(\delta_K+g)]$. Surplus is revenue less inference compute less depreciation, before staff and other costs; $^\dagger$ counts OpenAI's inference for paying users only (\$4.5bn of \$8.0bn). Revenue is the calendar-year integral of reported annualized revenue. Source: \citet{EpochCompanies2026}, compiled from company statements and press reports.\end{minipage}
\end{table}

\subsection{An alternative capability measure}\label{app:eci}

The Epoch Capabilities Index \citep{EpochECI2026} is an item-response aggregate of a different set of benchmarks from the one behind the paper's capability measure. It matches 118 panel versions, with a correlation of 0.95 (rank correlation 0.97) with the paper's measure. Table~\ref{tab:eci} re-estimates the main demand specification on the 93 matched versions with each measure, both standardized on the matched versions, and computes the Red Queen rate $\theta(1-s)g_A$ with $g_A$ the growth of the frontier in each measure's units. The two measures give the same rate within three percent. Both differ from the full-sample rate of 2.78 because the matched versions have a different capability coefficient and frontier growth, not because of the measure.

\begin{table}[htbp]\centering\small
\caption{The Red Queen rate under two capability measures (matched sample)}\label{tab:eci}
\begin{tabular}{lcc}
\toprule
 & Paper's index & Epoch index (ECI)\\
\midrule
$\theta$ per SD of matched versions & 1.85 (0.48) & 2.05 (0.40)\\
$\lambda$ & 0.36 (0.29) & 0.68 (0.29)\\
First-stage $F$ & 65 & 39\\
Frontier growth, SD per year & 1.24 & 1.16\\
Red Queen rate $\theta(1-s_i)g_A$ (log points/yr) & 2.27 & 2.34\\
\bottomrule
\end{tabular}
\begin{minipage}{0.9\textwidth}\footnotesize\emph{Notes:} Archived weeks, developer and week fixed effects, DeepSeek-window instrument; 93 versions and 4,243 version-weeks; standard errors clustered by version. Frontier growth is the slope of the running maximum over released versions, monthly from March 2024 to October 2026. $s_i=0.016$, the mean share of frontier versions in their first year, as in Section~\ref{sec:rq}.\end{minipage}
\end{table}

\subsection{Responses to list-price changes}\label{app:priceevents}

For horizon $k$ the local-projection regression \citep{Jorda2005} is
\[
\ln s_{i,t+k}-\ln s_{i,t-1}=b_k\,\Delta\ln p_{it}+c_k(\ln s_{i,t-1}-\ln s_{i,t-5})+\text{age-band effects}+\xi_t+e_{it},
\]
estimated on version-weeks with a price change and no other change in $(t,t+k]$, against versions with no price change in $[t-4,t+k]$; the implied coefficient is $-b_k/(1-\bar s)$. Figure~\ref{fig:priceevents} reports it for all changes of at least five percent and for portfolio repricing; the closed-weight estimates are in its notes. A placebo regression of the preceding month's share change on the coming price change gives $-0.53$ (0.27) for closed-weight versions, so shares were not flat before the cuts, and the closed-weight responses are indicative.

\begin{figure}[htbp]
\centering
\includegraphics[width=0.6\textwidth]{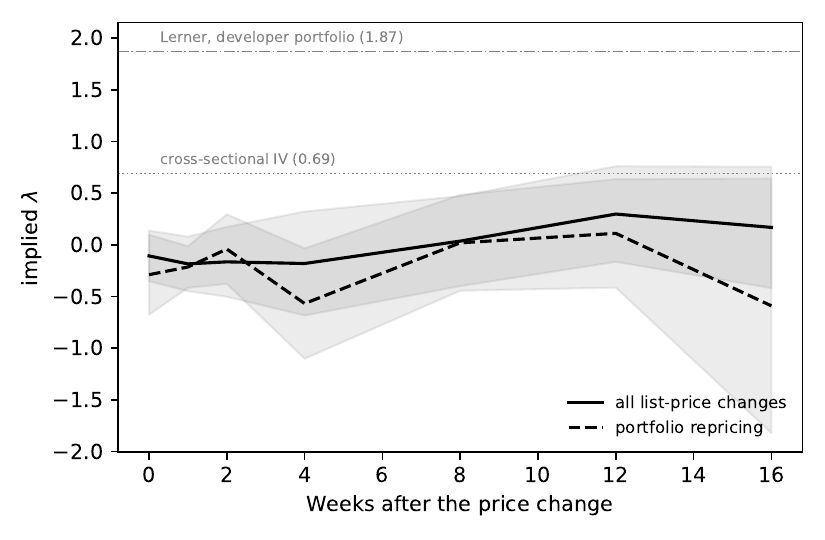}
\caption{Implied price coefficient by horizon after list-price changes}\label{fig:priceevents}
\begin{minipage}{0.9\textwidth}\footnotesize\emph{Notes:} Local projections of the log token share on the version's own list-price change, with week fixed effects, age bands and the preceding month's share trend; shaded bands are 95 percent intervals clustered by version. Solid: all 323 changes of at least five percent; dashed: portfolio repricing by a developer changing three or more prices in the same week. Reference lines: the cross-sectional instrumental-variable estimate and the coefficient implied by the Lerner condition. Closed-weight versions only (23 changes of at least ten percent): 0.06 (0.16) at one week, 0.63 (0.44) at eight, 0.89 (0.39) at twelve.\end{minipage}
\end{figure}

\subsection{Validation of the reconstructed panel}\label{app:validation}

The weekly model-level token series is reconstructed from Internet Archive captures of the marketplace's model pages and ranking tables and, for the most recent weeks, from its daily API series. Since May 2025 the marketplace's official data API reports weekly tokens by model within each of twelve usage categories, the 50 largest models per category. Summing a model's category tokens gives a lower bound on its total, and the comparison is a test of the reconstruction. Over 6,826 model-weeks and 430 models, the correlation of log tokens between the two sources is 0.91, and the median classified share of a model's reconstructed tokens is 0.79. For the ten largest models in each week the correlation of log levels is 0.96 and that of weekly log changes 0.73; for the 25 largest, the correlation of eight-week changes is 0.81. Weekly changes of small models agree less (0.59 for the 50 largest), as expected when the API's category lists truncate them.

Two further checks concern the linking of marketplace identifiers to versions. Matches made automatically with a name-similarity score below 0.8 carry 11 percent of matched tokens. Dropping every version with such a match changes the within-version decay of frontier versions' log revenue share from $-0.070$ (0.005; 105 versions) to $-0.076$ (0.006; 76 versions), and dropping the 66 weeks reconstructed only from top-$N$ ranking tables gives $-0.062$ (0.016).

\subsection{Developers' capability leads}\label{app:capshocks}

Figure~\ref{fig:capshocks} reports the local projections behind the depreciation rate of developers' capability capital in Section~\ref{sec:whichrate}. The developer-month panel covers 22 developers with at least three versions, March 2024 to October 2026. A jump is a month in which the developer's own release raises its position, its best released capability minus the mean of the five leading developers', by at least 0.15 standard deviations (47 events). The regressions include month effects and the change in the outcome over the preceding three months, with standard errors clustered by developer; the decay rate is the slope of the log response on the horizon after its peak, and its interval is from a bootstrap over developers. The lead falls from 0.51 standard deviations in the month of the jump to 0.33 after six months and 0.20 after twelve, a rate of 0.74 log points a year (90 percent interval 0.17 to 1.42); with a threshold of 0.25 (36 events) the rate is 0.47 ($-0.08$ to 1.09), and with the jump measured continuously 0.67. The developer's revenue response shows no measurable decay over the year (Figure~\ref{fig:capshocks}B), but its interval is wide. Unconditionally, developers' positions are persistent, with a six-month coefficient of 1.05 (0.07), and their log revenue shares relative to the month's mean revert at 0.6 to 1.0 log points a year at horizons of three to twelve months.

\begin{figure}[htbp]
\centering
\includegraphics[width=0.95\textwidth]{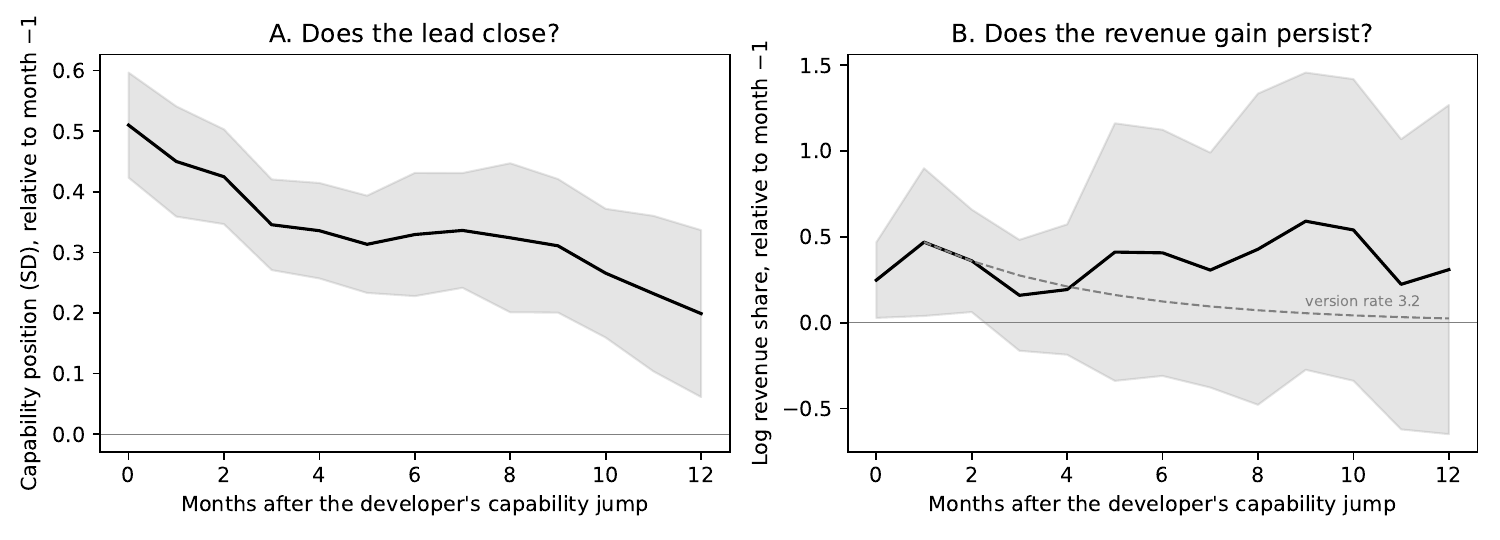}
\caption{Developers' capability leads and their revenue}\label{fig:capshocks}
\begin{minipage}{0.92\textwidth}\footnotesize\emph{Notes:} Responses to a jump of at least 0.15 SD in a developer's capability position caused by its own release, by months after the jump, relative to the month before; 95 percent intervals clustered by developer. (A) Capability position. (B) Log revenue share of the developer; the dashed line is the path a single version's share would follow at the version depreciation rate of 3.2 log points a year.\end{minipage}
\end{figure}

\end{document}